\documentclass{article}
\usepackage{amsthm,amsmath,amssymb}
\usepackage[margin=1in]{geometry} 
\usepackage{caption}
\usepackage{subcaption}
\usepackage{graphicx}
\usepackage{epstopdf}
\usepackage{color}
\usepackage{framed}
\usepackage{lscape}
\usepackage{rotating}
\usepackage{algorithm}
\usepackage{grffile}
\usepackage{multirow}
\usepackage{xcolor}
\usepackage{comment}
\usepackage{xr}
\usepackage{url}

\usepackage[colorlinks=true,linkcolor=black,citecolor=black,urlcolor=black]{hyperref}

\usepackage{algorithmicx}
\usepackage[noend]{algpseudocode}

\newtheorem{thm}{Theorem}[section]

\newtheorem{lem}[thm]{Lemma}

\newtheorem{cor}[thm]{Corollary}
\newtheorem{prop}[thm]{Proposition}

\theoremstyle{definition}
\newtheorem{rmk}{Remark}

\newcommand{\out}{\mathrm{o}}

\newcommand{\op}{\mathrm{op}}
\newcommand{\Fr}{\mathrm{F}}
\newcommand{\wh}{\mathrm{w}}

\newcommand{\ep}{\varepsilon}
\newcommand{\E}{\mathbb{E}}

\newcommand{\Var}{\text{Var}}
\newcommand{\tr}[1]{\mathrm{tr}(#1)}

\newcommand{\snr}{\mathrm{SNR}}

\newcommand{\R}{\mathbb{R}}
\newcommand{\A}{\mathcal{A}}
\newcommand{\AMSE}{\mathrm{AMSE}}

\newcommand{\opt}{\mathrm{opt}}
\newcommand{\lin}{\mathrm{lin}}

\newcommand{\wtilde}{\widetilde}
\newcommand{\diag}{\text{diag}}

\renewcommand{\P}{\mathbb{P}}
\renewcommand{\sp}{\text{span}}
\renewcommand{\L}{\mathcal{L}}
\newcommand{\W}{\mathcal{W}}
\newcommand{\U}{\mathcal{U}}

\newcommand{\argmin}{\operatorname*{arg\,min}}

\begin{document}

\title{Optimal spectral shrinkage and PCA with heteroscedastic noise}
\author{William Leeb\footnote{School of Mathematics, University of Minnesota, Twin Cities, Minneapolis, MN, USA.}\,\,  and Elad Romanov\footnote{School of Computer Science and Engineering, The Hebrew University, Jerusalem, Israel.}}
\date{}
\maketitle
\abstract{
This paper studies the related problems of prediction, covariance estimation, and principal component analysis for the spiked covariance model with heteroscedastic noise. We consider an estimator of the principal components based on whitening the noise, and we derive optimal singular value and eigenvalue shrinkers for use with these estimated principal components. Underlying these methods are new asymptotic results for the high-dimensional spiked model with heteroscedastic noise, and consistent estimators for the relevant population parameters. We extend previous analysis on out-of-sample prediction to the setting of predictors with whitening. We demonstrate certain advantages of noise whitening. Specifically, we show that in a certain asymptotic regime, optimal singular value shrinkage with whitening converges to the best linear predictor, whereas without whitening it converges to a suboptimal linear predictor. We prove that for generic signals, whitening improves estimation of the principal components, and increases a natural signal-to-noise ratio of the observations. We also show that for rank one signals, our estimated principal components achieve the asymptotic minimax rate.
}

%

\section{Introduction}

Singular value shrinkage and eigenvalue shrinkage are popular methods for denoising data matrices and covariance matrices. Singular value shrinkage is performed by computing a singular value decomposition of the observed matrix $Y$, adjusting the singular values, and reconstructing. The idea is that when $Y = X+N$, where $X$ is a low-rank signal matrix we wish to estimate, the additive noise term $N$ inflates the singular values of $X$; by shrinking them we can move the estimated matrix closer to $X$, even if the singular vectors remain inaccurate. Similarly, eigenvalue shrinkage for covariance estimation starts with the sample covariance of the data, and shrinks its eigenvalues. There has been significant recent activity on deriving optimal shrinkage methods \cite{shabalin2013reconstruction, gavish-donoho-2017, nadakuditi2014optshrink, donoho2014minimax, gavish2014optimal, donoho2018optimal, donoho2018condition}, and applying them to various scientific problems \cite{bhamre2016denoising, anden2018structural, moore2014improved, cordero2019complex}.

A standard setting for analyzing the performance of these methods is the \emph{spiked covariance model} \cite{johnstone2001distribution, baik2006eigenvalues, paul2007asymptotics, bai2008central, donoho2018optimal}. Here, the observation matrix is composed of iid columns $Y_j$ in $\R^p$, $j=1,\dots,n$ from some distribution consisting of signal vectors $X_j$ lying on a low-dimensional subspace, plus independent noise vectors $\ep_j$ with some covariance matrix $\Sigma_\ep$. The theory for prediction of $X_1,\dots,X_n$ in the spiked model with orthogonally invariant noise, i.e., when $\Sigma_\ep = \nu I_p$, is very well-developed \cite{donoho2014minimax, shabalin2013reconstruction, gavish-donoho-2017, leeb2020operator}. Singular value shrinkage is known to be minimax optimal, and asymptotically optimal shrinkers have been derived for a wide variety of loss functions.

Many applications in signal processing, imaging, and related fields involve noise that is \emph{heteroscedastic} \cite{nadakuditi2010sample, lund2006non, bendory2020single, bhamre2016denoising, krim1996two, anden2017factor, anden2018structural}.
%
%
%
This paper studies the effect of \emph{whitening} the noise; that is, working in rescaled coordinates, in which the noise is white. We first estimate the noise covariance matrix $\Sigma_\ep$. We then normalize, or \emph{whiten}, the observations $Y_j$ by applying $\Sigma_\ep^{-1/2}$; the resulting vectors $Y_j^{\wh}$ consist of a transformed signal component $X_j^{\wh} = \Sigma_{\ep}^{-1/2} X_j$, plus \emph{isotropic} noise $G_j = \Sigma_{\ep}^{-1/2} \ep_j$. Singular value shrinkage is then performed on this new, whitened observation matrix, after which the inverse transformation $\Sigma_{\ep}^{1/2}$ is applied. Similarly, we perform eigenvalue shrinkage to the sample covariance of the whitened data, and then apply the inverse transformation.

While this approach is restricted to cases when $\Sigma_{\ep}$ can be consistently estimated, when it does apply it has a number of advantages over competing methods. First, in the classical ``large $n$'' asymptotic limit, our method of singular value prediction with whitening, while non-linear in the observed data, converges to the best linear predictor of the data, an oracle method that requires knowledge of the population principal components. By contrast, singular value shrinkage without whitening (as in \cite{nadakuditi2014optshrink}) converges to a suboptimal linear filter. Further, we show that under certain modelling assumptions, whitening improves the estimation of the population singular vectors, and achieves the same rate of subspace estimation as the minimax optimal method derived in \cite{zhang2018heteroskedastic}. Next, because we compute the SVD of a matrix with isotropic noise, our method requires weaker assumptions on the principal components of the signal vectors than those in \cite{nadakuditi2014optshrink}.



As the key step in our procedures is performing spectral shrinkage to the whitened data or covariance matrices, the question arises: what are the optimal singular values/eigenvalues? While whitening has been used with shrinkage in previous works (e.g.\ in \cite{liu2016epca, dobriban2017optimal, bhamre2016denoising}) it appears that the question of optimal shrinkage has not been fully addressed. This paper derives the precise choice of optimal singular values and eigenvalues, and shows, using new asymptotic results, how to consistently estimate them from the observed data.

\subsection{Overview of results}

\subsubsection{Spectral shrinkage with noise whitening}

We introduce a new method for predicting $X$ from $Y$ when the noise matrix $N$ is heteroscedastic.
%
%
We first perform a linear transformation to the observations to \emph{whiten} the noise. The resulting vectors are still of the form ``low rank plus noise'', but the noise term has been transformed into an isotropic Gaussian, while the low-rank signal component has been rescaled along the principal components of the noise covariance.

Next, we shrink the singular values of the transformed matrix. Intuitively, this step removes the effect of the noise from the spectrum of the observed matrix. Finally, we arrive at a predictor of the signal matrix $X$ by applying the inverse change of variables, i.e., we \emph{unwhiten}.

This three-step procedure --- whiten, shrink, unwhiten --- depends on the choice of singular values used in the middle shrinkage step. As it turns out, there are precise,  optimal, and consistently estimable formulas for the optimal singular values. These are derived in Section \ref{sec-shrinker}, and the resulting method summarized in Algorithm \ref{alg:homshrink}.

For covariance estimation, we introduce an analogous procedure in which eigenvalue shrinkage is applied to the sample covariance of the whitened observations. After shrinkage, we then apply the  inverse whitening transformation. As with singular value shrinkage, this three-step procedure of whitening, shrinking the eigenvalues, and unwhitening depends crucially on the choice of eigenvalues for the middle step. In Section \ref{sec-eigs}, we will explain the method in detail, including the derivation of consistent estimators for the optimal eigenvalues for a variety of loss functions. The method is summarized in Algorithm \ref{alg:cov-est}.



%

\subsubsection{Singular value shrinkage and linear prediction}


In Section \ref{sec-linpred}, we show that in the classical regime (when $p \ll n$), singular value shrinkage with whitening converges to the optimal linear predictor of the data, while shrinkage without whitening will converge to a different, typically suboptimal, linear filter. In this sense, not only is shrinkage with whitening preferable to no whitening, but the whitening transform is an asymptotically optimal change of coordinates to apply to the data before shrinking in the classical setting.

In Section \ref{sec-oos}, we also derive the optimal coefficients for the out-of-sample prediction problem, described in \cite{dobriban2017optimal}. In this problem, the PCs estimated from a set of in-sample data $Y_1,\dots,Y_n$ are used to denoise an independently drawn out-of-sample observation. We show that the AMSE for singular value shrinkage with whitening is identical to the asymptotic expected loss achieve by out-of-sample denoising, which extends the analogous result from \cite{dobriban2017optimal}. The out-of-sample predictor is summarized in Algorithm \ref{alg:oos}.


%

\subsubsection{Subspace estimation and PCA}

The eigenspace of the estimated covariance $\hat \Sigma_x$ (equivalently, the left singular subspace of $\hat X$) is not spanned by the singular vectors of the raw data matrix $Y$. Rather, they are spanned by the vectors $\hat u_k$ obtained by applying the inverse whitening transformation to the top $r$ singular vectors of the whitened observation matrix. 


In Section \ref{sec-pca}, we will show under a generic model for the signal PCs, the estimated PCs $\hat u_1,\dots,\hat u_r$ improve upon estimation of the population PCs $u_1,\dots,u_r$, as compared to the left singular vectors of $Y$. We will show too that when $r=1$, $\hat u_1$ achieves the minimax rate of principal subspace estimation derived in \cite{zhang2018heteroskedastic}. That is, in a certain sense it is an optimal estimator of the signal principal subspace.




%

\subsubsection{Spiked model asymptotics}

The methods and analysis of this paper rely on precise descriptions of the asymptotic behavior of the singular values and singular vectors of the whitened matrix $Y^{\wh}$. While some of the necessary results are already found in the literature \cite{paul2007asymptotics, benaych2012singular}, we have also needed to derive several new results as well, which may be found in Theorems \ref{thm-main1} and \ref{thm-main2} in Section \ref{sec-asymptotics}. Whereas earlier work has characterized the angles between the singular vectors of $X^{\wh}$ and $Y^{\wh}$, we have provided formulas for the cosines of the angles between the singular vectors after the inverse whitening transformation has been performed -- that is, we characterize the change in angles resulting from unwhitening. These parameters are a key ingredient for deriving the optimal spectral shrinkers in Section \ref{sec-hidim}.

\subsection{Related work}

\subsubsection{Singular value shrinkage}
\label{intro-shrinkage}

The prediction method in this paper is a generalization of a standard method for predicting the matrix $X$ from the observed matrix $Y$, known as singular value shrinkage. Briefly, it is performed by leaving fixed the singular vectors of $Y$, while adjusting its singular values, to mitigate the effects of noise on the spectrum. It is shown in \cite{donoho2014minimax} that when the noise matrix $N$ is white Gaussian noise, or in other words $\Sigma_\ep = I_p$, then singular value shrinkage is minimax optimal for predicting $X$ from $Y$.

The  paper \cite{shabalin2013reconstruction} considers optimal singular value shrinkage for Frobenius loss and white noise. In \cite{gavish-donoho-2017}, optimal singular value shrinkers are derived for isotropic noise, for a much broader family of loss functions; the special case of operator norm loss is considered in \cite{leeb2020operator}. The effectiveness of these methods rests on the asymptotic spectral theory of the data matrix $Y$ developed in \cite{paul2007asymptotics, benaych2012singular} among others.

In the paper \cite{nadakuditi2014optshrink}, optimal singular value shrinkage (known as `OptShrink') is derived under much more general conditions on the noise matrix $N$, by exploiting the general asymptotic spectral theory developed in \cite{benaych2012singular} for non-isotropic noise. While OptShrink may be effectively applied when the noise is non-isotropic, it requires the signal principal components to be vectors with iid random entries (or orthonormalized versions thereof).

\subsubsection{Eigenvalue shrinkage}

Covariance estimation is a well-studied problem in statistics and its applications. A standard method for estimating the population covariance $\Sigma_x$ is \emph{eigenvalue shrinkage} \cite{Stein1956, Stein1986, donoho2018optimal, donoho2018condition}. Analogously to singular value shrinkage for predicting $X$, eigenvalue shrinkage leaves fixed the eigenvectors of the sample covariance $\hat \Sigma_y = \sum_{j=1}^{n} Y_j Y_j^\top / n = YY^\top / n$, or equivalently the left singular vectors of $Y$, and replaces the eigenvalues by estimated values to reduce the effect of the noise.

As we will discuss in Section \ref{sec-estimation}, it is often natural to consider different loss functions for measuring the error in covariance estimation \cite{donoho2018condition}. The paper \cite{donoho2018optimal} derives optimal eigenvalue shrinkers for a very large collection of loss functions. Their method is restricted to white noise, i.e., where $\Sigma_\ep$ is a multiple of the identity matrix.

\subsubsection{Heteroscedastic noise}

There have been a number of recent papers on the spiked model with heteroscedastic noise. The paper \cite{zhang2018heteroskedastic} devises an iterative algorithm for estimating the principal subspace of $X_j$ in this setting, and proves that their method achieves the optimal error rate. Our method uses a different estimator for the population PCs, which achieves an error that matches the optimal rate of \cite{zhang2018heteroskedastic} under an additional assumption \eqref{incoherence-condition} (which is vacuous when $r=1$).



The papers \cite{hong2018optimally,hong2016towards,hong2018asymptotic} consider a different but related model, in which each observation $Y_j$ has white noise but with noise strengths varying across the observations. In \cite{hong2018asymptotic}, they show that when the signal energy and noise energy are fixed, subspace estimation is optimal when the noise is white. The proof of our Theorem \ref{thm-angles} builds on this result, by combining it with our analysis of the change in angles between the empirical and population PCs after whitening. The work \cite{hong2018optimally} shows that an alternative choice of weighting is optimal for estimating the signal principal components. The aforementioned paper \cite{nadakuditi2014optshrink} designs optimal singular value shrinkers without whitening for a broad range of noise distributions, which include our noise model as a special case.


When working in the eigenbasis of the noise covariance, the whitening procedure we describe in this work is an example of what is called \emph{weighted PCA}, in which weights are applied to individual variables before the principal components are computed \cite{jolliffe2002principal, jackson2005users}. The inverse standard deviation of the noise is a standard choice of weights \cite{wold1987principal, yue2004weighted, wouters2003graphical}; in that sense, the present work can be seen as providing a theoretical analysis of this already widely-used choice.

\subsubsection{Shrinkage with whitening}

Previous works have proposed pairing the whitening transformation with spectral shrinkage, which we study in this work. The paper \cite{liu2016epca} proposes the use of whitening in conjunction with exponential family noise models for covariance estimation. The paper \cite{dobriban2017optimal} proposes whitening in the context of transformed spiked models for data prediction. The papers \cite{bhamre2016denoising, anden2018structural} use whitening and eigenvalue shrinkage for covariance estimation.

However, previous works on singular value shrinkage with whitening employed suboptimal shrinkers, developed from heuristic considerations. In this paper, we undertake a systematic study of this problem, and rigorously derive the optimal shrinkers, under Frobenius loss (in an asymptotic sense). For covariance estimation, \cite{liu2016epca} derives the optimal eigenvalue shrinker for the special case of operator norm loss, but their method does not apply to more general loss functions.

\subsection{Outline of the paper}

The rest of the paper is organized as follows. Section \ref{sec-prelims} contains a detailed description of the model and assumptions; statements of the prediction and estimation problems to be studied; and a review of known results on the spiked model and spectral shrinkage. Section \ref{sec-asymptotics} provides the asymptotic theory on the spiked model that will be used throughout the rest of the paper. Section \ref{sec-hidim} presents the optimal spectral shrinkers with whitening. Section \ref{sec-linpred} analyzes the behavior of weighted singular value shrinkage schemes in the classical ($p \ll n$) setting, and shows the optimality of whitening in this regime. Section \ref{sec-oos} describes and solves the out-of-sample prediction problem. Section \ref{sec-pca} derives several results on the theoretical benefits of whitening for principal component analysis. Section \ref{sec-numerical} presents the results of numerical experiments illuminating the theoretical analysis and demonstrating the performance of the proposed methods. Finally, Section \ref{sec-conclusion} provides a conclusion and suggestions for future research.

\section{Preliminaries}
\label{sec-prelims}

In this section, we will introduce the details of the spiked model with heteroscedastic noise, describe the problems we focus on in this paper, and review known results on the asymptotic spectral theory of the spiked model, singular value shrinkage, and eigenvalue shrinkage. This will also serve to introduce notation we will use throughout the text.

\subsection{The observation model}
\label{sec-model}

\begin{table}
\centering
\begin{tabular}{| c | c  | c | }  
\hline  
 Symbol &  Description & Reference  \\
\hline  
    $X_j$ & Signal  & \eqref{eq456}  \\
    $\ep_j$ & Heteroscedastic noise & \eqref{eq-noise}  \\
    $Y_j$ & Observed & \eqref{eq-observed}  \\
    $X_j^{\wh}$ &  Whitened signal & \eqref{eq-sigmaxw} \\
    $G_j$ & Whitened noise & \eqref{eq-noise}  \\
    $Y_j^{\wh}$ &  Whitened observation & \eqref{eq-sigmaxw} \\
    $z_k$ & Signal factor values & \eqref{eq456}, \eqref{eq-zs} \\
    $z_k^{\wh}$ & Whitened signal factor values & \eqref{eq-xh}, \eqref{eq-zs} \\
    $u_k$ & PC of $X_j$'s &  \eqref{eq456}  \\
    $u_k^{\wh}$ & PC of $X_j^{\wh}$'s & \eqref{eq-xh} \\
    $\overline u_k$ & $W^{-1} u_k^{\wh} / \| W^{-1} u_k^{\wh} \|$ & \eqref{ubar} \\
    $\hat u_k^{\wh}$ & Left singular vector  of $Y^{\wh}$ & Preceding \eqref{emp-pcs}  \\
    $\hat u_k$ & $W^{-1} \hat u_k^{\wh} / \| W^{-1} \hat u_k^{\wh} \|$ & \eqref{emp-pcs} \\
    $\overline u_k^{\wh}$ & $W u_k / \| W u_k \|$ & \eqref{eq-bar-uw} \\
    $v_k$ & Right singular vector of $X$  & Preceding \eqref{emp-pcs} \\
    $v_k^{\wh}$ & Right singular vector of $X^{\wh}$ & Preceding \eqref{emp-pcs} \\
    $\hat v_k^{\wh}$ & Right singular vector of $Y^{\wh}$ & Preceding \eqref{emp-pcs} \\
\hline
\end{tabular}  
\caption{Vectors used in this paper.}
\label{table:vectors}
\end{table}

We now specify the precise model we will be studying in this paper. We observe iid vectors $Y_1,\dots,Y_n$ in $\R^p$, of the form:
\begin{align}
\label{eq-observed}
Y_j = X_j + \ep_j.
\end{align}
The random \emph{signal} vectors $X_j$ are assumed to be mean zero and to have a rank $r$ covariance matrix $\Sigma_x = \sum_{k=1}^{r} \ell_k u_k u_k^\top$, where the vectors $u_k$ are taken to be orthonormal, and are called the \emph{principal components (PCs)} of the random vectors $X_j$. More precisely, and to distinguish them from estimated vectors we will introduce later, we will call them the \emph{population} PCs. The numbers $\ell_k$, which are the variances of the $X_j$ along $u_k$, are positive; we will specify their ordering later, in equation \eqref{ell_ordering} below.

The random \emph{noise} vectors $\ep_j$ are of the form
\begin{align}
\label{eq-noise}
\ep_j = \Sigma_{\ep}^{1/2} G_j,
\end{align}
where $G_j \in \R^p$ is a mean-zero Gaussian noise vector with covariance $I_p$, and $\Sigma_{\ep}$ is a full-rank positive definite covariance matrix, assumed to be known (though see Remark \ref{rmk-known}). The noise vectors $G_j$ are drawn independently from the $X_j$. 

We can write
\begin{align}
\label{eq456}
X_j = \sum_{k=1}^{r} \ell_k^{1/2} z_{jk} u_k
\end{align}
where $z_{jk}$ are uncorrelated (though not necessarily independent) random variables, with  $\E z_{jk} = 0$ and $\Var (z_{jk}) = 1$. We remark that the assumption that $X_j$ has mean zero is not essential; all the results of this paper will go through almost without modification if we first estimate the mean of $X$ by the sample mean and subtract it from each observation $Y_j$. We also note that in the terminology of factor analysis, the $z_{jk}$ may be called the factor values; for background on factor analysis, see, for instance, \cite{anderson1984, anderson2003introduction, schervish1987review, dobriban2017factor}.


In addition to the original observations $Y_j$, we will also be working with the \emph{whitened} (or \emph{homogenized} \cite{liu2016epca}) observations $Y_j^{\wh}$, defined by $Y_j^{\wh} = W Y_j$, where 
\begin{align}
\label{eq-hdef}
W = \Sigma_\ep^{-1/2}
\end{align}
is the \emph{whitening matrix}. The vectors $Y_j^{\wh}$ can be decomposed into a transformed signal $X_j^{\wh} = W X_j$ plus white noise $G_j$. The whitened vectors $X_j^{\wh}$ have rank $r$ covariance 
\begin{align}
\label{eq-sigmaxw}
\Sigma_{x}^{\wh} = W \Sigma_x W,
\end{align}
and lie in the $r$-dimensional subspace $\sp\{W u_1,\dots W u_r\}$. We will let $u_1^{\wh},\dots,u_r^{\wh}$ be the orthonormal PCs of $X_j^{\wh}$ -- that is, the leading $r$ eigenvectors (up to sign) of $\Sigma_x^{\wh}$ -- and write
\begin{align}
\label{eq-xh}
X_j^{\wh} = \sum_{k=1}^{r} (\ell_k^{\wh})^{1/2} z_{jk}^{\wh} u_k^{\wh},
\end{align}
where again $\E z_{jk}^{\wh} = 0$ and $\Var(z_{jk}^{\wh}) = 1$, the $\ell_k^{\wh}$ are strictly positive, and
\begin{align}
\label{eq:ellw-ordered}
\ell_1^{\wh} > \dots > \ell_r^{\wh} > 0.
\end{align}
In general, there is not a simple relationship between the PCs $u_1,\dots,u_r$ of $X_j$ and the PCs $u_1^{\wh},\dots,u_r^{\wh}$ of $X_j^{\wh}$, or between the eigenvalues $\ell_1,\dots,\ell_r$ and the eigenvalues $\ell_1^{\wh},\dots,\ell_r^{\wh}$.

We introduce some additional notation. We will denote the normalized matrices by $Y = [Y_1, \dots Y_n] / \sqrt{n}$, $Y^{\wh} = [Y_1^{\wh},\dots,Y_n^{\wh}] / \sqrt{n} $, $X = [X_1,\dots,X_n] / \sqrt{n}$, $X^{\wh} = [X_1^{\wh},\dots,X_n^{\wh}] / \sqrt{n}$, $G = [G_1,\dots,G_n] / \sqrt{n}$ and $N = [\ep_1,\dots,\ep_n] / \sqrt{n}$. Note that $Y = X + N$ and $Y^{\wh} = X^{\wh} + G$.

We will denote by $v_1,\dots,v_r$ the right singular vectors of the matrix $X$, and denote by $v_1^{\wh},\dots,v_r^{\wh}$ the right singular vectors of the matrix $X^{\wh}$. We denote by $\hat u_1^{\wh},\dots,\hat u_r^{\wh}$ and $\hat v_1^{\wh},\dots \hat v_r^{\wh}$ the top $r$ left and right singular vectors of the matrix $Y^{\wh}$. We define, for $1 \le k \le r$, the empirical vectors:
\begin{align}
\label{emp-pcs}
\hat u_k = \frac{W^{-1} \hat u_k^{\wh}}{\|W^{-1} \hat u_k^{\wh}\|}.
\end{align}
We also define the population counterparts, 
\begin{align}
\label{ubar}
\overline u_k = \frac{W^{-1} u_k^{\wh}}{\|W^{-1} u_k^{\wh}\|}.
\end{align}
Similarly, for $1 \le k \le r$ we define
\begin{align}
\label{eq-bar-uw}
\overline u_k^{\wh} = \frac{W u_k}{\|W u_k\|}.
\end{align}
Note that $\sp\{\overline u_1,\dots, \overline u_r\} = \sp\{u_1,\dots,u_r\}$, and $\sp\{\overline u_1^{\wh},\dots, \overline u_r^{\wh}\} = \sp\{u_1^{\wh},\dots,u_r^{\wh}\}$. However, the vectors $\overline u_1,\dots, \overline u_r$ will not, in general, be pairwise orthogonal; and similarly for $\overline u_1^{\wh},\dots, \overline u_r^{\wh}$.

Finally, we define the factor vectors $z_k$ and $z_k^{\wh}$ by
\begin{align}
\label{eq-zs}
z_k = (z_{1k},\dots,z_{nk})^\top, \quad z_k^{\wh} = (z_{1k}^{\wh},\dots,z_{nk}^{\wh})^\top.
\end{align}

We formally consider a \emph{sequence} of problems, where $n$ and $p=p_n$ both tend to $\infty$ with a limiting aspect ratio, $\gamma$:
\begin{align}
\label{eq-gamma}
\gamma = \lim_{n \to \infty} \frac{p_n}{n},
\end{align}
which is assumed to be finite and positive.
The number of population components $r$ and the variances $\ell_1,\dots,\ell_r$ are assumed to be fixed with $n$. Because $p$ and $n$ are increasing, all quantities that depend on $p$ and $n$ are elements of a sequence, which will be assumed to follow some conditions which we will outline below and summarized in Section \ref{sec-asy-assumptions}.  Though we might denote, for instance, the PC $u_k$ by $u_k^{(p)}$, $X$ by $X^{(p,n)}$, and so forth, to keep the notation to a minimum -- and in keeping with standard practice with the literature on the spiked model -- we will typically drop the explicit dependence on $p$ and $n$.

\begin{rmk}
Because $r$ is fixed as $p$ and $n$ grow, the left singular vectors of the $p$-by-$n$ population matrix $X = [X_1,\dots,X_n] / \sqrt{n}$ are asymptotically consistent estimators (up to sign) of the population PCs $u_1,\dots,u_r$. More precisely, if $\tilde u_1,\dots,\tilde u_r$ are the left singular vectors of $X$, then almost surely
\begin{align}
\label{eq1421}
\lim_{p \to \infty} |\langle u_k , \tilde u_k \rangle | = 1.
\end{align}
Similarly, if $\tilde u_1^{\wh},\dots,\tilde u_r^{\wh}$ are the left singular vectors of $X^{\wh}$, then almost surely
\begin{align}
\label{eq1422}
\lim_{p \to \infty} |\langle u_k^{\wh} , \tilde u_k^{\wh} \rangle | = 1.
\end{align}
The limits \eqref{eq1421} and \eqref{eq1422} may be easily derived from, for example, Corollary 5.50 in \cite{vershynin2010intro} (restated as Lemma \ref{covest123} in Appendix \ref{proof-blp}), since the effective dimension of the $X_j$ is $r$, not $p$. Because this paper is concerned only with first-order phenomena, we will not distinguish between $u_k$ (respectively, $u_k^{\wh}$) and $\tilde u_k$ (respectively, $\tilde u_k^{\wh}$).
\end{rmk}

\begin{rmk}
The unnormalized vectors $W^{-1} u_k^{\wh}$ are the \emph{generalized singular vectors} of the matrix $X$, with respect to the weight matrix $W^2$ \cite{vanloan1976generalizing}. In particular, they are orthonormal with respect to the weighted inner product defined by $W^2$. Similarly, the vectors $W^{-1} \hat u_k^{\wh}$ are generalized singular vectors of $Y$ with respect to $W^2$.
\end{rmk}



We assume that the values $\|W^{-1} u_k^{\wh}\|$, $1 \le k \le r$, have well-defined limits as $p \to \infty$, and we define the parameters $\tau_k$, $1 \le k \le r$, by
\begin{align}
\label{tau-def}
\tau_k = \lim_{p \to \infty} \|W^{-1} u_k^{\wh}\|^{-2}.
\end{align}
Note that the $\tau_k$ are \emph{not} known a priori; we will show, however, how they may be consistently estimated from the observed data.

With the $\tau_k$'s defined, we now specify the ordering of the principal components of $X_j$ that will be used throughout:
\begin{align}
\label{ell_ordering}
\ell_1 \tau_1 > \dots > \ell_r \tau_r > 0.
\end{align}

We will also assume that the spectrum of $\Sigma_\ep$ stays bounded between $a_{\min} > 0$ and $a_{\max} < \infty$. In order to have well-defined asymptotics in the large $p$, large $n$ regime, we will assume that the normalized trace of $\Sigma_{\ep}$ has a well-defined limit, which we will denote by $\mu_\ep$:
\begin{align}
\label{mu_def}
\mu_{\ep} &= \lim_{p \to \infty} \frac{\tr{\Sigma_\ep}}{p} \in (0,\infty).
\end{align}
For the convenience of the reader, Tables \ref{table:vectors} and \ref{table:scalars} summarize the notation for vectors and scalar parameters that will be used throughout this paper.

\begin{table}
\centering
\begin{tabular}{| c | c  | c | }  
\hline  
 Symbol &  Description & Reference  \\
\hline  
    $\ell_k$ & Signal variances & \eqref{eq456}, \eqref{ell_ordering}  \\
    $\ell_k^{\wh}$ & Whitened signal variances & \eqref{eq-xh}, \eqref{eq:ellw-ordered}  \\
    $\gamma$ & Aspect ratio & \eqref{eq-gamma}  \\
    $\tau_k$ & $\lim_{p \to \infty} \|W^{-1} u_k^{\wh}\|^{-2}$  & \eqref{tau-def}  \\
    $\overline \ell_k$ & $\ell_k^{\wh} / \tau_k$ & \eqref{eq:bar-ell}  \\
    $\mu_\ep$ & Normalized trace of $\Sigma_\ep$  & \eqref{mu_def}  \\
    $\sigma_k^{\wh}$ & Singular value of $Y^{\wh}$  & \eqref{eq-sigmaw}  \\
    $c_k^{\wh}$ & Cosine between $u_k^{\wh}$ and $\hat u_k^{\wh}$  & \eqref{cos_out}  \\
    $\tilde c_k^{\wh}$ & Cosine between $v_k^{\wh}$ and $\hat v_k^{\wh}$  & \eqref{cos_inn}  \\
    $c_k$ & Cosine between $u_k$ and $\hat u_k$ under \eqref{incoherence-condition}  
            & \eqref{prod-unwhite}  \\
\hline
\end{tabular}  
\caption{Scalar parameters used in this paper.}
\label{table:scalars}
\end{table}

\begin{rmk}
\label{rmk-known}
We will assume for most of the paper that the noise covariance $\Sigma_\ep$ is known a priori (though see Section \ref{sec-noise}). However, all of the theoretical results, and resulting algorithms, go through unchanged if the true $\Sigma_\ep$ is replaced by any estimator $\hat \Sigma_\ep$ that is consistent in operator norm, i.e.,
\begin{align}
\lim_{p \to \infty} \|\Sigma_\ep - \hat \Sigma_\ep\|_{\op} = 0.
\end{align}
Examples of such estimators $\hat \Sigma_\ep$ are discussed in Section \ref{sec-noise}.


\end{rmk}

\subsubsection{The asymptotic assumptions}
\label{sec-asy-assumptions}

We enumerate the assumptions we have made on the asymptotic model:
\begin{enumerate}

\item
\label{assumption-gamma}
$p,n\to\infty$ and the aspect ratio $p/n$ converges to $\gamma > 0$.

\item
The eigenvalues of $\Sigma_\ep$ lie between $a_{\min} > 0$ and $a_{\max} < \infty$.

\item
The limit $\lim_{p \to \infty} \tr{\Sigma_\ep} / p$ is well-defined, finite, and non-zero.

\item
\label{assumption-tau}
The limits $\lim_{p \to \infty} \|W^{-1} u_k^{\wh}\|$ are well-defined, finite, and non-zero.

\end{enumerate}

Assumptions \ref{assumption-gamma}--\ref{assumption-tau} will be in effect throughout the entire paper. In addition, some of the results, namely Theorems \ref{thm-main2} and \ref{thm-sin-theta}, will require an additional assumption, which we refer to as \emph{weighted orthogonality} of the PCs $u_1,\dots,u_r$:
\begin{enumerate}\setcounter{enumi}{4}

\item
\label{assumption-incoherence}
For $j \ne k$, the vectors $u_j$ and $u_k$ are asymptotically orthogonal with respect to the $W^2 = \Sigma_\ep^{-1}$ inner product:
\begin{align}
\label{incoherence-condition}
\lim_{p \to \infty} u_j^\top W^2 u_k = 0.
\end{align}

\end{enumerate}

The assumptions \ref{assumption-gamma}--\ref{assumption-tau} listed above are conceptually very benign. In applications, the practitioner will be faced with a finite $p$ and $n$, for which all the listed quantities exist and are finite. The asymptotic assumptions \ref{assumption-gamma}--\ref{assumption-tau} allow us to precisely quantify the behavior when $p$ and $n$ are large. By contrast, assumption \ref{assumption-incoherence} is stronger than assumptions \ref{assumption-gamma}--\ref{assumption-tau}, in that it posits not only that certain limits exist, but also their precise values (namely, $0$). Note that assumption \ref{assumption-incoherence} is trivially satisfied when $r=1$.

\subsubsection{Weighted orthogonality and random PCs}
\label{sec-randompcs}
At first glance, the weighted orthogonality condition \eqref{assumption-incoherence}, which will be used in Theorems \ref{thm-main2} and \ref{thm-sin-theta}, may seem quite strong. However, it is a considerably weaker assumption than what is often assumed by methods on the spiked model. For instance, the method of OptShrink in \cite{nadakuditi2014optshrink} assumes that the PCs $u_1,\dots,u_r$ be themselves random vectors with iid entries (or orthonormalized versions thereof). Under this model, the inner products $u_j^\top W^2 u_k$ almost surely converge to $0$; see Proposition 6.2 in \cite{benaych2011fluctuations}.

In fact, we may introduce a more general random model for random PCs, under which assumption \ref{assumption-incoherence} will hold. For each $1 \le k \le r$, we assume there is a $p$-by-$p$ symmetric matrix $B_k$ with bounded operator norm ($\|B_k\|_{\op} \le C < \infty$, where $C$ does not depend on $p$), and $\tr{B_k} / p = 1$. We then take $u_1,\dots,u_r$ to be the output of Gram-Schmidt performed on the vectors $B_k w_k$, where the $w_k$ are vectors with iid subgaussian entries with variance $1/p$. Then $u_j^\top W^2 u_k = w_j^\top B_j^\top W^2 B_k w_k$, which converges to zero almost surely, again using \cite{benaych2011fluctuations} and the bounded operator norm of $B_j W^2 B_k$.

\begin{rmk}
Under the random model just described the parameters $\tau_k$ are well-defined and equal to $\lim_{p\to\infty} \tr{B_k^\top W^2 B_k}/p$, so long as this limit exists. Indeed, it follows  from \eqref{incoherence-condition}  that $u_k^{\wh}$ is asymptotically identical to $W u_k / \|W u_k\|$ (see Theorem \ref{thm-main2}), and so $\lim_{p \to \infty}\|W^{-1} u_k^{\wh}\|^{-2} = \lim_{p \to \infty} \|W u_k\|^2 = \lim_{p \to \infty} \tr{B_k^\top W^2 B_k}/p$, where we have once again invoked \cite{benaych2011fluctuations}.
\end{rmk}

%

%

\subsection{The prediction and estimation problems}
\label{sec-estimation}

This paper considers three central tasks: denoising the observations $Y_j$ to recover $X_j$ -- what we refer to as \emph{prediction}, since the $X_j$'s are themselves random -- estimating the population covariance $\Sigma_x$, and estimating the principal subspace $\sp\{u_1,\dots,u_r\}$.

For predicting the signal vectors $X_j$, or equivalently the normalized signal matrix $X = [X_1, \dots ,X_n] / \sqrt{n}$, we will use the asymptotic mean squared error to measure the accuracy of a predictor $\hat{X}$:
\begin{align}
\AMSE = \lim_{n \to \infty} \E \|\hat{X} - X\|_{\Fr}^2 
    = \lim_{n\to\infty} \frac{1}{n} \sum_{j=1}^{n} \E \| \hat X_j - X_j \|^2.
\end{align}

For covariance estimation, our goal is to estimate the covariance of the signal vectors, $\Sigma_x = \E[X_j X_j^\top]$ (under the convention that the $X_j$ are mean zero; otherwise, we subtract off the mean). While the Frobenius loss, or MSE, is natural for signal estimation, for covariance estimation it is useful to consider a wider range of loss functions depending on the statistical problem at hand; see \cite{donoho2018condition} and the references within for an elucidation of this point.

We will denote our covariance estimator as $\hat \Sigma_x$. Denote the loss function by $\L(\hat \Sigma_x, \Sigma_x)$; for instance, Frobenius loss
\begin{math}
\L(\hat \Sigma_x, \Sigma_x) = \|\hat \Sigma_x- \Sigma_x\|_{\Fr}^2,
\end{math}
or operator norm loss
\begin{math}
\L(\hat \Sigma_x, \Sigma_x) = \|\hat \Sigma_x- \Sigma_x\|_{\op}.
\end{math}
For a specified loss function $\L$, we seek to minimize the asymptotic values of these loss functions for our estimator,
\begin{align}
\lim_{n \to \infty} \E \L(\hat \Sigma_x, \Sigma_x).
\end{align}

For both the data prediction and covariance estimation problems, it will be a consequence of our analysis that the limits of the errors are, in fact, well-defined quantities.

Finally, we are also concerned with principal component analysis (PCA), or estimating the principal subspace $\U = \sp\{u_1,\dots,u_r\}$, in which the signal vectors $X_j$ lie. We measure the discrepancy between the estimated subspace $\hat \U$ and the true subspace $\U$ by the angle $\Theta(\U,\hat \U)$ between these subspaces, defined by
\begin{align}
\sin \Theta(\U, \hat \U) = \| \hat U_\perp^\top U  \|_{\op},
\end{align}
where $\hat U_\perp$ and $U$ are matrices whose columns are orthonormal bases of $\hat \U^\perp$ and $\U$, respectively.

\subsection{Review of the spiked model}

\subsubsection{Asymptotic spectral theory of the spiked model}
\label{sec-asymptotics}

The spectral theory of the observed matrix $Y$ has been thoroughly studied in the large $p$, large $n$ regime, when $p=p_n$ grows with $n$. We will offer a brief survey of the relevant results from the literature \cite{paul2007asymptotics,benaych2012singular,dobriban2017optimal}.

In the case of isotropic Gaussian noise (that is, when $\Sigma_\ep = I_p$), the  $r$ largest singular values of the matrix $Y$ converge to $\sigma_k$, defined by:
\begin{align}
\label{eq:sigma}
\sigma_{k}^2 = 
\begin{cases}
(\ell_k + 1) (1 + \gamma/\ell_k ), 
    &\text{ if } \ell_k > \sqrt{\gamma}, \\
(1 + \sqrt{\gamma})^2,
    &\text{ if } \ell_k \le \sqrt{\gamma}
\end{cases}.
\end{align}

Furthermore, the top singular vectors $\hat u_k^y$ and $\hat v_k^y$ of $Y$ make asymptotically deterministic angles with the singular vectors $u_k$ and $v_k$ of $X$. More precisely,  the  absolute cosines $|\langle \hat u_j^y , u_k \rangle|$ converge to $c_k = c_k(\gamma,\ell_k)$, defined by
\begin{align}
\label{eq:ck}
c_k^2 =
\begin{cases}
\frac{1 - \gamma/\ell^2 }{1 + \gamma / \ell}
    &\text{ if } j = k \text{ and } \ell_k > \sqrt{\gamma} \\
0   &\text{ otherwise}
\end{cases},
\end{align}
and the  absolute cosines $|\langle \hat v_j^y , v_k \rangle|$ converge to $\tilde c_k = \tilde c_k(\gamma,\ell_k)$, defined by
\begin{align}
\label{eq:tildeck}
\tilde c_k^2 =
\begin{cases}
\frac{1 - \gamma/\ell^2 }{1 + 1 / \ell}
    &\text{ if } j = k \text{ and } \ell_k > \sqrt{\gamma} \\
0   &\text{ otherwise}
\end{cases}.
\end{align}

When $\ell_k > \sqrt{\gamma}$, the population variance $\ell_k$ can be estimated consistently from the observed singular value $\sigma_k$. Since $c_k$ and $\tilde c_k$ are functions of $\ell_k$ and the aspect ratio $\gamma$, these quantities can then also be consistently estimated.


\begin{rmk}
Due to the orthogonal invariance of the noise matrix $N=G$ when $\Sigma_\ep = I_p$, formulas \eqref{eq:sigma}, \eqref{eq:ck} and \eqref{eq:tildeck} are valid for any rank $r$ matrix $X$, so long as $X$'s singular values do not change with $p$ and $n$. The paper \cite{benaych2012singular} derive the asymptotics for more general noise matrices $N$, but with the additional assumption that the singular vectors of $X$ are themselves random (see the discussion in Section \ref{sec-randompcs}). The formulas for the asymptotic singular values and cosines found in \cite{benaych2012singular} are in terms of the Stieltjes transform \cite{bai2009spectral} of the asymptotic distribution of singular values of $Y$, which can be estimated consistently using the observed singular values of $Y$.
\end{rmk}


%

\subsubsection{Optimal shrinkage with Frobenius loss and white noise}
\label{sec-froshr}

We review the theory of shrinkage with respect to Frobenius loss; we briefly mention that the paper  \cite{gavish-donoho-2017} extends these ideas to a much wider range of loss functions for the spiked model.

We suppose that our predictor of $X$ is a rank $r$ matrix of the form
\begin{align}
\hat X = \sum_{k=1}^{r} t_k \hat u_k \hat v_k^\top,
\end{align}
where $\hat u_k$ and $\hat v_k$ are estimated vectors. We will assume that the vectors $\hat v_k$ are orthogonal, and that their cosines with the population vectors $v_k$ of $X$ are asymptotically deterministic. More precisely, we assume that $\langle v_j , \hat v_k\rangle^2 \to \tilde c_k^2$ when $j=k$, and converges to $0$ when $j\ne k$. Similarly, we will assume that $\langle u_k , \hat u_k\rangle^2 \to c_k^2$; however, we do not need to assume any orthogonality condition on the $u_j$'s and $\hat u_j$'s for the purposes of this derivation.

Expanding the squared Frobenius loss between $\hat X$ and $X$ and using the orthogonality conditions on the $v_j$'s and $\hat v_k$'s, we get:
\begin{align}
\|\hat X - X\|_{\Fr}^2 
&= \left\| \sum_{k=1}^{r} \left( 
    t_k \hat u_k \hat v_k^\top  - \ell_k^{1/2} u_k v_k^\top
        \right)\right\|_{\Fr}^2
    \nonumber \\
&= \sum_{k=1}^{r} \left\| 
    t_k \hat u_k \hat v_k^\top  - \ell_k^{1/2} u_k v_k^\top \right\|_{\Fr}^2
    + \sum_{j \ne k} \left\langle 
        t_j \hat u_j \hat v_j^\top  - \ell_j^{1/2} u_j v_j^\top, 
            t_k \hat u_k \hat v_k^\top  - \ell_k^{1/2} u_k v_k^\top\right\rangle_{\Fr}
    \nonumber \\
&\sim \sum_{k=1}^{r} \|t_k \hat u_k \hat v_k^\top  - \ell_k^{1/2} u_k v_k^\top\|_{\Fr}^2,
\end{align}
where $\sim$ denotes almost sure equality as $p,n\to\infty$.

Since the loss separates over the different components, we may consider each component separately. Using the asymptotic cosines, we have:
\begin{align}
\label{amse000}
\|t_k \hat u_k \hat v_k^\top  - \ell_k^{1/2} u_k v_k^\top\|_{\Fr}^2
    \sim t_k^2 + \ell_k - 2 \ell_k^{1/2} c_k \tilde c_k t_k,
\end{align}
which is minimized by taking
\begin{align}
t_k = \ell_k^{1/2} c_k \tilde c_k.
\end{align}
These values of $t_k$, therefore, are the optimal ones for predicting $X$ in Frobenius loss.

Furthermore, we can also derive an estimable formula for the AMSE. Indeed, plugging in $t_k = \ell_k^{1/2} c_k \tilde c_k$ to \eqref{amse000}, we get:
\begin{align}
\AMSE = \sum_{k=1}^{r} \ell_k^2(1 - c_k^2 \tilde c_k^2).
\end{align}

Note that this derivation of the optimal $t_k$ and the AMSE does not require the vectors $\hat u_k$ and $\hat v_k$ to be the singular vectors of $Y$. Rather, we just require the asymptotic cosines to be well-defined, and the $v_j$'s and $\hat v_j$'s to be orthogonal across different components. Implementing this procedure, however, requires consistent estimates of $\ell_k$, $c_k$ and $\tilde c_k$.

\subsubsection{Eigenvalue shrinkage for covariance estimation}
\label{sec-eig-review}

Similar to the task of predicting the data matrix $X$ is estimating the covariance matrix $\Sigma_x = \E[X_j X_j^\top] = \sum_{k=1}^{r} \ell_k u_k u_k^\top$. The procedure we consider in this setting is known as \emph{eigenvalue shrinkage}. Given orthonormal vectors $\hat u_1,\dots,\hat u_r$ estimating the  PCs $u_1,\dots,u_r$, we consider estimators of the form
\begin{align}
\hat \Sigma_x = \sum_{k=1}^{r} t_k^2 \hat u_k \hat u_k^\top,
\end{align}
where $t_k^2$ are estimated population eigenvalues, which it is our goal to determine.

In \cite{donoho2018optimal}, a large family of loss functions are considered for estimating $\Sigma_x$ in white noise. All these loss functions satisfy two conditions. First, they are \emph{orthogonally-invariant}, meaning that if both the estimated and population PCs are rotated, the loss does not change. Second, they are \emph{block-decomposable}, meaning that if both the estimated and population covariance matrices are in block-diagonal form, the loss can be written as functions of the losses between the individual blocks.

The method of \cite{donoho2018optimal} rests on an observation from linear algebra. If (asymptotically) the $\langle \hat u_k , u_k \rangle = c_k$, and $\hat u_j \perp u_k$ for all $1 \le j \ne k \le r$, then there is an orthonormal basis of $\R^p$ with respect to which both $\Sigma_x$ and any rank $r$ covariance $\hat \Sigma_x$ are simultaneously block-diagonalizable, with $r$ blocks of size $2$-by-$2$. More precisely, there is a $p$-by-$p$ orthogonal matrix $O$ so that:
\begin{align}
O \Sigma_x O^\top = \bigoplus_{k=1}^{r} A_k,
\end{align}
and 
\begin{align}
O \hat \Sigma_x O^\top = \bigoplus_{k=1}^{r} \hat \ell_k B_k,
\end{align}
where 
\begin{align}
A_k = 
\left(
\begin{array}{c c}
\ell_k & 0 \\
 0     & 0 \\
\end{array}
\right),
\end{align}
and
\begin{align}
B_k = 
\left(
\begin{array}{c c}
c_k^2                   & c_k \sqrt{1-c_k^2} \\
 c_k \sqrt{1-c_k^2}     & 1-c_k^2 \\
\end{array}
\right).
\end{align}

If $\L(\hat \Sigma, \Sigma)$ is a loss function that is orthogonally-invariant and block-decomposable, then the loss between $\Sigma_x$ and $\hat \Sigma_x$ decomposes into the losses between each $A_k$ and $B_k$, which depend only on the one parameter $\hat \ell_k$. Consequently, 
\begin{align}
\label{min345}
\hat \ell_k = \argmin_{\ell} \L(A_k, \ell B_k).
\end{align}
The paper \cite{donoho2018optimal} contains solutions for $\hat \ell_k$ for a wide range of loss functions $\L$. For example, with Frobenius loss, the optimal value is $\hat \ell_k = \ell_k c_k^2$, whereas for operator norm loss the optimal value is $\hat \ell_k = \ell_k$. Even when closed form solutions are unavailable, one may perform the mimimization \eqref{min345} numerically.

\section{Asymptotic theory}
\label{sec-asymptotics}

A precise understanding of the asymptotic behavior of the spiked model is crucial for deriving optimal spectral shrinkers, as we have seen in Sections \ref{sec-froshr} and \ref{sec-eig-review}. In this section, we provide expressions for the asymptotic cosines between the empirical PCs and the population PCs, as well as limiting values for other parameters. The formulas from Theorem \ref{thm-main1} below will be employed in Section \ref{sec-shrinker} for optimal singular value shrinkage with whitening; and the formulas from Theorem \ref{thm-main2} below will be employed in Section \ref{sec-eigs} for optimal eigenvalue shrinkage with whitening.

The first result, Theorem \ref{thm-main1}, applies to the standard spiked model with white noise. It gives a characterization of the asymptotic angles of the population PCs and empirical PCs with respect to an inner product $x^\top A y$ given by a symmetric positive-definite matrix $A$. Parts \ref{main1-spike} and \ref{main1-product2} are standard results on the spiked covariance model \cite{paul2007asymptotics,benaych2012singular}; we include them here for easy reference. A special case of part \ref{main1-norm} appears in \cite{liu2016epca}, in a somewhat different form; and part \ref{main1-product} appear to be new.

\begin{thm}
\label{thm-main1}
Suppose $Y_1^{\wh},\dots,Y_n^{\wh}$ are iid vectors in $\R^p$ from the spiked model with white noise, with  $Y_j^{\wh} = X_j^{\wh} + G_j$ where $X_j^{\wh}$ is of the form \eqref{eq-xh} and $G_j \sim N(0,I)$. Let $A = A_p$ be an element of a sequence of symmetric, positive-definite $p$-by-$p$ matrices with bounded operator norm ($\|A_p\|_{\op} \le C < \infty$ for all $p$), whose asymptotic normalized trace is well-defined and finite:
\begin{align}
\mu_a = \lim_{p \to \infty} \frac{1}{p} \tr{A_p} < \infty.
\end{align}
Suppose too that for $1 \le k \le r$, the following quantity $\tau_k^a$ is also well-defined and finite:
\begin{align}
\tau_k^a = \lim_{p \to \infty} \| A_p^{1/2} u_k^{\wh} \|^{-2} < \infty.
\end{align}
Define $c_k^{\wh} > 0$ by:
\begin{align}
\label{cos_out}
(c_k^{\wh})^2 &= 
\begin{cases}
\frac{1 - \gamma/(\ell_k^{\wh})^2 }{1 + \gamma /  \ell_k^{\wh}},
    &\text{ if } j = k \text{ and } \ell_k^{\wh} > \sqrt{\gamma} \\
0,   &\text{ otherwise}
\end{cases},
\end{align}
and let $s_k^{\wh} = \sqrt{1 - (c_k^{\wh})^2}$.
Also define $\tilde c_k^{\wh} > 0$ by:
\begin{align}
\label{cos_inn}
(\tilde c_k^{\wh})^2 &= 
\begin{cases}
\frac{1 - \gamma/(\ell_k^{\wh})^2 }{1 + 1 /  \ell_k^{\wh}},
    &\text{ if } j = k \text{ and } \ell_k^{\wh} > \sqrt{\gamma} \\
0,   &\text{ otherwise}
\end{cases},
\end{align}
and $\tilde s_k^{\wh} = \sqrt{1 - (\tilde c_k^{\wh})^2}$.

Then for any $1 \le j, k \le r$, we have, as $n \to \infty$ and $p / n \to \gamma$:
\begin{enumerate}

\item
\label{main1-spike}

The $k^{th}$ largest singular value of $Y^{\wh}$ converges almost surely to
\begin{align}
\label{eq-sigmaw}
\sigma_k^{\wh} = 
\begin{cases}
\sqrt{(\ell_k^{\wh} + 1)\left( 1 + \frac{\gamma}{\ell_k^{\wh}}\right) },
    &\text{ if } \ell_k^{\wh} > \sqrt{\gamma} \\
1 + \sqrt{\gamma},   &\text{ otherwise}
\end{cases}.
\end{align}

\item
\label{main1-norm}
The $A$-norm of $\hat u_k^{\wh}$ converges almost surely:
\begin{align}
\lim_{p \to \infty} \|A_p^{1/2} \hat u_k^{\wh}\|^2 
    &= \frac{(c_k^{\wh})^2}{\tau_k^a} + (s_k^{\wh})^2 \mu_a.
\end{align}

\item
\label{main1-product}
The $A$-inner product between $u_k^{\wh}$ and $\hat u_k^{\wh}$ converges almost surely:
\begin{align}
\label{prod-unwhite777}
\lim_{p \to \infty} \langle  A_p u_k^{\wh}, \hat u_k^{\wh} \rangle^2
=
\begin{cases}
(c_k^{\wh} / \tau_k^a)^2,
   & \text{ if } \ell_k^{\wh} > \sqrt{\gamma} \\
0,  & \text{otherwise}
\end{cases}.
\end{align}

\item
\label{main1-product2}

The inner product between $v_j^{\wh}$ and $\hat v_k^{\wh}$ converges almost surely:
\begin{align}
\lim_{n \to \infty} \langle  v_j^{\wh}, \hat v_k^{\wh} \rangle^2 = 
\begin{cases}
(\tilde c_k^{\wh})^2, 
    & \text{ if } j = k \text{ and } \ell_k^{\wh} > \sqrt{\gamma} \\
0, & \text{ otherwise}
\end{cases}.
\end{align}

\end{enumerate}
\end{thm}

\begin{rmk}
\label{rmk:lowrank}
In fact, as will be evident from its proof Theorem \ref{thm-main1} is applicable to any rank $r$ matrix $X^{\wh}$, viewing $u_k^{\wh}$ and $v_k^{\wh}$ as the singular vectors of $X^{\wh}$. In particular, the columns of $X^{\wh}$ need not be drawn iid from a mean zero distribution. All that is needed for Theorem \ref{thm-main1} is that the singular values of $X^{\wh}$ remain constant as $p$ and $n$ grow, and that the parameters $\tau_k$ are well-defined.
\end{rmk}

Theorem \ref{thm-main1} is concerned only with the standard spiked model with white noise, $Y_j^{\wh} = X_j^{\wh} + G_j$. By contrast, the next result, Theorem \ref{thm-main2}, deals with the spiked model with colored noise, $Y_j = X_j + \ep_j$, where $\ep_j \sim N(0,\Sigma_\ep)$. In Section \ref{sec-model}, we defined the whitening matrix $W = \Sigma_\ep^{-1/2}$ that transforms $Y_j$ into the standard white-noise model $Y_j^{\wh}$; that is, $Y_j^{\wh} = W Y_j = W X_j + W \ep_j = X_j^{\wh} + G_j$. In stating and applying Theorem \ref{thm-main2}, we refer to the parameters for both models described in Section \ref{sec-model}.



\begin{thm}
\label{thm-main2}
Assume that the PCs $u_1,\dots,u_r$ satisfy the weighted orthogonality condition \eqref{incoherence-condition}, i.e., for $1 \le j \ne k \le r$,
\begin{align}
\lim_{p \to \infty} u_j^\top W^2 u_k = 0.
\end{align}
Order the principal components of $X_j$ by decreasing value of $\ell_k \tau_k$, as in \eqref{ell_ordering}; that is, we assume $\Sigma_x = \sum_{k=1}^{r} \ell_k u_k u_k^\top$, with
\begin{align}
\ell_1 \tau_1 > \dots > \ell_r \tau_r > 0,
\end{align}
where $\tau_k = \lim_{p\to\infty} \|W^{-1} u_k^{\wh}\|^{-2}$ as in \eqref{tau-def}.

Define $c_k > 0$, $1 \le k \le r$, by:
\begin{align}
\label{prod-unwhite}
c_k^2 \equiv
\begin{cases}
\frac{(c_k^{\wh})^2}
            {(c_k^{\wh})^2 + (s_k^{\wh})^2 \cdot \mu_\ep \cdot \tau_k},
    & \text{ if } \ell_k^{\wh} > \sqrt{\gamma} \\
0,  & \text{otherwise}
\end{cases},
\end{align}
where $c_k^{\wh}$ is given by \eqref{cos_out}, $\ell_k^{\wh}$ is defined from \eqref{eq-xh} with $X_j^{\wh} = W X_j$, and $\mu_{\ep} = \lim_{p \to \infty} \frac{\tr{\Sigma_\ep}}{p}$ as in \eqref{mu_def}.

Then for any $1 \le j , k \le r$,

\begin{enumerate}
\item
The vectors $\overline u_k$ and $u_k$ are almost surely asymptotically identical:
\begin{align}
\lim_{p \to \infty} \langle u_k, \overline u_k \rangle^2 = 1.
\end{align}

\item
The vectors $v_k^{\wh}$ and $v_k$ are almost surely asymptotically identical:
\begin{align}
\lim_{n \to \infty} \langle v_k, v_k^{\wh} \rangle^2 = 1.
\end{align}

\item
The inner product between $u_j$ and $\hat u_k$ converges almost surely:
\begin{align}
\label{cos_out2}
\lim_{p \to \infty} \langle u_j, \hat u_k \rangle^2 = 
\begin{cases}
c_k^2,
 & \text{ if } j=k \text{ and } \ell_k^{\wh} > \sqrt{\gamma} \\
0, & \text{ otherwise }
\end{cases},
\end{align}
where $c_k^2$ is defined in \eqref{prod-unwhite}.

\item
The vectors $\hat u_j$ and $\hat u_k$ are asymptotically orthogonal if $j \ne k$:
\begin{align}
\lim_{p \to \infty} \langle \hat u_j , \hat u_k \rangle^2 = \delta_{jk}.
\end{align}

\item
The parameter $\tau_k$ is almost surely asymptotically equal to $\|W u_k\|^2$:
\begin{align}
\lim_{p \to \infty} (\tau_k - \| W u_k\|^2 ) = 0.
\end{align}

\item
The variance $\ell_k^{\wh}$ of $X_j^{\wh}$ along $u_k^{\wh}$ is almost surely asymptotically equal to $ \ell_k \tau_k$:
\begin{align}
\lim_{p \to \infty} (\ell_k^{\wh} - \ell_k\tau_k) = 0.
\end{align}

\end{enumerate}

\end{thm}

The proofs for both Theorem \ref{thm-main1} and Theorem \ref{thm-main2} may be found in Appendix \ref{proof-asymptotics}.

\section{Optimal spectral shrinkage with whitening}
\label{sec-hidim}
In this section, we will derive the optimal spectral shrinkers for signal prediction and covariance estimation to be used in conjunction with whitening.

%

\subsection{Singular value shrinkage}
\label{sec-shrinker}

Given the noisy matrix $Y = X + N$, we consider a class of predictors of $X$ defined as follows. First, we whiten the noise, replacing $Y$ with $Y^{\wh} = WY$. We then apply singular value shrinkage to the transformed matrix $Y^{\wh}$. That is, if $\hat u_1^{\wh},\dots,\hat u_r^{\wh}$ and $\hat v_1^{\wh},\dots, \hat v_r^{\wh}$ are the top left and right singular vectors of $Y^{\wh}$, we define the new matrix
\begin{align}
\label{hat_xh}
\hat X^{\wh} = \sum_{k=1}^{r} t_k \hat u_k^{\wh} (\hat v_k^{\wh})^\top,
\end{align}
for some positive scalars $t_k$ which we have yet to determine.

Finally, we recolor the noise, to bring the data back to its original scaling. That is, we define our final predictor $\hat X$ by
\begin{align}
\hat X = W^{-1} \hat X^{\wh}.
\end{align}

In this section, we will show how to optimally choose the singular values $t_1,\dots,t_r$ in \eqref{hat_xh} to minimize the AMSE:
\begin{align}
\AMSE = \lim_{n \to \infty} \E \| \hat X - X\|_{\Fr}^2.
\end{align}

\begin{rmk}
Loss functions other than Frobenius loss (i.e., mean-squared error) may be considered as well. This will be done for the problem of covariance estimation in Section \ref{sec-eigs}, where it is more natural \cite{donoho2018condition}. For recovering the data matrix $X$ itself, however, the MSE is the natural loss, and the optimal $t_k$ can be derived for minimizing the AMSE without any additional assumptions on the model.
\end{rmk}

Once we have whitened the noise, our resulting matrix $Y^{\wh} = X^{\wh} + G$ is from the standard spiked model and consequently satisfies the conditions of Theorem \ref{thm-main1}, since $G$ is a Gaussian matrix with iid $N(0,1)$ entries. We will apply the asymptotic results of Theorem \ref{thm-main1}, taking the matrix $A = W^{-1}$. Recalling the definitions of $\hat u_k$ and $\overline u_k$ from \eqref{emp-pcs} and \eqref{ubar}, respectively, we obtain an immediate corollary to Theorem \ref{thm-main1}:

\begin{cor}
For $1 \le k \le r$, the cosine between the vectors $\overline u_k$ and $\hat u_k$ converges almost surely:
\begin{align}
\lim_{p \to \infty} \langle \overline u_k, \hat u_k \rangle^2=
c_k^2 \equiv 
\begin{cases}
\frac{(c_k^{\wh})^2}{(c_k^{\wh})^2 + (s_k^{\wh})^2 \cdot \mu_\ep \cdot \tau_k}, 
    & \text{ if } \ell_k^{\wh} > \sqrt{\gamma} \\
0, & \text{ otherwise}
\end{cases}.
\end{align}
\end{cor}

We derive the optimal $t_k$. We write:
\begin{align}
X^{\wh} \sim \sum_{k=1}^{r} (\ell_k^{\wh})^{1/2} u_k^{\wh} (v_k^{\wh})^\top,
\end{align}
and so
\begin{align}
X = W^{-1} X^{\wh}
 \sim \sum_{k=1}^{r} (\ell_k^{\wh})^{1/2} W^{-1} u_k^{\wh} (v_k^{\wh})^\top
= \sum_{k=1}^{r} (\ell_k^{\wh} / \tau_k)^{1/2} \overline u_k (v_k^{\wh})^\top.
\end{align}

Furthermore,
\begin{align}
\hat X^{\wh} = \sum_{k=1}^{r} t_k \hat u_k^{\wh} (\hat v_k^{\wh})^\top
\end{align}
and so
\begin{align}
\hat X = W^{-1}\hat X^{\wh} 
= \sum_{k=1}^{r} t_k W^{-1} \hat u_k^{\wh} (\hat v_k^{\wh})^\top
= \sum_{k=1}^{r} t_k \|W^{-1} \hat u_k^{\wh}\| \hat u_k (\hat v_k^{\wh})^\top.
\end{align}

It is convenient to reparametrize the problem in terms of
\begin{align}
\label{eq:bar-ell}
\overline \ell_k \equiv \ell_k^{\wh} / \tau_k,
\end{align}
and
\begin{align}
\tilde t_k \equiv t_k \|W^{-1} \hat u_k^{\wh}\|
\sim t_k \left(\frac{(c_k^{\wh})^2}{\tau_k} + (s_k^{\wh})^2 \mu_\ep\right)^{1/2},
\end{align}
where we have used Theorem \ref{thm-main1}.

In this notation, we have $X = \sum_{k=1}^{r} \overline \ell_k^{1/2} \overline u_k (v_k^{\wh})^\top$, and $\hat X = \sum_{k=1}^{r} \tilde t_k \hat u_k (\hat v_k^{\wh})^\top$. From Theorem \ref{thm-main1}, the vectors $v_j^{\wh}$ and $\hat v_k^{\wh}$ are orthogonal if $j \ne k$, and the cosine between $v_k^{\wh}$ and $\hat v_k^{\wh}$ is $\tilde c_k \equiv \tilde c_k^{\wh}$. The derivation from Section \ref{sec-froshr} shows that the optimal values $\tilde t_k$ are then given by
\begin{align}
\tilde t_k = \overline \ell_k^{1/2} c_k \tilde c_k
\end{align}

For this to define a valid estimator, we must show how to estimate the values $\overline \ell_k$, $c_k$ and $\tilde c_k$ from the observed data itself.

To that end, from Theorem \ref{thm-main1} $\ell_k^{\wh}$ can be estimated by
\begin{align}
\ell_k^{\wh} = \frac{(\sigma_k^{\wh})^2 - 1 - \gamma + \sqrt{((\sigma_k^{\wh})^2 - 1 - \gamma)^2 - 4\gamma}}{2}
\end{align}
where $\sigma_k^{\wh}$ is the $k^{th}$ singular value of $Y^{\wh}$. The cosines $c_k^{\wh}$ and $\tilde c_k^{\wh}$ can then be estimated by formulas \eqref{cos_out} and \eqref{cos_inn}.

Now, rearranging part \ref{main1-norm} from Theorem \ref{thm-main1}, we can solve for $\tau_k$ in terms of the estimable quantities $c_k^{\wh}$, $s_k^{\wh}$, $\mu_\ep$ and $\|\Sigma_\ep^{1/2} \hat u_k^{\wh}\|^2$:
\begin{align}
\tau_k \sim \frac{(c_k^{\wh})^2}{\|\Sigma_\ep^{1/2} \hat u_k^{\wh}\|^2 - (s_k^{\wh})^2 \mu_\ep}.
\end{align}
Indeed, this quantity can be estimated consistently: $c_k^{\wh}$ and $s_k^{\wh}$ are estimable from \eqref{cos_out}, $\|\Sigma_\ep^{1/2} \hat u_k^{\wh}\|^2$ is directly observed, and $\mu_\ep \sim \tr{\Sigma_\ep}/p$.

Having estimated $\tau_k$, we apply formula $\overline \ell_k = \ell_k^{\wh} / \tau_k$, and formula \eqref{cos_out2} for $c_k$. This completes the derivation of the optimal singular value shrinker. The entire procedure is described in Algorithm \ref{alg:homshrink}.

\begin{figure}[h]
\centering
\includegraphics[scale=0.45]{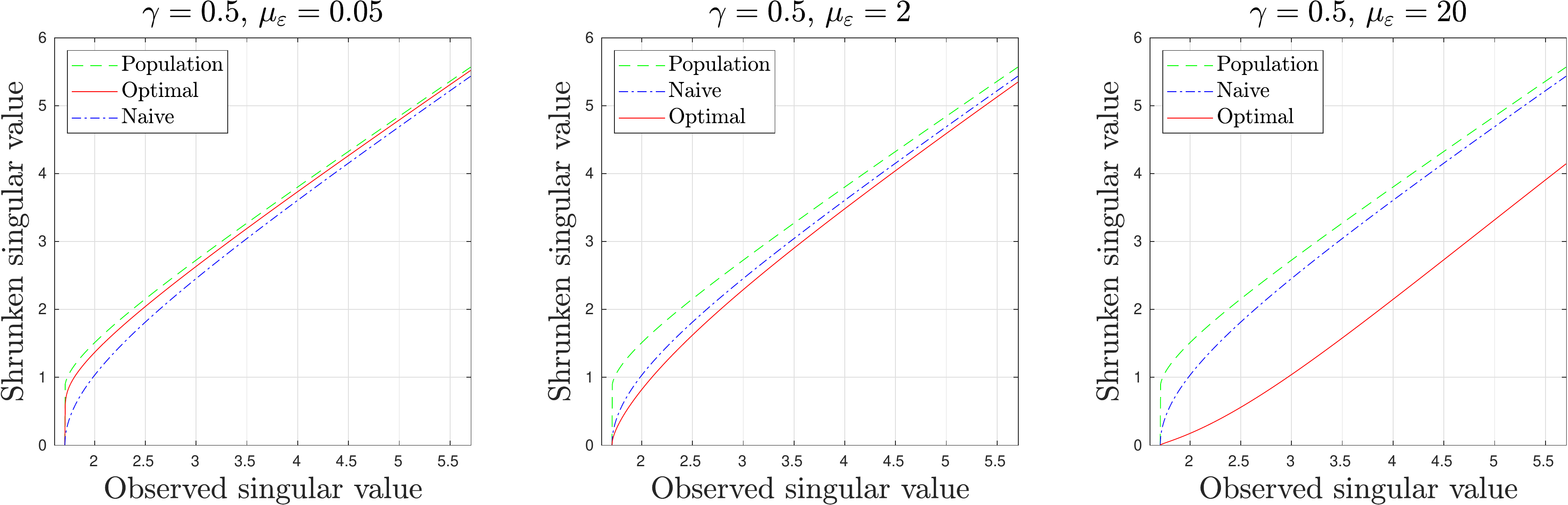}
\caption{Optimal shrinker, naive shrinker, and population shrinker, for $\tau=1$ and $\gamma=0.5$.}
\label{fig-shrinkers}
\end{figure}

Figures \ref{fig-shrinkers} and \ref{fig-shrinkers2} plot the optimal shrinker, i.e., the function that sends each top observed singular value $\sigma_k^{\wh}$ of $Y^{\wh}$ to the optimal $t_k$. For contrast, we also plot the ``population'' shrinker, which maps $\sigma_k^{\wh}$ to the corresponding $\sqrt{\ell_k^{\wh}}$; and the ``naive'' shrinker, which maps $\sigma_k^{\wh}$ to $\sqrt{\ell_k^{\wh}} c_k^{\wh} \tilde c_k^{\wh}$. This latter shrinker is considered in the paper \cite{dobriban2017optimal}, and is naive in that it optimizes the Frobenius loss before the unwhitening step without accounting for the change in angles between singular vectors resulting from unwhitening. In Figure \ref{fig-shrinkers} we set $\gamma = 0.5$, while in Figure \ref{fig-shrinkers2} we set $\gamma=2$. We fix $\tau=1$ but consider different values of $\mu_\ep$ (the behavior depends only on the ratio of $\mu_\ep$ and $\tau$).

\begin{figure}[h]
\centering
\includegraphics[scale=0.45]{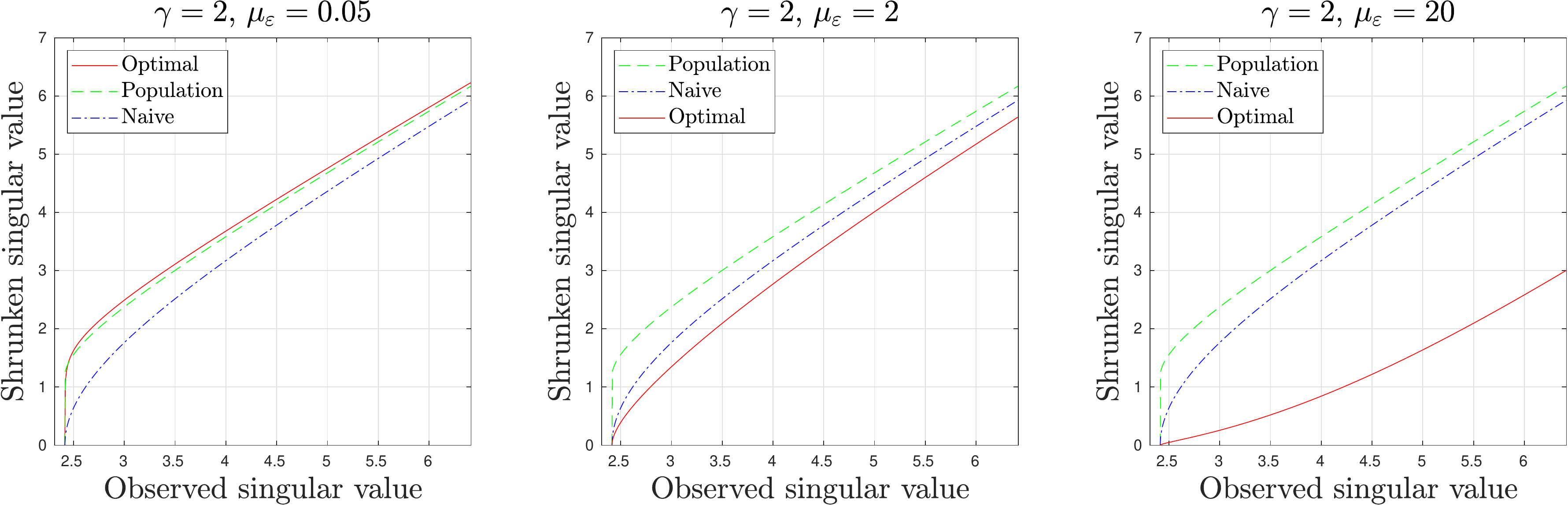}
\caption{Optimal shrinker, naive shrinker, and population shrinker, for $\tau=1$ and $\gamma=2$.}
\label{fig-shrinkers2}
\end{figure}

\begin{algorithm}[ht]
\caption{Optimal singular value shrinkage with whitening}
\label{alg:homshrink}
\begin{algorithmic}[1]
\item {\bf Input}: observations $Y_1,\dots,Y_n$; noise covariance $\Sigma_\ep$; rank $r$

\item
{\bf Define} $Y = [Y_1,\dots,Y_n]/\sqrt{n}$; $W = \Sigma_\ep^{-1/2}$;$Y^{\wh} = WY$

\item
{\bf Compute} rank $r$ SVD of $Y^{\wh}$: $\hat u_1^{\wh},\dots, \hat u_r^{\wh}$;
             $\hat v_1^{\wh},\dots, \hat v_r^{\wh}$; $\sigma_1^{\wh},\dots,\sigma_r^{\wh}$

\ForAll{$k=1,\dots,r$}
\If{$\sigma_k^{\wh} > 1 + \sqrt{\gamma}$}
\begin{tabbing}
\hspace*{1.2cm}\=\kill
      \> $\ell_k^{\wh} = \left[ (\sigma_k^{\wh})^2 - 1 - \gamma +
            \sqrt{((\sigma_k^{\wh})^2 - 1 - \gamma)^2 - 4\gamma}\right] \big/ 2$ \\
          \>  $c_k^{\wh} = \sqrt{ \left(1 - \gamma/(\ell_k^{\wh})^2\right) \big/
                \left(1 + \gamma /  \ell_k^{\wh}\right)}$ \\
          \> $s_k^{\wh} = \sqrt{1 - (c_k^{\wh})^2}$\\
          \>  $\tilde c_k = \sqrt{ \left(1 - \gamma/(\ell_k^{\wh})^2\right) \big/
                \left(1 + 1 /  \ell_k^{\wh}\right)}$ \\
          \> $\mu_\ep = \tr{\Sigma_\ep} / p$  \\
          \> $\tau_k 
                = (c_k^{\wh})^2 
          \big/ \left[\|\Sigma_\ep^{1/2} \hat u_k^{\wh}\|^2 - (s_k^{\wh})^2 \mu_\ep\right]$ \\
          \> $t_k = (\ell_k^{\wh})^{1/2} c_k^{\wh} \tilde c_k 
            \big/ \left[(c_k^{\wh})^2 + (s_k^{\wh})^2 \mu_\ep \tau_k\right]$
\end{tabbing}
\ElsIf{$\sigma_k^{\wh} \le 1 + \sqrt{\gamma}$}
\begin{tabbing}
\hspace*{1.2cm}\=\kill
          \> $t_k = 0$
\end{tabbing}


\EndIf{}
\EndFor

\item 
{\bf Output}: $\hat X = W^{-1}\sum_{k=1}^{r} t_k \hat u_k^{\wh} (\hat v_k^{\wh})^\top$

\end{algorithmic}
\end{algorithm}

\begin{rmk}
\label{rmk-rank}
In practice, the rank $r$ may not be known a priori. In Section \ref{sec-rank}, we describe several methods for estimating $r$ from the data.
\end{rmk}

\begin{rmk}
Algorithm \ref{alg:homshrink} may be applied to denoising any rank $r$ matrix $X$ from the observed matrix $Y = X+N$. As pointed out in Remark \ref{rmk:lowrank}, the assumption that the columns of $X$ are drawn iid from a mean zero distribution with covariance $\Sigma_x$ is not needed for the parameter estimates used by Algorithm \ref{alg:homshrink} to be applicable, so long as the singular values of the whitened matrix $X^{\wh}$ stay fixed (or converge almost surely) as $p$ and $n$ grow, and the parameters $\tau_k$ are well-defined.
\end{rmk}

\subsection{Eigenvalue shrinkage}
\label{sec-eigs}

We turn now to the task of estimating the covariance $\Sigma_x$ of $X_j$. Throughout this section, we will assume the conditions of Theorem \ref{thm-main2}, namely conditon \eqref{incoherence-condition}.

Analogous to the procedure for singular value shrinkage with whitening, we consider the procedure of eigenvalue shrinkage with whitening. We first whiten the observations $Y_j$, producing new observations $Y_j^{\wh} = W Y_j$. We then form the sample covariance $\hat \Sigma_y^{\wh}$ of the $Y_j^{\wh}$. We apply eigenvalue shrinkage to $\hat \Sigma_y^{\wh}$, forming a matrix of the form
\begin{align}
\hat \Sigma_x^{\wh} =  \sum_{k=1}^{r} t_k^2 \hat u_k^{\wh} (\hat u_k^{\wh} )^\top,
\end{align}
where $\hat u_1,\dots, \hat u_r^{\wh}$ are the top $r$ eigenvectors of $\hat \Sigma_y^{\wh}$, or equivalently the top $r$ left singular vectors of the whitened data matrix $Y^{\wh}$; and the $t_k^2$ are the parameters we will determine. Finally, we form our final estimator of $\Sigma_x$ by unwhitening:
\begin{align}
\hat \Sigma_x = W^{-1} \hat \Sigma_x^{\wh} W^{-1}.
\end{align}

It remains to define the eigenvalues $t_1^2,\dots,t_r^2$ of the matrix $\hat \Sigma_x^{\wh}$. We let $\L$ denote any of the loss functions considered in \cite{donoho2018optimal}. As a reminder, all these loss functions satisfy two conditions. First, they are \emph{orthogonally-invariant}, meaning that if both the estimated and population PCs are rotated, the loss does not change. Second, they are \emph{block-decomposable}, meaning that if both the estimated and population covariance matrices are in block-diagonal form, the loss can be written as functions of the losses between the individual blocks.

The estimated covariance matrix $\hat \Sigma_x = W^{-1} \hat \Sigma_x^{\wh} W^{-1}$ can be written as:
\begin{align}
\hat \Sigma_x 
= W^{-1} \hat \Sigma_x^{\wh} W^{-1} 
= \sum_{k=1}^{r} t_k^2 W^{-1} \hat u_k^{\wh} (W^{-1}\hat u_k^{\wh})^\top
= \sum_{k=1}^{r} t_k^2 \|W^{-1} \hat u_k^{\wh}\|^2 \hat u_k \hat u_k^\top
= \sum_{k=1}^{r} \tilde t_k^2 \hat u_k \hat u_k^\top,
\end{align}
where we have defined $\tilde t_k^2$ by:
\begin{align}
\label{tildetk2}
\tilde t_k^2 \equiv t_k^2 \|W^{-1} \hat u_k^{\wh}\|^2 
\sim t_k^2 \left(\frac{(c_k^{\wh})^2}{\tau_k} + (s_k^{\wh})^2 \mu_\ep\right).
\end{align}
We also write out the eigendecomposition of $\Sigma_x$:
\begin{align}
\Sigma_x = \sum_{k=1}^{r} \ell_k u_k u_k^\top.
\end{align}
From Theorem \ref{thm-main2}, the empirical PCs $\hat u_1,\dots,\hat u_r$ are asymptotically pairwise orthonormal, and $\hat u_j$ and $u_k$ are asymptotically orthogonal if $j \ne k$, and have absolute inner product $c_k$ when $j = k$, given by \eqref{prod-unwhite}.

Consequently, from Section \ref{sec-eig-review} the optimal $\tilde t_k^2$ are defined by:
\begin{align}
\label{argmin125}
\tilde t_k^2 = \argmin_{\ell} \L(A_k, \ell B_k),
\end{align}
where:
\begin{align}
A_k = 
\left(
\begin{array}{c c}
\ell_k & 0 \\
 0     & 0 \\
\end{array}
\right),
\end{align}
and
\begin{align}
B_k = 
\left(
\begin{array}{c c}
c_k^2        & c_k \sqrt{1-c_k^2} \\
 c_k \sqrt{1-c_k^2}     & 1-c_k^2 \\
\end{array}
\right).
\end{align}
As noted in Section \ref{sec-eig-review}, \cite{donoho2018optimal} provides closed form solutions to this minimization problem for many loss functions $\L$. For example, when operator norm loss is used the optimal $\tilde t_k^2$ is $\ell_k$, and when Frobenius norm loss is used, the optimal $\tilde t_k^2$ is $\ell_k c_k^2$. When no such closed formula is known, the optimal values may be obtained by numerical minimization of \eqref{argmin125}.

Finally, the eigenvalues $t_k^2$ are obtained by inverting formula \eqref{tildetk2}:
\begin{align}
t_k^2 = \tilde t_k^2 \left(\frac{(c_k^{\wh})^2}{\tau_k} + (s_k^{\wh})^2 \mu_\ep\right)^{-1}.
\end{align}
We summarize the covariance estimation procedure in Algorithm \ref{alg:cov-est}.

\begin{algorithm}[ht]
\caption{Optimal eigenvalue shrinkage with whitening}
\label{alg:cov-est}
\begin{algorithmic}[1]
\item {\bf Input}: observations $Y_1,\dots,Y_n$; noise covariance $\Sigma_\ep$; rank $r$

\item
{\bf Define} $Y = [Y_1,\dots,Y_n]/\sqrt{n}$; $W = \Sigma_\ep^{-1/2}$; $Y^{\wh} = WY$

\item
{\bf Compute} top $r$ left singular vectors/values of $Y^{\wh}$: $\hat u_1^{\wh},\dots, \hat u_r^{\wh}$;
             $\sigma_1^{\wh},\dots,\sigma_r^{\wh}$

\ForAll{$k=1,\dots,r$}
\If{$\sigma_k^{\wh} > 1 + \sqrt{\gamma}$}
\begin{tabbing}
\hspace*{1.2cm}\=\kill
      \> $\ell_k^{\wh} = \left[ (\sigma_k^{\wh})^2 - 1 - \gamma +
            \sqrt{((\sigma_k^{\wh})^2 - 1 - \gamma)^2 - 4\gamma}\right] \big/ 2$ \\
          \>  $c_k^{\wh} = \sqrt{ \left(1 - \gamma/(\ell_k^{\wh})^2\right) \big/
                \left(1 + \gamma /  \ell_k^{\wh}\right)}$ \\
          \> $\mu_\ep = \tr{\Sigma_\ep} / p$  \\
          \> $\tau_k 
                = (c_k^{\wh})^2 
          \big/ \left[\|\Sigma_\ep^{1/2} \hat u_k^{\wh}\|^2 - (1-(c_k^{\wh})^2) \mu_\ep\right]$ \\
          \> $\ell_k = \ell_k^{\wh} / \tau_k$  \\
          \> $c_k = c_k^{\wh} / \sqrt{(c_k^{\wh})^2+(1-(c_k^{\wh})^2) \mu_\ep \tau_k}$  \\
          \> $A_k = 
                \left(
                \begin{array}{c c}
                \ell_k & 0 \\
                 0     & 0 \\
                \end{array}
                \right)$                \\
          \> $B_k = 
                \left(
                \begin{array}{c c}
                c_k^2        & c_k \sqrt{1-c_k^2} \\
                 c_k \sqrt{1-c_k^2}     & 1-c_k^2 \\
                \end{array}
                \right)$               \\
          \> $\tilde t_k^2 = \argmin_{\ell} \L(A_k, \ell B_k)$ \\
          \> $t_k^2 = \tilde t_k^2 \tau_k / [(c_k^{\wh})^2 + (1-(c_k^{\wh})^2) \mu_\ep \tau_k]$
\end{tabbing}
\ElsIf{$\sigma_k^{\wh} \le 1 + \sqrt{\gamma}$}
\begin{tabbing}
\hspace*{1.2cm}\=\kill
          \> $t_k^2 = 0$
\end{tabbing}


\EndIf{}
\EndFor

\item 
{\bf Output}: 
    $\hat \Sigma_x = \sum_{k=1}^{r} t_k^2 (W^{-1}\hat u_k^{\wh}) (W^{-1}\hat u_k^{\wh})^\top$

\end{algorithmic}
\end{algorithm}

\begin{rmk}
As stated in Remark \ref{rmk-rank}, in practice the rank $r$ will likely not be known a priori. We refer to Section \ref{sec-rank} for a description of data-driven methods that may be used to estimate $r$.
\end{rmk}

\subsection{Estimating the noise covariance $\Sigma_{\ep}$}
\label{sec-noise}

Algorithms \ref{alg:homshrink} and \ref{alg:cov-est} require access to the whitening transformation $W = \Sigma_\ep^{-1/2}$, or equivalently the noise covariance matrix $\Sigma_\ep$. However, the same method and analysis goes through unscathed if $\Sigma_\ep$ is replaced with an estimate $\hat \Sigma_\ep$ that is consistent in operator norm, i.e., where
\begin{align}
\label{eq-sighat}
\lim_{p \to \infty} \|\Sigma_\ep - \hat \Sigma_\ep\|_{\op} = 0
\end{align}
almost surely as $p/n \to \gamma$. Indeed, the distribution of the top $r$ singular values and singular vectors of $Y^{\wh}$ will be asymptotically identical whether the true $W = \Sigma_\ep^{-1/2}$ is used to perform whitening or the estimated $\hat W = \hat \Sigma_\ep^{-1/2}$ is used instead.

\begin{rmk}
Because we assume that the maximum eigenvalue of $\Sigma_\ep$ is bounded and the minimum eigenvalue is bounded away from $0$, \eqref{eq-sighat} is equivalent to consistent estimation of the whitening matrix $W = \Sigma_\ep^{-1/2}$ by $\hat W = \hat \Sigma_\ep^{-1/2}$.
\end{rmk}

An estimator $\hat \Sigma_\ep$ satisfying \eqref{eq-sighat} may be obtained when we have access to an iid sequence of pure noise vectors $\ep_1,\dots,\ep_{n'}$ in addition to the $n$ signal-plus-noise vectors $Y_1,\dots,Y_n$. This is the setting considered in \cite{nadakuditi2010sample}, where a number of applications are also discussed. Here, we assume that $n' = n'(n)$ grows faster than $p = p(n)$, that is,
\begin{align}
\lim_{n \to \infty} \frac{p(n)}{n'(n)} = 0.
\end{align}
In this case, we replace $\Sigma_\ep$ by the sample covariance:
\begin{align}
\label{eq-sample}
\hat \Sigma_\ep = \frac{1}{n'} \sum_{j=1}^{n'} \ep_j \ep_j^\top,
\end{align}
which converges to $\Sigma_\ep$ in operator norm; that is, \eqref{eq-sighat} holds. In Section \ref{sec:sample-covariance}, we will illustrate the use of this method in simulations.

\begin{rmk}
If $p/n'$ does not converge to $0$, then $\hat \Sigma_\ep$ given by \eqref{eq-sample} is \emph{not} a consistent estimator of $\Sigma_\ep$ in operator norm. Indeed, when $\Sigma_\ep = I_p$ the distribution of $\hat \Sigma_\ep$'s eigenvalues converges to the Marchenko-Pastur law \cite{marchenko1967distribution}, and more generally converges to a distribution whose Stieltjes transform is implicitly defined by a fixed point equation \cite{bai2009spectral, silverstein1995empirical, silverstein1995strong}.
\end{rmk}

\subsubsection{Alternative estimators of $\Sigma_\ep$}

Without access to an independent sequence of $n' \gg p$ pure noise samples, estimating the noise covariance $\Sigma_\ep$ consistently (with respect to operator norm) is usually hard as $p \to \infty$. However, it may still be practical when $\Sigma_\ep$ is structured. Examples include: when $\Sigma_\ep$ is sparse \cite{bickel2008covariance}; when $\Sigma_\ep^{-1}$ is sparse \cite{yuan2010high}; when $\Sigma_\ep$ is a circulant or Toeplitz matrix, corresponding to stationary noise \cite{cai2013optimal}; and more generally, when the eigenbasis of $\Sigma_\ep$ is known a priori.

To elaborate on the last condition, let us suppose that the eigenbasis of $\Sigma_\ep$ is known, and without loss of generality that $\Sigma_\ep$ is diagonal; and suppose that and the $u_k$'s are \emph{delocalized} in that $\| u_k\|_\infty \to 0$ as $p \to \infty$. Write $\Sigma_\ep = \diag(\nu_1,\dots,\nu_p)$, for unknown $\nu_i$. In this setting, the sample variance of each coordinate will converge almost surely to the variance of the noise in that coordinate; that is, for $i=1,\dots,p$, we have:
\begin{align}
\hat \nu_i = \frac{1}{n}\sum_{j=1}^{n} Y_{ij}^2
= \frac{1}{n}\sum_{j=1}^{n}  \left( \sum_{k=1}^{r} \ell_k u_{ki} z_{jk} \right)^2
    + \frac{1}{n}\sum_{j=1}^{n} \ep_{ij}^2 
    + 2\frac{1}{n}\sum_{j=1}^{n} \ep_{ij} \sum_{k=1}^{r} \ell_k u_{ki} z_{jk}
\to \nu_i,
\end{align}
where the limit is almost sure as $p,n \to \infty$. We have made use of the strong law of large numbers and the limit $\| u_k \|_\infty \to 0$.

Let $\hat \Sigma_\ep$ have $i^{th}$ diagonal entry $\hat \nu_i$. Then $\hat \Sigma_\ep - \Sigma_\ep$ is a mean-zero diagonal matrix, with diagonal entries $\hat \nu_i - \nu_i$; and the operator norm $\| \hat \Sigma_\ep - \Sigma_\ep \|_{\op} = \max_{1 \le i \le p} |\hat \nu_i - \nu_i|$, which is easily shown to go to $0$ almost surely as $p \to \infty$ using the subgaussianity of the observations.

%

\subsection{Estimating the rank $r$}
\label{sec-rank}

A challenging question in principal component analysis is selecting the number of components corresponding to signal, and separating these from the noise. In our model, this corresponds to estimating the rank $r$ of the matrix $X$, which is an input to Algorithms \ref{alg:homshrink} and \ref{alg:cov-est}. A simple and natural estimate $\hat r$ of the rank is the following:
\begin{align}
\hat r = \min \{k : \sigma_k^{\wh} > 1 + \sqrt{\gamma} + \epsilon_n\}.
\end{align}
That is, we estimate the rank as the number of singular values of $Y^{\wh} = X^{\wh} + G$ exceeding the largest singular value of the noise matrix $G$, plus a small finite-sample correction factor $\epsilon_n > 0$. Any singular value exceeding $1 + \sqrt{\gamma} + \epsilon_n$ is attributable to signal, whereas any value below is consistent with pure noise.

When $\epsilon_n \equiv \epsilon$ for all $n$, it may be shown that in the large $p$, large $n$ limit, $\hat r$ converges almost surely to the number of singular values of $X^{\wh}$ exceeding $1 + \sqrt{\gamma} + \epsilon$. For small enough $\epsilon$, this will recover all singular values of $X^{\wh}$ exceeding $\sqrt{\gamma}$, and  is likely sufficient for many applications. Furthermore, the correction $\epsilon_n$ may be calibrated using the Tracy-Widom distribution of the operator norm of $GG^\top$ by taking $\epsilon_n \sim n^{-2/3}$. Though a detailed discussion is beyond the scope of this paper, we refer to \cite{kritchman2008determining} for an approach along these lines.

An alternative procedure is similar to $\hat r$, but uses the original matrix $Y$ rather than the whitened matrix $Y^{\wh}$:
\begin{align}
\hat r' = \min \{k : \sigma_k > b_+ + \epsilon_n\},
\end{align}
where $b_+$ is the asymptotic operator norm of the noise matrix $N$, and $\epsilon_n$ is a finite-sample correction factor. The value $b_+$ may be evaluated using, for example, the method from \cite{leeb2020rapid}. An estimator like this is proposed in \cite{nadakuditi2014optshrink}. In Section \ref{sec-numerical-rank}, we present numerical evidence that $\hat r$ may outperform $\hat r'$. More precisely, it appears that whitening can increase the gap between the smallest signal singular value and the bulk edge of the noise, making detection of the signal components more reliable.

\begin{rmk}
We also remark that a widely-used method for rank estimation in non-isotropic noise is known as \emph{parallel analysis} \cite{horn1965rationale, buja1992remarks,brown2014confirmatory}, which has been the subject of recent investigation \cite{dobriban2017factor, dobriban2017deterministic}. Other methods have also been explored \cite{josse2012selecting}.
\end{rmk}

\section{Singular value shrinkage and linear prediction}
\label{sec-linpred}

In this section, we examine the relationship between singular value shrinkage and linear prediction. A linear predictor of $X_j$ from $Y_j$ is of the form $A Y_j$, where $A$ is a fixed matrix. It is known (see, e.g.\ \cite{mackay2004deconv}) that to minimize the expected mean-squared error, the best linear predictor, also called the \emph{Wiener filter}, takes $A = \Sigma_x \left( \Sigma_x + \Sigma_\ep \right)^{-1}$, and hence is of the form:
\begin{align}
\label{eq-blp}
\hat X_j^{\opt} &= \Sigma_x \left( \Sigma_x + \Sigma_\ep \right)^{-1} Y_j.
\end{align}

We will prove the following result, which shows that in the classical regime $\gamma \to 0$, optimal shrinkage with whitening converges to the Wiener filter.

\begin{thm}
\label{thm-blp}
Suppose $Y_1,\dots,Y_n$ are drawn from the spiked model with heteroscedastic noise, $Y_j = X_j + \ep_j$. Let $\hat X_1,\dots, \hat X_n$ be the predictors of $X_1,\dots,X_n$ obtained from singular value shrinkage with whitening, as described in Section \ref{sec-shrinker} and Algorithm \ref{alg:homshrink}. Then almost surely in the limit $p/n \to 0$,
\begin{align}
\lim_{n \to \infty} \|\hat X^{\opt} - \hat X\|_{\Fr}^2 
    = \lim_{n \to \infty} \frac{1}{n} \sum_{j=1}^{n} \| \hat X_j^{\opt} - \hat X_j\|^2 
    = 0.
\end{align}
In other words, the predictor $\hat X_j$ is asymptotically equivalent to the best linear predictor $\hat X_{j}^{\opt}$.

\end{thm}

Theorem \ref{thm-blp} is a consequence of the following result.

\begin{thm}
\label{thm-general}
Suppose that the numbers $s_k$, $1 \le k \le r$ satisfy
\begin{align}
\lim_{\gamma \to 0} \frac{s_k}{\sigma_k^{\wh}} = \frac{\ell_k^{\wh}}{\ell_k^{\wh} + 1}.
\end{align}
Then the predictor defined by
\begin{align}
\hat X' = \sum_{k=1}^{r} s_k W^{-1} \hat u_k^{\wh} (\hat v_k^{\wh})^\top
\end{align}
satisfies
\begin{align}
\lim_{n \to \infty} \|\hat X^{\opt} - \hat X'\|_{\Fr}^2 
    = 0,
\end{align}
where the limit holds almost surely as $p/n \to 0$.
\end{thm}

We will also show that in the context of shrinkage methods, whitening is an \emph{optimal} weighting of the data. To make this precise, we consider the following class of weighted shrinkage methods, which subsumes both ordinary singular value shrinkage and singular value shrinkage with noise whitening. For a fixed weight matrix $Q$, we multiply $Y$ by $Q$, forming the matrix $Y^q = [QY_1,\dots,QY_n] / \sqrt{n}$. We then apply singular value shrinkage to $Y^q$, with singular values $s_1^q,\dots,s_r^q$, after which we apply the inverse weighting $Q^{-1}$. Clearly, ordinary shrinkage is the special case when $Q = I_p$, whereas singular value shrinkage with whitening is the case when $Q=W = \Sigma_\ep^{-1/2}$.

When the singular values $s_1^q,\dots,s_r^q$ are chosen optimally to minimize the AMSE, we will call the resulting predictor $\hat X_Q$, and denote by $\hat X_{Q,j}$ the denoised vectors so that $\hat X_{Q} = [\hat X_{Q,1},\dots,\hat X_{Q,n}] / \sqrt{n}$. In this notation, $\hat X = \hat X_W$ is optimal shrinkage with whitening, whereas $\hat X_I$ is ordinary shrinkage without whitening. The natural question is, what is the optimal matrix $Q$? 

To answer this question, we introduce the linear predictors $\hat X_{Q,j}^{\lin}$, defined by
\begin{align}
\hat X_{Q,j}^{\lin} 
    = \sum_{k=1}^{r} \eta_k^q \langle Q Y_j , u_k^q \rangle Q^{-1} u_k^q,
\end{align}
where the $u_1^q,\dots,u_r^q$ are the eigenvectors of $Q \Sigma_x Q$, and the $\eta_k^q$ are chosen optimally to minimize the average AMSE across all $n$ observations. We prove the following result, which is again concerned with the classical $\gamma \to 0$ regime.

\begin{thm}
\label{prop-lin-pred}

Let $Q = Q_p$ be an element of a sequence of symmetric, positive-definite $p$-by-$p$ matrices with bounded operator norm ($\|Q_p\|_{\op} \le C < \infty$ for all $p$). Then in the limit $p/n \to 0$, we have almost surely:
\begin{align}
\label{linear1234}
\lim_{n \to \infty} \|\hat X_Q^{\lin} - \hat X_Q\|_{\Fr}^2 
    = \lim_{n \to \infty} \frac{1}{n} \sum_{j=1}^{n} \| \hat X_{Q,j}^{\lin} - \hat X_{Q,j}\|^2 
    = 0.
\end{align}
In other words, the weighted shrinkage predictor $\hat X_{Q,j}$ is asymptotically equal to the linear predictor $\hat X_{Q,j}^{\lin}$.

Furthermore, $Q=W$ minimizes the AMSE:
\begin{align}
\label{comparison641}
W = \argmin_{Q} \lim_{n \to \infty} \E \|\hat X_Q - X\|_{\Fr}^2.
\end{align}
\end{thm}

The first part of Theorem \ref{prop-lin-pred}, namely \eqref{linear1234}, states that any weighted shrinkage method converges to a linear predictor when $\gamma \to 0$. The second part of Theorem \ref{prop-lin-pred}, specifically \eqref{comparison641}, states that of all weighted shrinkage schemes, whitening is optimal in the $\gamma \to 0$ regime.

\begin{rmk}
A special case of Theorem \ref{thm-general} is the suboptimal ``naive'' shrinker with whitening, which uses singular values $\sqrt{\ell_k^{\wh}} c_k^{\wh} \tilde c_k^{\wh}$; see Figures \ref{fig-shrinkers} and \ref{fig-shrinkers2} and the accompanying text. It is easily shown that Theorem \ref{thm-general} applies to this shrinker, and consequently that in the $\gamma\to0$ limit this shrinker converges to the BLP. This fact will be illustrated numerically in Section \ref{compare-optshrink}.
\end{rmk}

We give detailed proofs of Theorems \ref{thm-blp}, \ref{thm-general} and \ref{prop-lin-pred} in Appendix \ref{proof-blp}. In Section \ref{sec-colwise}, we make a simple observation which underlies the proofs, which is of independent interest.


\subsection{Columns of weighted singular value shrinkage}
\label{sec-colwise}

In this section, we show how to write the predictor $\hat X_Q$ in terms of the individual columns of $Y^q = [QY_1,\dots,QY_n] / \sqrt{n}$. This observation will be used in the proofs of Theorems \ref{thm-blp}, \ref{thm-general} and \ref{prop-lin-pred}, and also motivates the form of the out-of-sample predictor we will study in Section \ref{sec-oos}.

Let $m = \min(p,n)$. Consistent with our previous notation (when $Q=W$), we will denote by $\hat u_1^q,\dots, \hat u_m^q$ the left singular vectors of the matrix $Y^q$, and we will denote by $\hat v_1^q,\dots, \hat v_m^q$ the right singular vectors and $\sigma_1^q,\dots,\sigma_m^q$ the corresponding singular values.

\begin{lem}
\label{lem-colwise}
Each column $\hat X_{Q,j}$ of $\sqrt{n} \cdot \hat X_Q$ is given by the formula
\begin{align}
\label{est000}
\hat X_{Q,j}
= Q^{-1}\sum_{k=1}^r 
        \eta_k^q  \langle QY_j, \hat u_k^q \rangle \hat u_k^q,
\end{align}
where $\eta_k^q = s_k^q / \sigma_k^q$ is the ratio of the new and old singular values.
\end{lem}

To see this, observe that we can write the $j^{th}$ column of the matrix $\sqrt{n}\cdot Y^q$ as:
\begin{align}
QY_j = \sum_{k=1}^{m} \sigma_k^q \hat u_k^q \hat v_{jk}^q,
\end{align}
and so by the orthogonality of $\hat u_k^q$, $\hat v_{jk}^q = \langle QY_j, \hat u_k^q \rangle / \sigma_k^q$. Consequently, when $\hat X_Q$ is obtained from $Y^q$ by singular value shrinkage with singular values $s_1^q,\dots,s_r^q$, followed by multiplication with $Q^{-1}$, we obtain formula \eqref{est000}.

\section{Out-of-sample prediction}
\label{sec-oos}

We now consider the problem of \emph{out-of-sample} prediction. In Section \ref{sec-colwise}, specifically Lemma \ref{lem-colwise}, we saw that when applying the method of  shrinkage with whitening, as described in Algorithm \ref{alg:homshrink}, each denoised vector $\hat X_j$ can be written in the form:
\begin{align}
\label{ins00}
\hat X_j = \sum_{k=1}^{r} \eta_k \langle W Y_j, \hat u_k^{\wh} \rangle W^{-1} \hat u_k^{\wh},
\end{align}
where $\hat u_1^{\wh},\dots, \hat u_r^{\wh}$ are the top $r$ left singular vectors of $Y^{\wh} = WY$, and $\eta_k$ are deterministic coefficients. We observe that the expression \eqref{ins00} may be evaluated for \emph{any} vector $Y_j$, even when it is not one of the original $Y_1,\dots,Y_n$, so long as we have access to the singular vectors $\hat u_k^{\wh}$.

To formalize the problem, we suppose we have computed the sample vectors $\hat u_1^{\wh},\dots, \hat u_r^{\wh}$ based on $n$ observed vectors $Y_1,\dots,Y_n$, which we will call the \emph{in-sample} observations. That is, the $\hat u_k^{\wh}$ are the top left singular vectors of the whitened matrix $Y^{\wh} = [Y_1^{\wh},\dots,Y_n^{\wh}] / \sqrt{n}$. We now receive a new observation $Y_0 = X_0 + \ep_0$ from the same distribution, which we will refer to as an \emph{out-of-sample} observation, and our goal is to predict the signal $X_0$.

We will consider predictors of the out-of-sample $X_0$ of the same form as \eqref{ins00}:
\begin{align}
\label{oos00}
\hat X_0 = \sum_{k=1}^{r} \eta_k^{\out} \langle WY_0, \hat u_k^{\wh} \rangle W^{-1} \hat u_k^{\wh}.
\end{align}
We wish to choose the coefficients $\eta_k^{\out}$ to minimize the AMSE, $\lim_{n \to \infty}\E \|\hat X_0 - X_0\|^2$.

\begin{rmk}
We emphasize the difference between the in-sample prediction \eqref{ins00} and the out-of-sample prediction \eqref{oos00}, beyond the different coefficients $\eta_k$ and $\eta_k^{\out}$. In \eqref{ins00}, the vectors $u_1^{\wh},\dots,u_r^{\wh}$ are \emph{dependent} on the in-sample observation $Y_j$, $1 \le j \le n$, because they are the top $r$ left singular vectors of $Y^{\wh}$. However, in \eqref{oos00} they are \emph{independent} of the out-of-sample observation $Y_0$, which is drawn independently from $Y_1,\dots,Y_n$. As we will see, it is this difference that necessitates the different choice of coefficients $\eta_k$ and $\eta_k^{\out}$ for the two problems.
\end{rmk}

In this section, we prove the following result comparing optimal out-of-sample prediction and in-sample prediction. Specifically, we derive the explicit formulas for the optimal out-of-sample coefficients $\eta_k^{\out}$ and the in-sample coefficients $\eta_k$; show that the coefficients are \emph{not} equal; and show that the AMSE for both problems are nevertheless \emph{identical}. Throughout this section, we assume the conditions and notation of Theorem
\ref{thm-main1}.

\begin{thm}
\label{prop-oos-summary}

Suppose $Y_1,\dots,Y_n$ are drawn iid from the spiked model, $Y_j = X_j + \ep_j$, and $\hat u_1^{\wh},\dots,\hat u_r^{\wh}$ are the top $r$ left singular vectors of $Y^{\wh}$. Suppose $Y_0 = X_0 + \ep_0$ is another sample from the same spiked model, drawn independently of $Y_1,\dots,Y_n$. Then the following results hold:

\begin{enumerate}

\item
The optimal in-sample coefficients $\eta_k$ are given by :
\begin{align}
\label{eta-ins}
\eta_k = \frac{(c_k^{\wh})^2}{(c_k^{\wh})^2 + (s_k^{\wh})^2 \mu_\ep \tau_k } 
            \cdot \frac{\ell_k^{\wh} }{\ell_k^{\wh}  + 1} .
\end{align}

\item
The optimal out-of-sample coefficients $\eta_k^{\out}$ are given by:
\begin{align}
\eta_k^{\out} = \frac{(c_k^{\wh})^2}{(c_k^{\wh})^2 + (s_k^{\wh})^2 \mu_\ep \tau_k } 
            \cdot \frac{\ell_k^{\wh}}{\ell_k^{\wh} (c_k^{\wh})^2 + 1} .
\end{align}
\item
The AMSEs for in-sample and out-of-sample prediction are identical, and equal to:
\begin{align}
\label{eq-amse}
\text{AMSE} = \sum_{k=1}^{r} \left( \frac{\ell_k^{\wh} }{\tau_k}
                - \frac{(\ell_k^{\wh})^2 (c_k^{\wh})^4}{\ell_k^{\wh} (c_k^{\wh})^2 + 1} 
                    \frac{1}{ \alpha_k \tau_k  } \right),
\end{align}
where $\alpha_k = \left((c_k^{\wh})^2 + (s_k^{\wh})^2 \mu_\ep \tau_k \right)^{-1}$.

\end{enumerate}
\end{thm}

\begin{rmk}
To be clear, denoising each in-sample observation $Y_1,\dots,Y_n$ by applying \eqref{ins00} with $\eta_k$ defined by \eqref{eta-ins} is \emph{identical} to denoising $Y_1,\dots,Y_n$ by singular value shrinkage with whitening described in Algorithm \ref{alg:homshrink}. We derive this alternate form only to show that the coefficients $\eta_k$ are different from the the optimal out-of-sample coefficients $\eta_k^{\out}$ to be used when $Y_0$ is independent from the $\hat u_k^{\wh}$.
\end{rmk}

\begin{rmk}
Theorem \ref{prop-oos-summary} extends the analogous result from \cite{dobriban2017optimal}, which was restricted to the standard spiked model with white noise.
\end{rmk}

The proof of Theorem \ref{prop-oos-summary} may be found in Appendix \ref{proofs-oos}. In Algorithm \ref{alg:oos}, we summarize the optimal out-of-sample prediction method, with the optimal coefficients derived in Theorem \ref{prop-oos-summary}.

\begin{algorithm}[ht]
\caption{Optimal out-of-sample prediction}
\label{alg:oos}
\begin{algorithmic}[1]
\item {\bf Input}: $Y_0$; $\hat u_1^{\wh},\dots, \hat u_r^{\wh}$; 
             $\sigma_1^{\wh},\dots,\sigma_r^{\wh}$

\ForAll{$k=1,\dots,r$}
\If{$\sigma_k^{\wh} > 1 + \sqrt{\gamma}$}
\begin{tabbing}
\hspace*{1.2cm}\=\kill
      \> $\ell_k^{\wh} = \left[ (\sigma_k^{\wh})^2 - 1 - \gamma +
            \sqrt{((\sigma_k^{\wh})^2 - 1 - \gamma)^2 - 4\gamma}\right] \big/ 2$ \\
          \>  $c_k^{\wh} = \sqrt{ \left(1 - \gamma/(\ell_k^{\wh})^2\right) \big/
                \left(1 + \gamma /  \ell_k^{\wh}\right)}$ \\
          \> $s_k^{\wh} = \sqrt{1 - (c_k^{\wh})^2}$\\
          \> $\mu_\ep = \tr{\Sigma_\ep} / p$  \\
          \> $\tau_k 
                = (c_k^{\wh})^2 
          \big/ \left[\|\Sigma_\ep^{1/2} \hat u_k^{\wh}\|^2 - (s_k^{\wh})^2 \mu_\ep\right]$ \\
          \>  $\alpha_k = 1 / \left((c_k^{\wh})^2 + (s_k^{\wh})^2 \mu_\ep \tau_k \right) $ \\
          \> $\eta_k^{\out} = \alpha_k \ell_k^{\wh} (c_k^{\wh})^2 / (\ell_k^{\wh} (c_k^{\wh})^2 + 1) $
\end{tabbing}
\ElsIf{$\sigma_k^{\wh} \le 1 + \sqrt{\gamma}$}
\begin{tabbing}
\hspace*{1.2cm}\=\kill
          \> $\eta_k^{\out} = 0 $
\end{tabbing}


\EndIf{}
\EndFor

\item 
{\bf Output}: $\hat X_0 
          = \sum_{k=1}^{r} \eta_k^{\out} \langle WY_0, \hat u_k^{\wh} \rangle W^{-1}\hat u_k^{\wh}$

\end{algorithmic}
\end{algorithm}


%

%

\section{Subspace estimation and PCA}
\label{sec-pca}

In this section, we focus on the task of \emph{principal component analysis (PCA)}, or the estimation of the principal components $u_1,\dots,u_r$ of the signal $X_j$, and their span. Specifically, we assess the quality of the empirical PCs $\hat u_1,\dots,\hat u_r$ defined in \eqref{emp-pcs}. The reader may recall that these are constructed by whitening the observed vectors $Y_j$ to produce $Y_j^{\wh}$; computing the top $r$ left singular vectors of $Y_j^{\wh}$; and unwhitening and normalizing. 

We first observe that in the classical regime $\gamma \to 0$, the angle between the subspaces $\sp\{\hat u_1,\dots,\hat u_r\}$ and $\sp\{u_1,\dots,u_r\}$ converges to 0 almost surely; we recall that the sine of the angle between subspaces $\mathcal{A}$ and $\mathcal{B}$ of $\R^p$ is defined by
\begin{align}
\sin \Theta(\mathcal{A}, \mathcal{B}) = \| A_\perp^\top B  \|_{\op},
\end{align}
where $A_\perp$ and $B$ are matrices whose columns are orthonormal bases of $\mathcal{A}^\perp$ and $\mathcal{B}$, respectively.

\begin{prop}
\label{prop-pc0}
Suppose $Y_1,\dots,Y_n$ are drawn from the spiked model, $Y_j = X_j + \ep_j$. Let $\U = \sp\{u_1,\dots,u_r \}$ be the span of the population PCs, and $\hat \U = \sp\{\hat u_1,\dots,\hat u_r\}$ be the span of the empirical PCs. Then
\begin{align}
\lim_{n \to 0} \sin \Theta(\U,\hat \U) = 0,
\end{align}
where the limit holds almost surely as $n \to \infty$ and $p/n \to 0$.
\end{prop}
The proof of Proposition \ref{prop-pc0} may be found in Appendix \ref{proofs-pca}.

Proposition \ref{prop-pc0} shows consistency of principal subspace estimation in the classical regime. We ask what happens in the high-dimensional setting $\gamma > 0$, where we typically do not expect to be able to have consistent estimation of the principal subspace. Our task here is to show that whitening will still improve estimation. To that end, in Section \ref{sec-pcs}, we will show that under a uniform prior on the population PCs $u_k$, whitening improves estimation of the PCs. In Section \ref{sec-minimax}, we will derive a bound on the error of estimating the principal subspace $\sp\{u_1,\dots,u_r\}$, under condition \eqref{incoherence-condition}; we will show that the error rate matches the optimal rate of the estimator in \cite{zhang2018heteroskedastic}. Finally, in Section \ref{sec-snr} we will complement these results by showing that under the uniform prior, whitening improves a natural signal-to-noise ratio.

%

\subsection{Whitening improves subspace estimation for generic PCs}
\label{sec-pcs}

In this section, we consider the effect of whitening on estimating the PCs $u_1,\dots,u_r$. More precisely, we contrast two estimators of the $u_k$. On the one hand, we shall denote by $\hat u_1^\prime,\dots \hat u_r^\prime$ the left singular vectors of the raw data matrix $Y$, without applying any weighting matrix. On the other hand, we consider the vectors $\hat u_1,\dots,\hat u_r$ obtained by whitening, taking the top singular vectors of $Y^{\wh}$, unwhitening, and normalizing, as expressed by formula \eqref{emp-pcs}.

We claim that ``generically'', the vectors $\hat u_1,\dots,\hat u_r$ are superior estimators of $u_1,\dots,u_r$. By ``generically'', we mean when we impose a uniform prior over the population PCs $u_1,\dots,u_r$; that is, we assume the $u_k$ are themselves random, drawn uniformly from the sphere in $\R^p$ and orthogonalized. This is precisely the ``orthonormalized model'' considered in \cite{benaych2012singular}.

We set $\tau = \lim_{p \to \infty} \tr{\Sigma_\ep^{-1}} / p$, assuming this limit exists; and let $\varphi = \tau \cdot \mu_\ep$. By Jensen's inequality, $\varphi \ge 1$, with strict inequality so long as $\Sigma_\ep$ is not a multiple of the identity.

\begin{thm}
\label{thm-angles}
Suppose $\Sigma_\ep$ has a finite number of distinct eigenvalues, each occurring with a fixed proportion as $p \to \infty$. Suppose too that $u_1,\dots,u_r$ are uniformly random orthonormal vectors in $\R^p$. Let $\hat u_1',\dots,\hat u_r'$ be the left singular vectors of $Y$, and $\hat u_1,\dots,\hat u_r$ be the empirical PCs defined by \eqref{emp-pcs}. Then with probability approaching $1$ as $n \to \infty$ and $p/n \to \gamma > 0$,
\begin{align}
|\langle \hat u_k' , u_k \rangle|^2  \le  R(\varphi) |\langle \hat u_k , u_k \rangle|^2,
    \quad 1 \le k \le r,
\end{align}
where $R$ is decreasing, $R(1) = 1$, and $R(\varphi) < 1$ for $\varphi > 1$.

Furthermore, if $\hat v_1',\dots,\hat v_r'$ are the right singular vectors of $Y$, and $\hat v_1,\dots,\hat v_r$ are the left singular vectors of $Y^{\wh}$, then
\begin{align}
|\langle \hat v_k' , z_k \rangle|^2  \le  \tilde R(\varphi)|\langle \hat v_k , z_k \rangle|^2,
    \quad 1 \le k \le r,
\end{align}
with probability approaching $1$ as $n \to \infty$ and $p/n \to \gamma > 0$, where $z_k = (z_{1k},\dots,z_{nk})^\top / \sqrt{n}$, and where $\tilde R$ is decreasing, $\tilde R(1) = 1$, and $\tilde R(\varphi) < 1$ for $\varphi > 1$.
\end{thm}

The proof of Theorem \ref{thm-angles} may be found in Appendix \ref{proofs-pca}. It rests on a result from the recent paper \cite{hong2018asymptotic}, combined with the formula \eqref{prod-unwhite} for the asymptotic cosines between $\hat u_k$ and $u_k$.

\begin{rmk}
\label{rmk-tau}
The definition of $\tau = \tr{\Sigma_\ep^{-1}} / p$ is consistent with our definition of $\tau_k = \lim_{p \to \infty} \|W^{-1} u_k^{\wh}\|^{-2}$ from \eqref{tau-def}. Indeed, since Theorem \ref{thm-angles} assumes that $u_1,\dots,u_r$ are uniformly random unit vectors, the PCs $u_k^{\wh}$ of $X^{\wh}$ are asymptotically identical to $W u_k / \|W u_k\|$, since these vectors are almost surely orthogonal as $p \to \infty$. Consequently, for each $1 \le k \le r$ we have
\begin{align}
\tau_k = \lim_{p \to \infty} \frac{1}{\|W^{-1} u_k^{\wh}  \|^2} 
= \lim_{p \to \infty} \|W u_k  \|^2
\sim \frac{1}{p} \tr{W^2} = \frac{1}{p} \tr{\Sigma_p^{-1}} \sim \tau.
\end{align}
\end{rmk}

\subsection{Minimax optimality of the empirical PCs}
\label{sec-minimax}

In this section, we consider the question of whether the empirical PCs $\hat u_1,\dots,\hat u_r$ can be significantly improved upon. In the recent paper \cite{zhang2018heteroskedastic}, an estimator $\hat \U$ of the principal subspace $\U = \sp\{ u_1,\dots,u_r\}$ is proposed that achieves the following error rate:
\begin{align}
\label{eq:cai-rate}
\E [\sin \Theta(\hat \U, \U)] \le  \min\left\{ C\sqrt{\gamma}
    \left( \frac{\mu_\ep^{1/2} + (r/p)^{1/2}\|\Sigma_\ep\|_{\op}^{1/2}}{ \min_k\ell_k^{1/2}}  
            + \frac{\mu_\ep^{1/2} \|\Sigma_\ep\|_{\op}^{1/2}}{\min_k \ell_k} \right) ,1 \right\},
\end{align}
where $C$ is a constant dependent on the \emph{incoherence} of $u_1,\dots,u_r$, defined by
\begin{math}
I(U) = \max_{1 \le j \le p} \|e_j^\top U\|^2
\end{math}
where $U = [u_1,\dots,u_r] \in \R^{p \times r}$. Furthermore, the error rate \eqref{eq:cai-rate} is shown to be \emph{minimax optimal} over the class of models with PCs of bounded incoherence.

In this section, we show that when \eqref{incoherence-condition} holds, then the empirical PCs $\hat u_1,\dots, \hat u_r$ achieve the same error rate \eqref{eq:cai-rate} almost surely in the limit $n \to \infty$, $p/n \to \gamma$. More precisely, we show the following:

%
\begin{thm}
\label{thm-sin-theta}
Assume that the weighted orthogonality condition \eqref{incoherence-condition} holds. Suppose that $\Sigma_\ep$ is diagonal, and that there is a constant $C$ so that
\begin{align}
\label{eq-incoherence}
|u_{jk}| \le \frac{C}{\sqrt{p}}
\end{align}
for all $k=1,\dots,r$ and $j=1,\dots,p$. Suppose $Y_1,\dots,Y_n$ are drawn iid from the spiked model. Let $\hat u_1,\dots, \hat u_r$ be the estimated PCs from equation \eqref{emp-pcs}, and let $\hat \U = \sp\{\hat u_1,\dots, \hat u_r\}$ and $\U = \sp\{u_1,\dots,u_r\}$.

Then almost surely in the limit $p/n \to \gamma$
\begin{align}
  \sin^2 \Theta(\hat \U, \U) \le \min\left\{ K \gamma \mu_\ep
    \left( \frac{1}{ \min_k \ell_k}  
            + \frac{ \|\Sigma_\ep\|_{\op}}{\min_k \ell_k^2} \right) ,1 \right\},
\end{align}
where $K$ is a constant depending only on $C$ from \eqref{eq-incoherence}.
\end{thm}
%

\begin{rmk}
Theorem \ref{thm-sin-theta} shows that in the case $r=1$, the estimate $\hat u$ obtained by whitening $Y$, computing the top left singular vector of $Y^{\wh}$, and then unwhitening and normalizing, is asymptotically minimax optimal. When $r > 1$, we require the extra condition \eqref{incoherence-condition} which does not appear in the minimax lower bound from \cite{zhang2018heteroskedastic}.
\end{rmk}

The proof of Theorem \ref{thm-sin-theta} follows from the formula \eqref{prod-unwhite} for the cosines between $u_k$ and $\hat u_k$ from Theorem \ref{thm-main2}. The details are found in Appendix \ref{proofs-pca}.

\subsection{Whitening increases the operator norm SNR}
\label{sec-snr}

In this section, we define a natural signal-to-noise ratio (SNR) for the spiked model, namely the ratio of operator norms between the signal and noise sample covariances. We show that under the generic model from Section \ref{sec-pcs} for the signal principal components $u_k$, the SNR increases after whitening.

We define the SNR by:
\begin{align}
\snr = \frac{\| \hat \Sigma_x \|_{\op}}{\|\hat \Sigma_\ep\|_{\op}}
%
\end{align}
where $\hat \Sigma_x = \frac{1}{n} \sum_{j=1}^n X_j X_j^\top$ and $\hat \Sigma_\ep = \frac{1}{n} \sum_{j=1}^n \ep_j \ep_j^\top$ are the sample covariances of the signal and noise components, respectively (neither of which are observed).

After whitening, the observations change into:
\begin{align}
Y_j^{\wh} = X_j^{\wh} + G_j,
\end{align}
and we define the new SNR to be:
\begin{align}
\snr^{\wh} = \frac{\| \hat \Sigma_x^{\wh} \|_{\op}}{\|\hat \Sigma_g\|_{\op}}
%
\end{align}
where $\hat \Sigma_x^{\wh} = \frac{1}{n} \sum_{j=1}^{n} X_j^{\wh} (X_j^{\wh})^\top$ and $\hat \Sigma_g = \frac{1}{n} \sum_{j=1}^{n} G_j G_j^\top$.

As in Section \ref{sec-pcs}, let $\tau = \lim_{p \to \infty} \tr{\Sigma_{\ep}^{-1}}/p$ (assuming the limit exists), and define $\varphi = \tau \cdot \mu_\ep$. Note that by Jensen's inequality, $\varphi \ge 1$, with strict inequality unless $\Sigma_\ep = \nu I_p$. We will prove the following:

\begin{prop}
\label{prop-snr}
Suppose the population principal components $u_1,\dots,u_r$ are uniformly random orthonormal vectors in $\R^p$. Then in the limit $p/n \to \gamma > 0$,
\begin{align}
\snr^{\wh} \ge \varphi \snr.
\end{align}
\end{prop}
In other words, Proposition \ref{prop-snr} states that for generic signals whitening increases the operator norm SNR by a factor of at least $\varphi \ge 1$. The proof may be found in Appendix \ref{proofs-pca}.

\begin{rmk}
As explained in Remark \ref{rmk-tau}, under the generic model assumed by Proposition \ref{prop-snr}, the notation $\tau$ is consistent with the definition of $\tau_k$ in \eqref{tau-def}.
\end{rmk}

\begin{rmk}
Proposition \ref{prop-snr} is similar in spirit to a result in \cite{liu2016epca}, which essentially shows that the SNR defined by the nuclear norms, rather than operator norms, increases after whitening. However, in the $p \to \infty$ limit, defining the SNR using the ratio of nuclear norms is not as meaningful as using operator norms, because the ratio of nuclear norms always converges to 0 in the high-dimensional limit. Indeed, we have:
\begin{align}
\|\hat \Sigma_x\|_* \to \sum_{k=1}^r \ell_k,
\end{align}
almost surely as $p,n \to \infty$. On the other hand,
\begin{align}
\frac{1}{p}\|\hat \Sigma_\ep\|_*  \to \mu_\ep.
\end{align}
In particular, $\|\hat \Sigma_\ep\|_*$ grows like $p$, whereas $\|\hat \Sigma_x\|_*$ is bounded with $p$. When $p$ is large, therefore, the norm of the noise swamps the norm of the signal. On the other hand, the operator norms of $\hat \Sigma_x$ and $\hat \Sigma_\ep$ are both bounded, and may therefore be comparable in size.
\end{rmk}

\begin{figure}[h]
\centering
\includegraphics[scale=0.5]{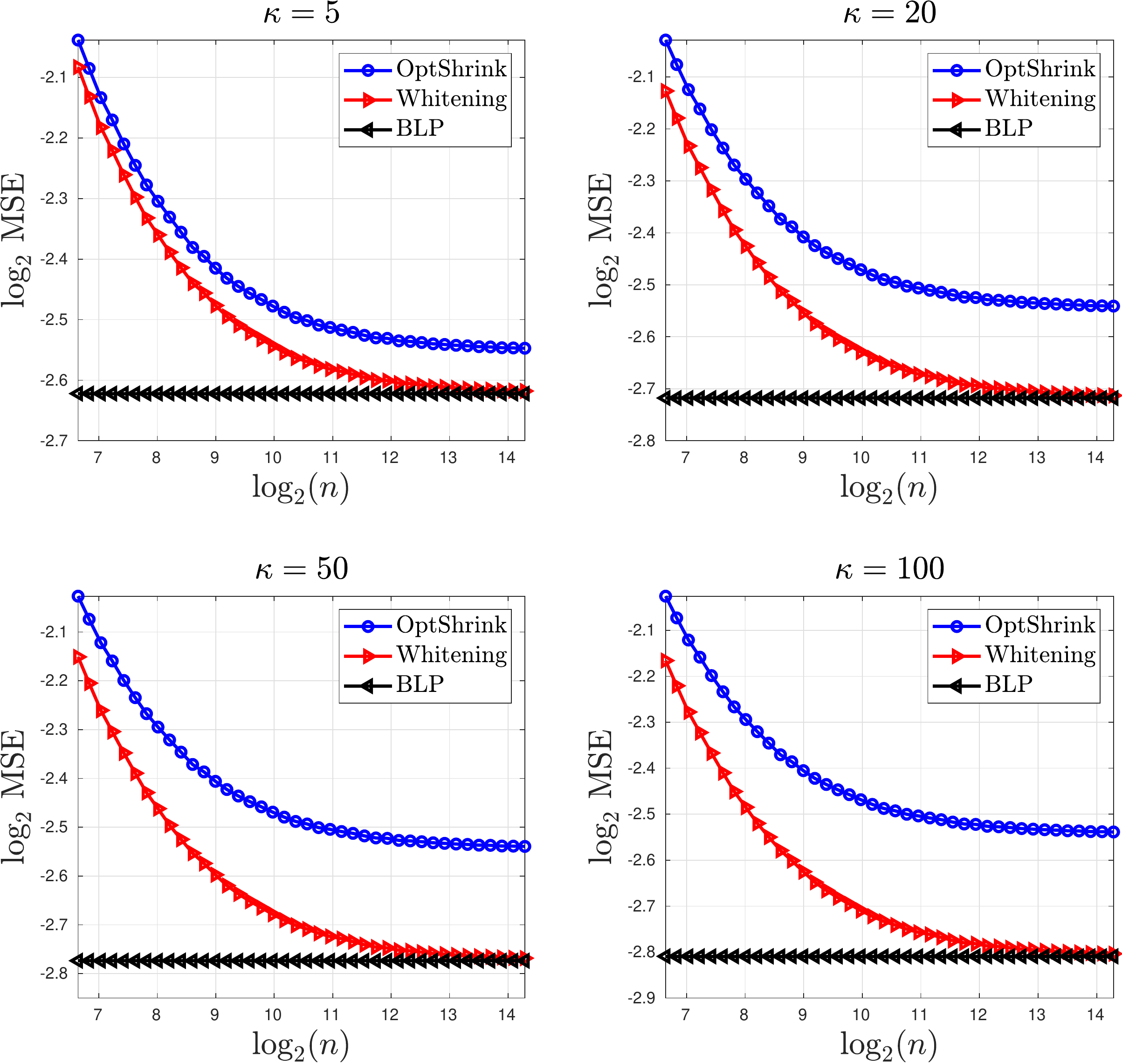}
\caption{
Prediction errors for the optimal whitened shrinker, the optimal unwhitened shrinker (OptShrink), and the best linear predictor (an oracle method).
}
\label{fig-blp-1}
\end{figure}

\section{Numerical results}
\label{sec-numerical}

In this section we report several numerical results that illustrate the performance of our predictor in the spiked model, as well as several beneficial properties of whitening. Code implementing the shrinkage with whitening algorithms will be made available online.

\subsection{Comparison to the best linear predictor}
\label{compare-blp}

In this experiment, we compared our predictor to the best linear predictor (BLP), defined in equation \eqref{eq-blp}. The BLP is an oracle method, as it requires knowledge of the population covariance $\Sigma_x$, which is not accessible to us. However, Theorem \ref{thm-blp} predicts that as $p/n \to 0$, the optimal shrinkage with whitening predictor will behave identically to the BLP.

In the same experiments, we also compare our method to OptShrink \cite{nadakuditi2014optshrink}, the optimal singular value shrinker without any transformation. Theorem \ref{prop-lin-pred} predicts that as $p/n \to 0$, OptShrink will behave identically to a suboptimal linear filter.

In these these tests, we fixed a dimension equal to $p=100$, and let $n$ grow. Each signal was rank 3, with PCs chosen so that the first PC was a completely random unit vector, the second PC was set to zero on the first $p/2$ coordinates and random on the remaining coordinates, and the third PC was completely random on the first $p/2$ coordinates and zero on the remaining coordinates. The signal random variables $z_{jk}$ were chosen to be Gaussian.

The noise covariance matrix $\Sigma_\ep$ was generated by taking equally spaced values between $1$ and a specified condition number $\kappa > 1$, and then normalizing the resulting vector of eigenvalues to be a unit vector. This normalization was done so that in each test, the total energy of the noise remained constant.

Figure \ref{fig-blp-1} plots the average prediction errors as a function of $n$ for the three methods, for different condition numbers $\kappa$ of the noise covariance $\Sigma_\ep$. The errors are averaged over 500 runs of the experiment, with different draws of signal and noise. As expected, the errors for optimal shrinkage with whitening converge to those of the oracle BLP, while the errors for OptShrink appear to converge to a larger value, namely the error of the limiting suboptimal linear filter.

\begin{rmk}
\label{rmk-comparison}
Unlike shrinkage with whitening, OptShrink does not make use of the noise covariance. Though access to the noise covariance would permit faster evaluation of the OptShrink algorithm using, for instance, the methods described in \cite{leeb2020rapid}, we have found that this does not change the estimation accuracy of the method. Similarly, the BLP uses the true PCs of $X_j$, which are not used by either shrinkage method. The comparison between the methods must be understood in that context.
\end{rmk}

\begin{figure*}
\center
%
%
\begin{subfigure}[t]{.35\textwidth} 
\includegraphics[width=\textwidth]{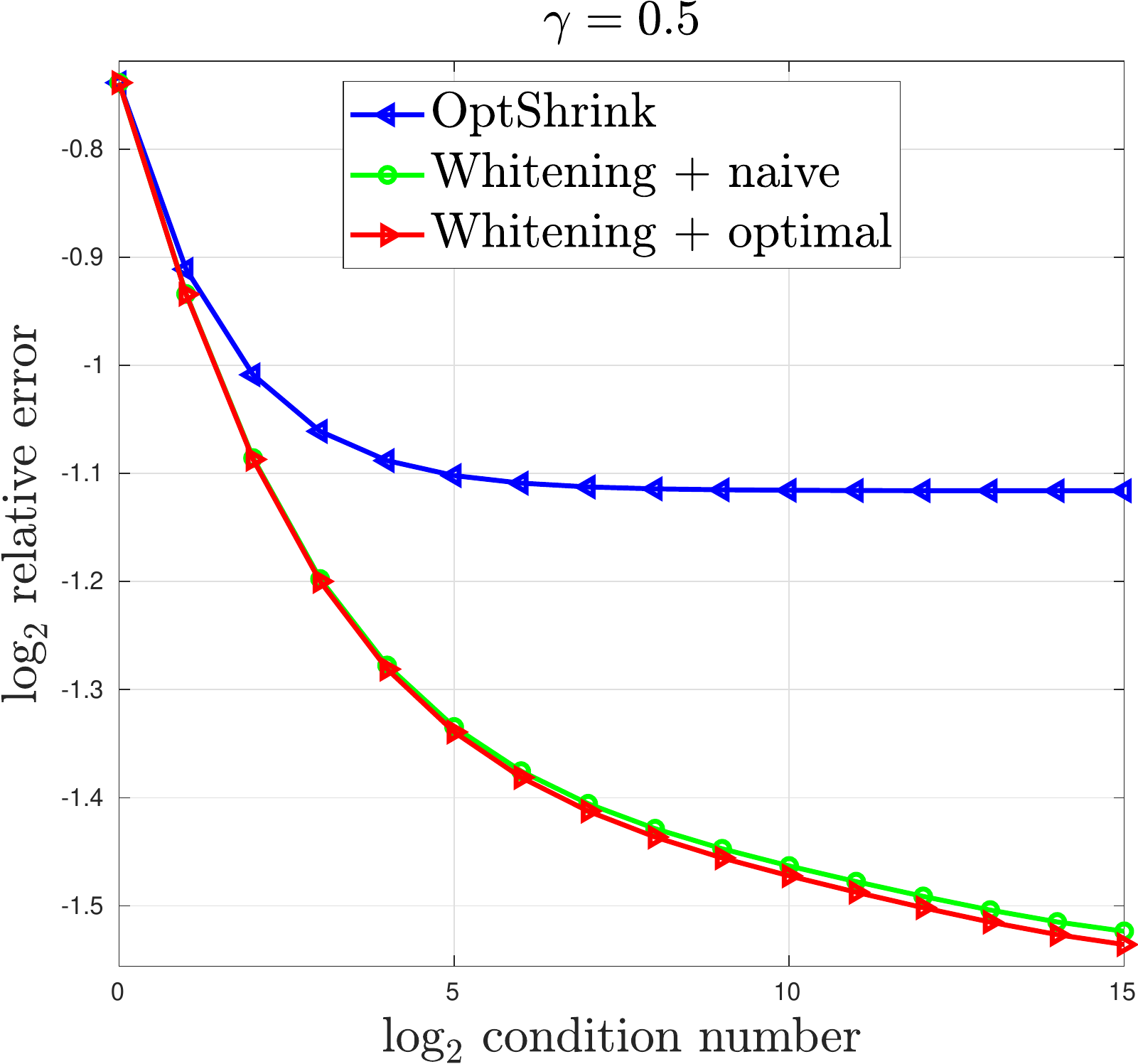}
\end{subfigure}
%
%
\begin{subfigure}[t]{.35\textwidth} 
\includegraphics[width=\textwidth]{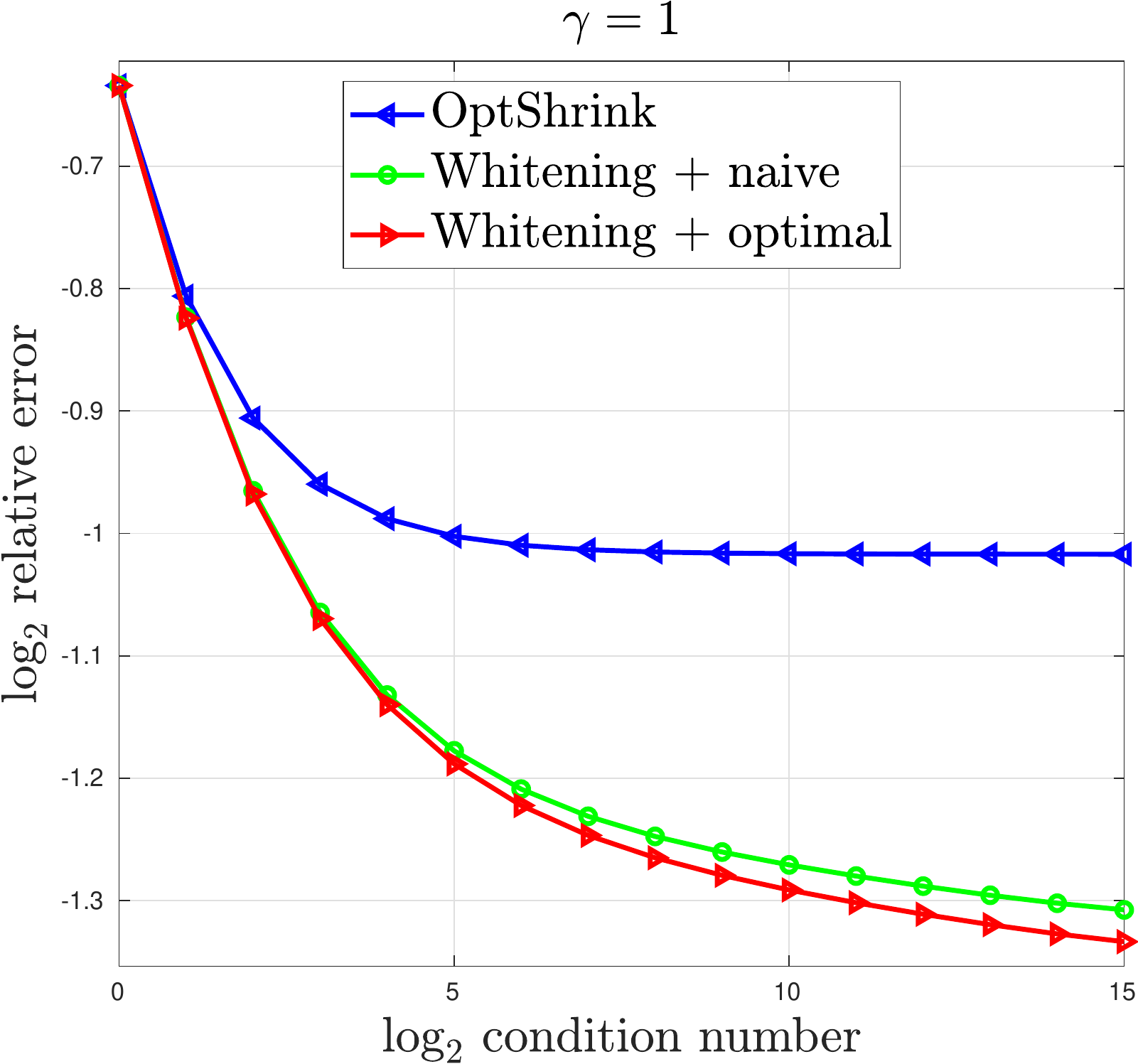}
\end{subfigure}
%
%
\begin{subfigure}[t]{.35\textwidth} 
\includegraphics[width=\textwidth]{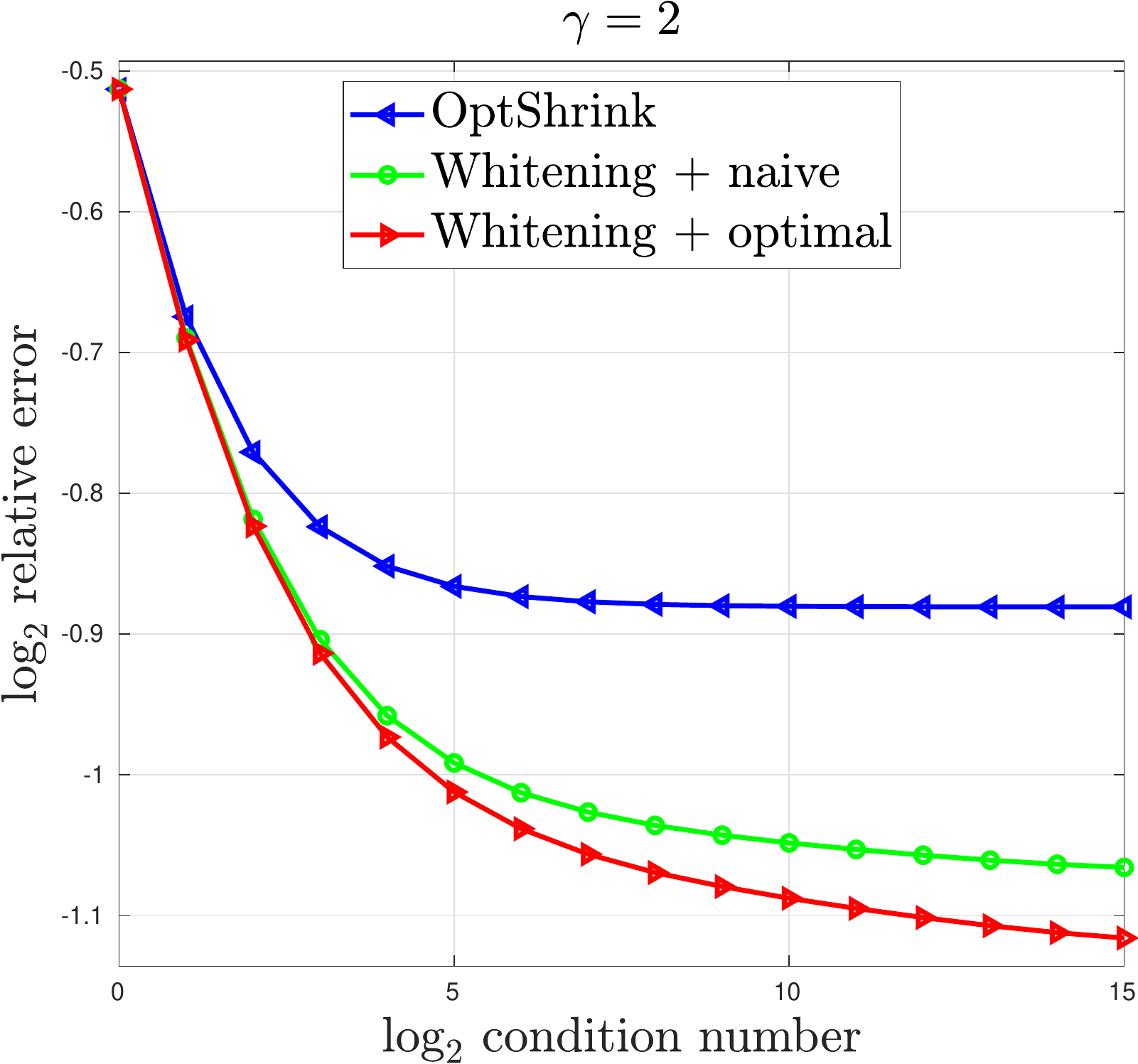}
\end{subfigure}
%
%
\begin{subfigure}[t]{.35\textwidth} 
\includegraphics[width=\textwidth]{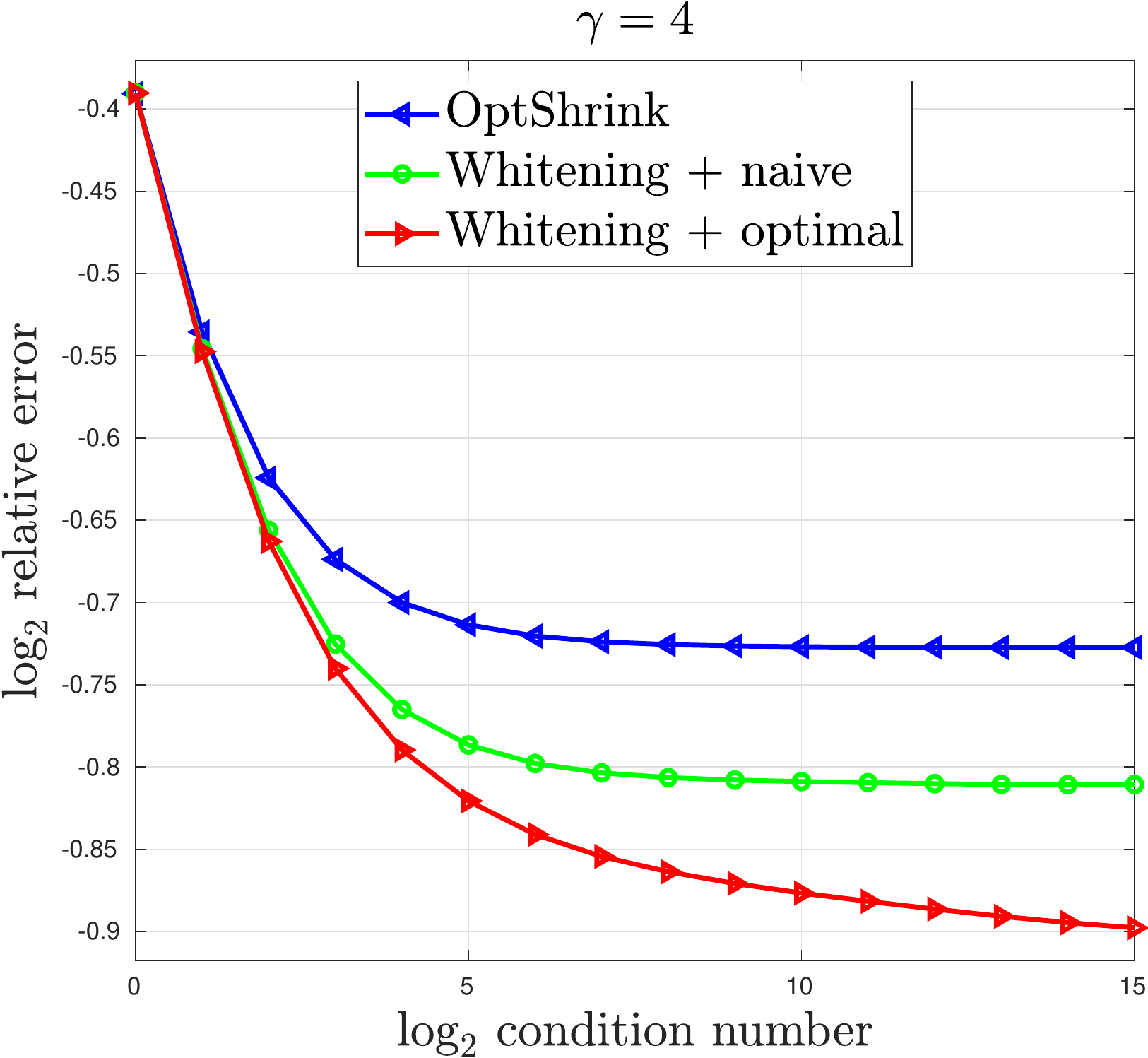}
\end{subfigure}
\caption{
Comparison of whitening with optimal shrinkage; whitening with naive shrinkage; and OptShrink (no whitening), as a function of the noise covariance matrix's condition number $\kappa$.
}
\label{fig-comparison}
\end{figure*}

%

\subsection{Performance of singular value shrinkage}
\label{compare-optshrink}

We examine the performance of optimal shrinkage with whitening for different values of $\gamma$ and different condition numbers of the noise covariance. We compare to OptShrink \cite{nadakuditi2014optshrink} and the naive shrinker with whitening employed in \cite{dobriban2017optimal}, which uses singular values $\sqrt{\ell_k^{\wh}} c_k^{\wh} \tilde c_k^{\wh}$; see Figures \ref{fig-shrinkers} and \ref{fig-shrinkers2} and the associated text. This latter shrinker does not account for the change in angle between the singular vectors resulting from unwhitening.

In each run of the experiment, we fix the dimension $p=1000$. We use a diagonal noise covariance with a specified condition number $\kappa$, whose entries are linearly spaced between $1/\kappa$ and $1$, and increase with the index. We generate the orthonormal basis of PCs $u_1$, $u_2$, $u_3$ from the model described in Section \ref{sec-randompcs}, as follows: $u_1$ is a unifomly random unit vector; $u_2$ has Gaussian entries with linearly-spaced variances $a_1,\dots,a_p$, where $a_p < a_{p-1} < \dots < a_1$, $\sum_{i=1}^{p}a_i = 1$, and $a_1 / a_p = 10$; and $u_3$ has Gaussian entries with linearly-spaced variances $b_1,\dots,b_p$, where $b_1 < b_2 < \dots < b_p$, $\sum_{i=1}^{p}b_i = 1$, and $b_p / b_1 = 10$. Gram-Schmidt is then performed on $u_1$, $u_2$, and $u_3$ to ensure they are orthonormal. For aspect ratio $\gamma$, the three signal singular values are $\gamma^{1/4} + i / 2$, $i=1,2,3$.

For different values of $n$, and hence of $\gamma$, we generate 50 draws of the data and record the average relative errors for each of the three methods. The results are plotted in Figure \ref{fig-comparison}. As is apparent from the figures, both whitening methods typically outperform OptShrink. Furthermore, when $n$ is large, both optimal shrinkage and naive shrinkage perform very similarly; this makes sense because both methods converge to the BLP as $n \to \infty$. By contrast, when $\gamma$ is large, the benefits of using the optimal shrinker over the naive shrinker are more apparent.

\begin{rmk}
As noted in Remark \ref{rmk-comparison}, we emphasize that unlike both whitening methods, OptShrink does not make use of the noise covariance, and the comparison between the methods must be understood in that context.
\end{rmk}

\subsection{Performance of eigenvalue shrinkage}

We examine the performance of optimal eigenvalue shrinkage with whitening for different values of $\gamma$ and different condition numbers of the noise covariance. We use nuclear norm loss, for which the optimal $\tilde t_k^2$ in Algorithm \ref{alg:cov-est} is given by the formula
\begin{align}
\tilde t_k^2 = \max\{\ell_k (2 c_k^2 - 1),0\}.
\end{align}
This formula is derived in \cite{donoho2018optimal}.

We compare to two other methods. We consider optimal eigenvalue shrinkage without whitening, where the population eigenvalues and cosines between observed and population eigenvectors are estimated using the methods from \cite{nadakuditi2014optshrink}. We also consider the whitening and eigenvalue shrinkage procedure from \cite{liu2016epca}, which shrinks the eigenvalues to the population values $\ell_k$; this is an optimal procedure for operator norm loss \cite{donoho2018optimal}, but suboptimal for nuclear norm loss.

As in Section \ref{compare-optshrink}, in each run of the experiment, we fix the dimension $p=1000$. We use a diagonal noise covariance with a specified condition number $\kappa$, whose entries are linearly spaced between $1/\kappa$ and $1$, and increase with the index. We generate the orthonormal basis of PCs $u_1$, $u_2$, $u_3$ from the model described in Section \ref{sec-randompcs}, as follows: $u_1$ is a unifomly random unit vector; $u_2$ has Gaussian entries with linearly-spaced variances $a_1,\dots,a_p$, where $a_p < a_{p-1} < \dots < a_1$, $\sum_{i=1}^{p}a_i = 1$, and $a_1 / a_p = 10$; and $u_3$ has Gaussian entries with linearly-spaced variances $b_1,\dots,b_p$, where $b_1 < b_2 < \dots < b_p$, $\sum_{i=1}^{p}b_i = 1$, and $b_p / b_1 = 10$. Gram-Schmidt is then performed on $u_1$, $u_2$, and $u_3$ to ensure they are orthonormal. For aspect ratio $\gamma$, the three signal singular values are $\gamma^{1/4} + i$, $i=1,2,3$.

For different values of $n$, and hence of $\gamma$, we generate 50 draws of the data and record the average relative errors $\| \hat \Sigma_x - \Sigma_x\|_* / \|\Sigma_x\|_*$ for each of the three methods. The results are plotted in Figure \ref{fig-comparison-cov}. As is apparent from the figures, optimal shrinkage with whitening outperforms the other two methods. For the smaller values of $\gamma$, optimal shrinkage without whitening outperforms the population shrinker with whitening when the condition number $\kappa$ is small, since the benefits of whitening are not large; however, as $\kappa$ grows, whitening with the suboptimal population shrinker begins to outperform. For larger $\gamma$, the cost of using the wrong shrinker outweigh the benefits of whitening, and the population shrinker with whitening is inferior to both other methods. This illustrates the importance of using a shrinker designed for the intended loss function.

\begin{figure*}
\center
%
%
\begin{subfigure}[t]{.35\textwidth} 
\includegraphics[width=\textwidth]{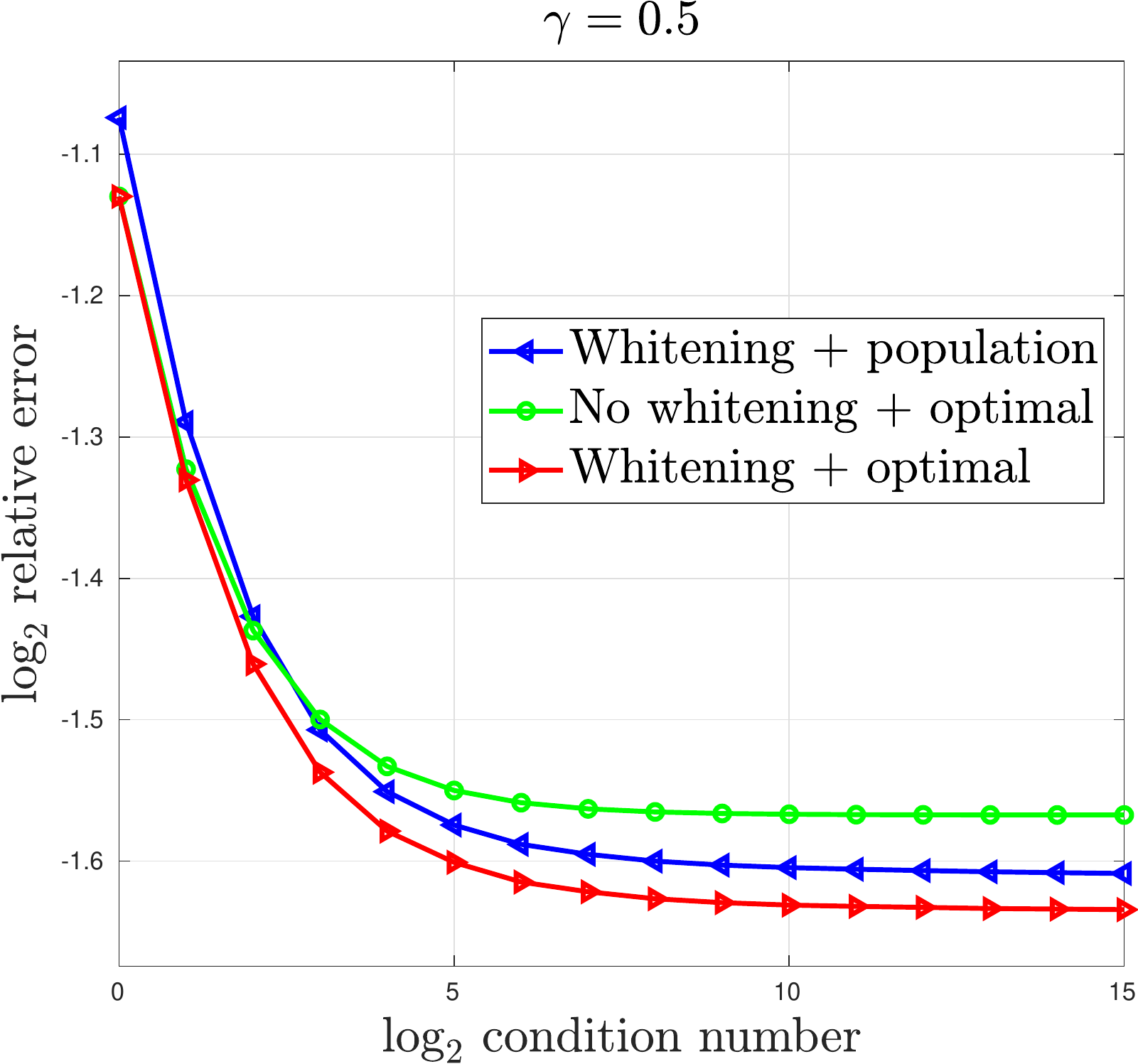}
\end{subfigure}
%
%
\begin{subfigure}[t]{.35\textwidth} 
\includegraphics[width=\textwidth]{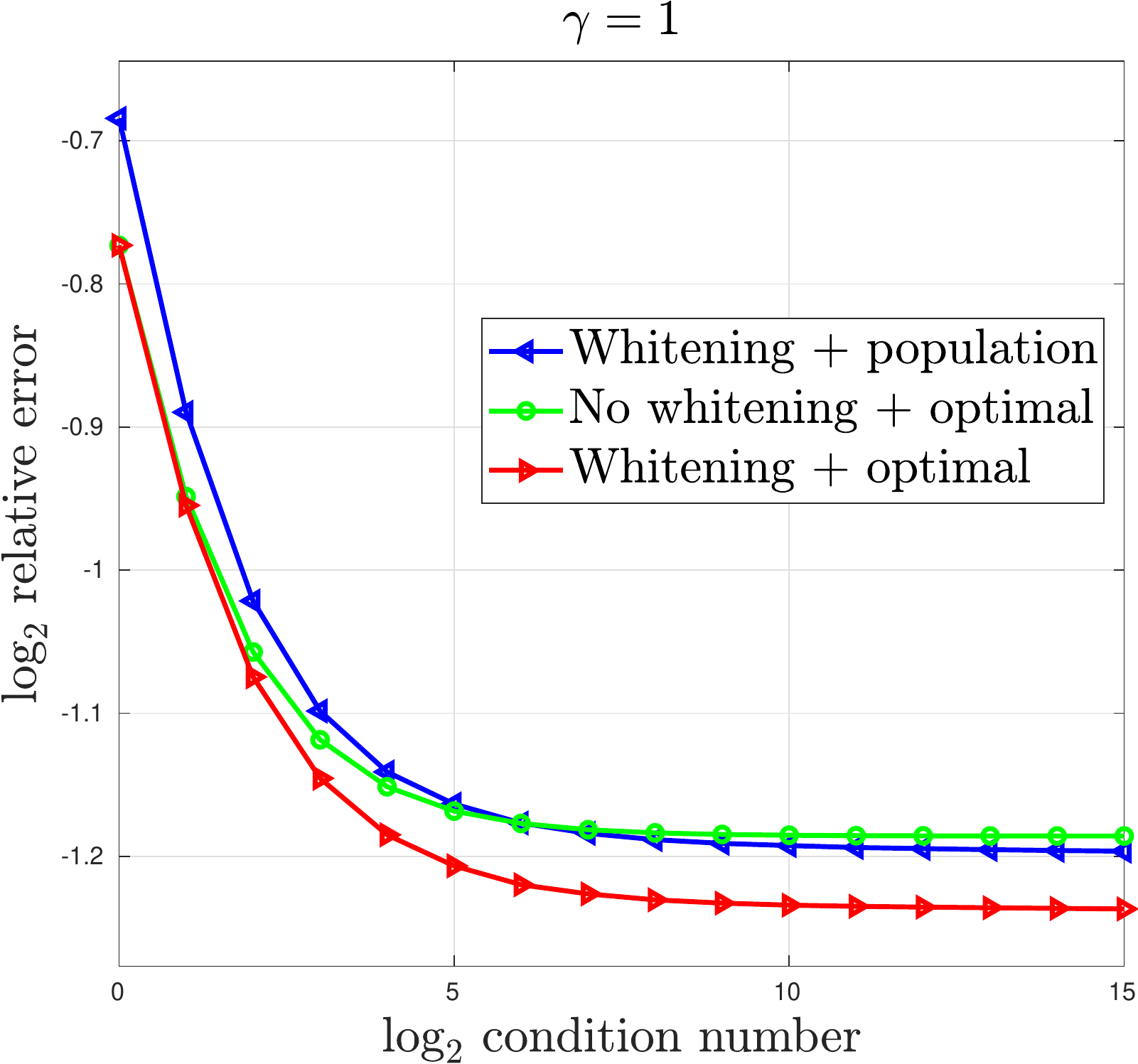}
\end{subfigure}
%
%
\begin{subfigure}[t]{.35\textwidth} 
\includegraphics[width=\textwidth]{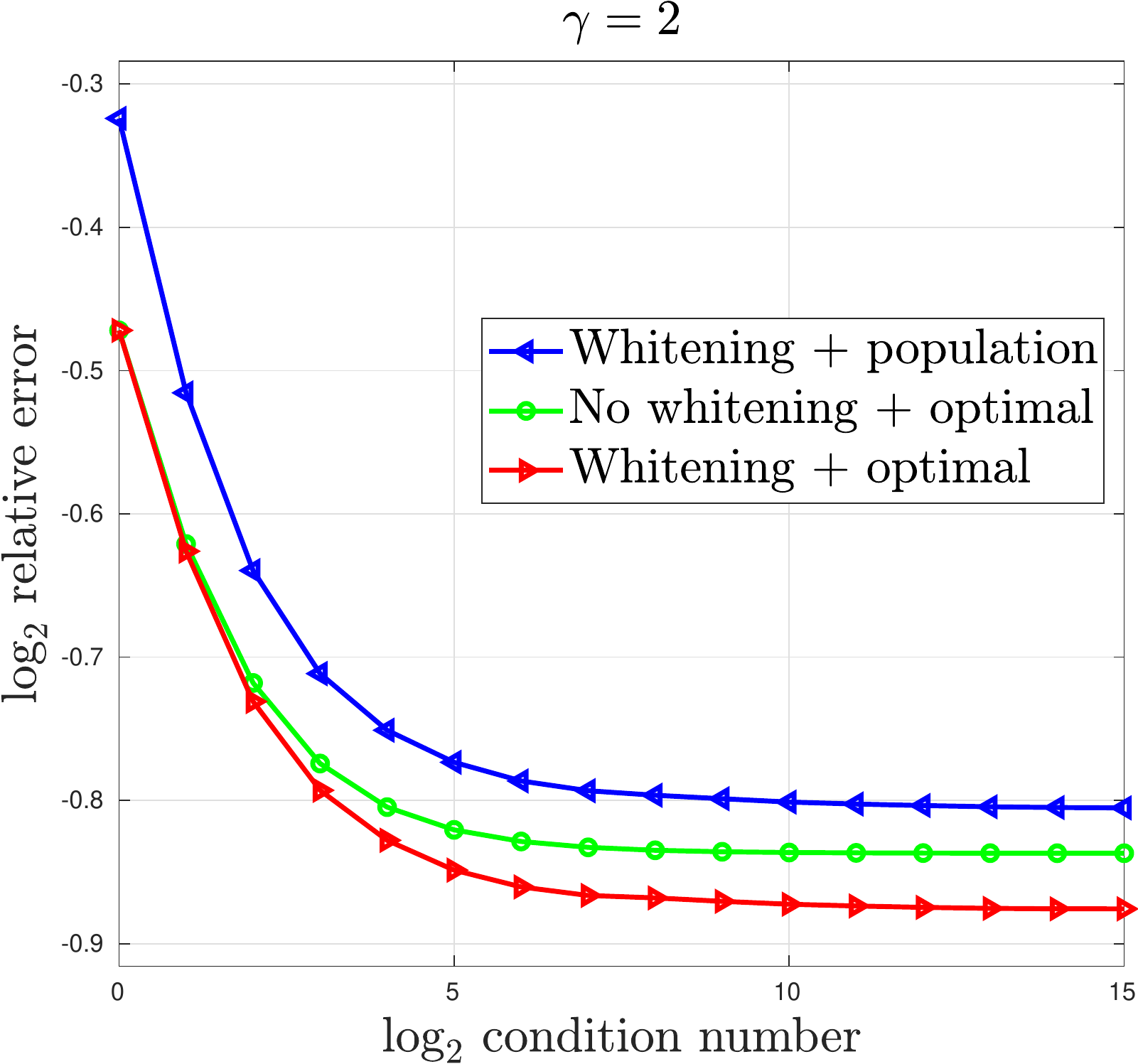}
\end{subfigure}
%
%
\begin{subfigure}[t]{.35\textwidth} 
\includegraphics[width=\textwidth]{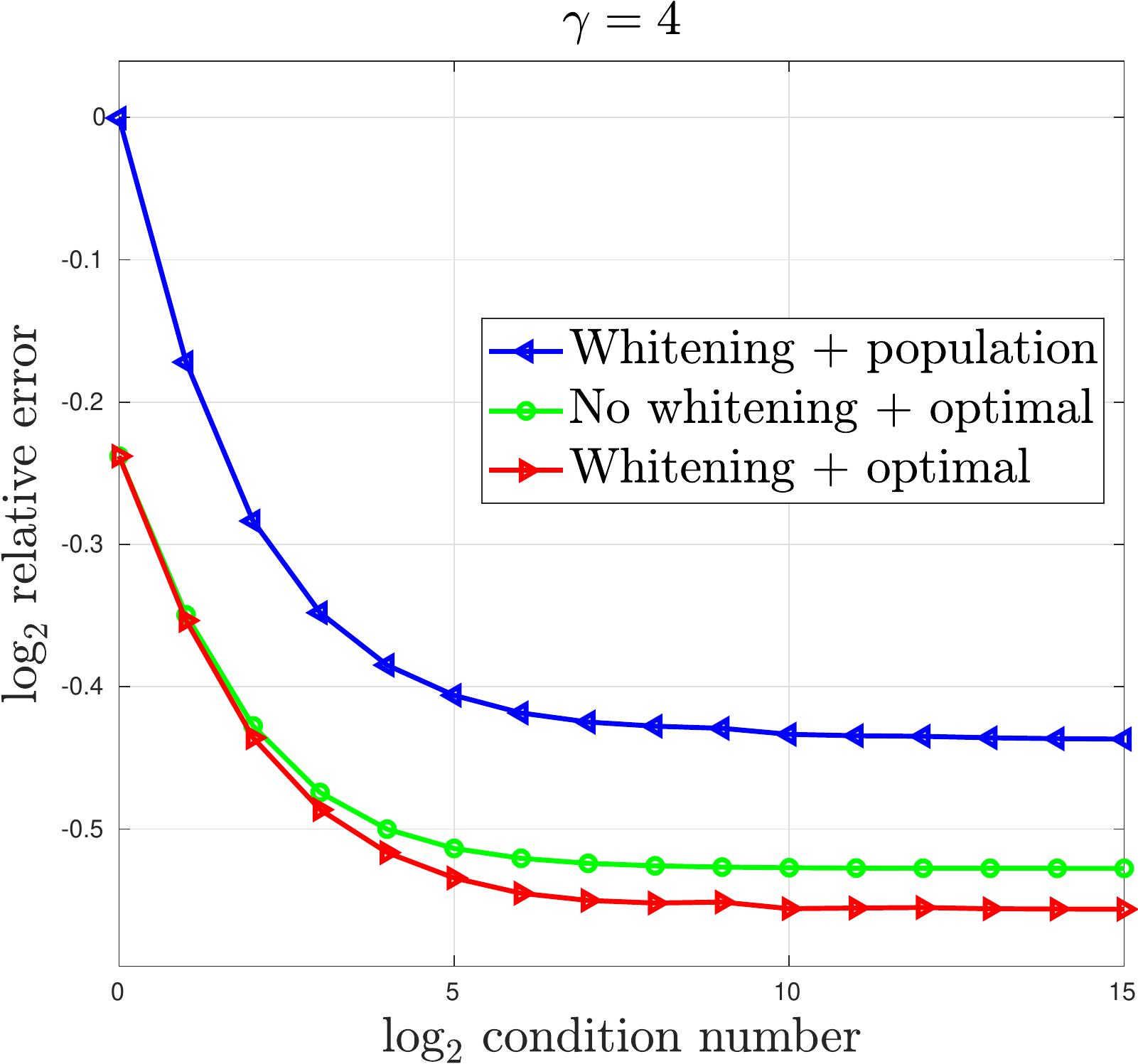}
\end{subfigure}
\caption{
Comparison of whitening with optimal shrinkage; whitening with naive shrinkage; and OptShrink (no whitening), as a function of the noise covariance matrix's condition number $\kappa$.
}
\label{fig-comparison-cov}
\end{figure*}

%
%
\begin{figure}
\center
\includegraphics[scale=0.4]{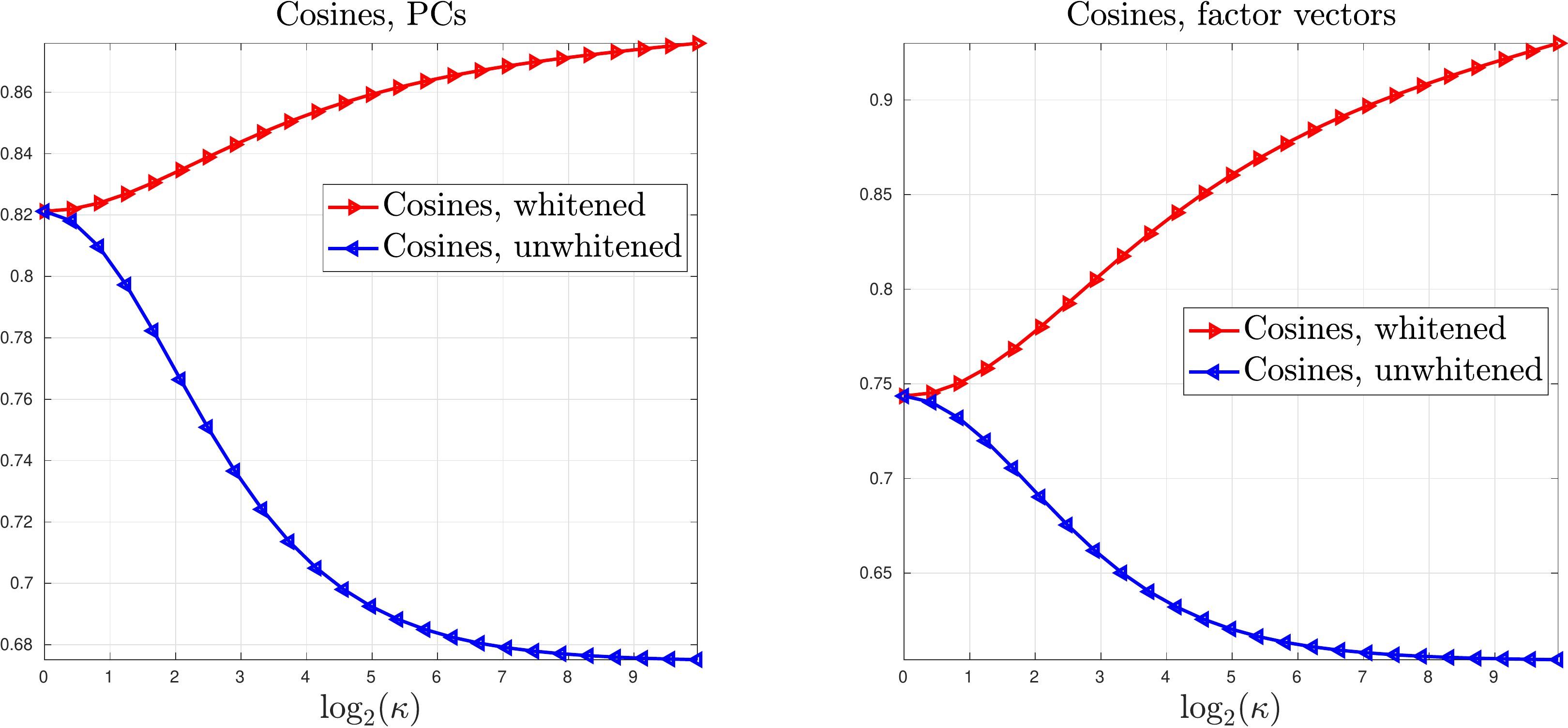}
\caption{
Comparison of the cosines between the empirical and population singular vectors, for the raw data and the whitened data, as a function of the noise covariance matrix's condition number $\kappa$.
}
\label{fig-cosines}
\end{figure}

\subsection{Numerical comparison of the angles}


In this section, we numerically illustrate Theorem \ref{thm-angles} by examining the angles between the spanning vectors $\hat u_k$ (the empirical PCs) and $\hat v_k$ of $\hat X$ and, respectively, the population vectors $u_k$ (the population PCs) and $v_k$. We show that these angles are smaller (or equivalently, their cosines are larger) than the corresponding angles between the population $u_k$ and $v_k$ and the singular vectors of the unwhitened data matrix $Y$.

Figure \ref{fig-cosines} plots the cosines as a function of the condition number $\kappa$ of the noise matrix $\Sigma_\ep$. In this experiment, we consider a rank 1 signal model for simplicity, with a uniformly random PC. We used dimension $p=500$, and drew $n=1000$ observations. For each condition number $\kappa$ of $\Sigma_\ep$, we generate $\Sigma_\ep$ as described in Section \ref{compare-blp}. For each test, we average the cosines over 50 runs of the experiment (drawing new signals and new noise each time). Both signal and noise are Gaussian. As we see, the cosines improve dramatically after whitening. As $\kappa$ grows, i.e., the noise becomes more heteroscedastic, the improvement becomes more pronounced.



%

\begin{figure}[h]
\center
\includegraphics[scale=0.4]{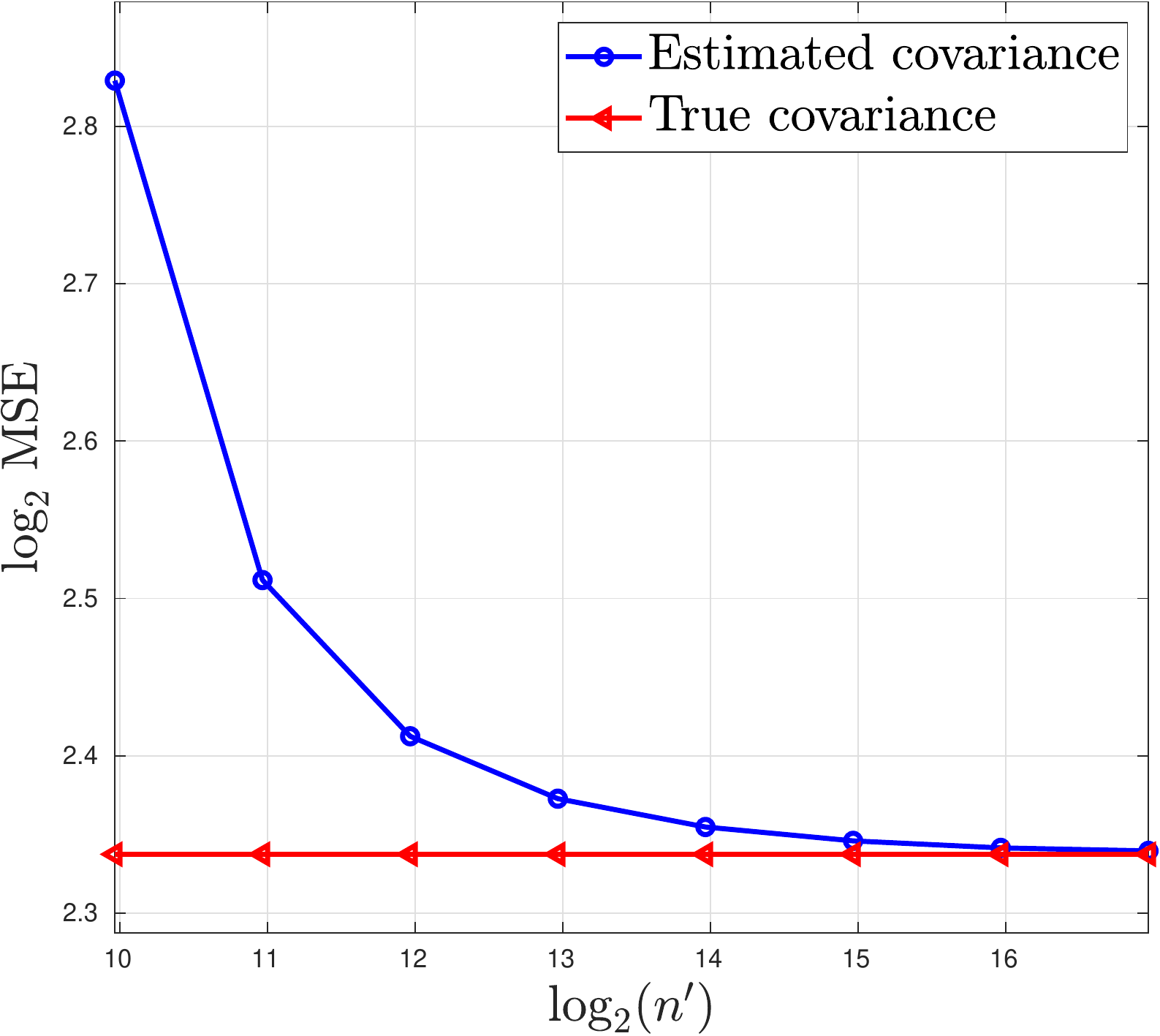}
\caption{Comparison of the errors when using the true noise covariance $\Sigma_\ep$ and the sample noise covariance $\hat \Sigma_\ep$ estimated from $n'$ samples.}
\label{fig-estcov}
\end{figure}

\subsection{Estimating the noise covariance}
\label{sec:sample-covariance}

In many applications, the true noise covariance may not be accessible. In this experiment, we consider the effect of estimating the noise covariance by the sample covariance from $n'$ iid samples of pure noise, $\ep_1,\dots,\ep_{n'}$, as $n'$ grows.

We fix the dimension $p=500$ and number of signal-plus-noise observations $n=625$, and $r=2$ signal singular values $3$ and $5$. We take the noise covariance to have condition number $\kappa = 500$, with eigenvalues equispaced between $1/100$ and $1/5$. The eigenvectors of the noise covariance are drawn uniformly at random.

For increasing values of $n' \ge p$, we draw $n'$ iid realizations of the noise $\ep_1,\dots,\ep_{n'}$, and form the sample covariance:
\begin{align}
\hat \Sigma_{\ep} = \frac{1}{n'} \sum_{i=1}^{n'} \ep_i \ep_i^\top.
\end{align}
For each $n'$, we perform Algorithm \ref{alg:homshrink} using the sample covariance $\hat \Sigma_{\ep}$. The experiment is repeated 2000 times for each value of $n'$, and the errors averaged over these 2000 runs. Figure \ref{fig-estcov} plots the average error as a function of $n'$. We also apply Algorithm \ref{alg:homshrink} using the true noise covariance $\Sigma_\ep$, and plot the average error (which does not depend on $n'$) in Figure \ref{fig-estcov} as well. The error when using the estimated covariance converges to the error when using the true covariance, indicating that Algorithm 1 is robust to estimation of the covariance.

\subsection{Accuracy of error formulas and estimates}

In this experiment, we test the accuracy of the error formula \eqref{eq-amse}. There are three distinct quantities that we define. The first is the oracle AMSE, which we define from the known population parameters. The second is the estimated AMSE, which we will denote by $\widehat \AMSE$; this is estimated using the observations $Y_1,\dots,Y_n$ themselves. The third is the mean-squared error itself, $\|\hat X - X\|_{\Fr}^2 / n$. Of the three quantities, only $\widehat \AMSE$ would be directly observed in practice. We define the discrepancy between $\AMSE$ and $\|\hat X - X\|_{\Fr}^2/n$ as $|\AMSE - \|\hat X - X\|_{\Fr}^2/n|$, and the discrepancy between $\widehat \AMSE$ and $\|\hat X - X\|_{\Fr}^2$ as $|\widehat \AMSE - \|\hat X - X\|_{\Fr}^2/n|$.

Figure \ref{fig-discrepancies} plots the log discrepancies against $\log_2(p)$. We also include a table of the values themselves. In all experiments, we use the following parameters: the aspect ratio is $\gamma=0.8$, the rank $r=2$, the signal singular values are $3$ and $2$, $u_1$ is $\sqrt{2/p}$ on entries $1,\dots,p/2$ and 0 elsewhere, $u_2$ is $\sqrt{2/p}$ on entries $p/2+1,\dots,p$ and 0 elsewhere, and the noise covariance is diagonal with variances linearly spaced from $1/200$ to $3/2$, increasing with the coordinates.

\begin{figure}[h]
\centering
\includegraphics[scale=0.4]{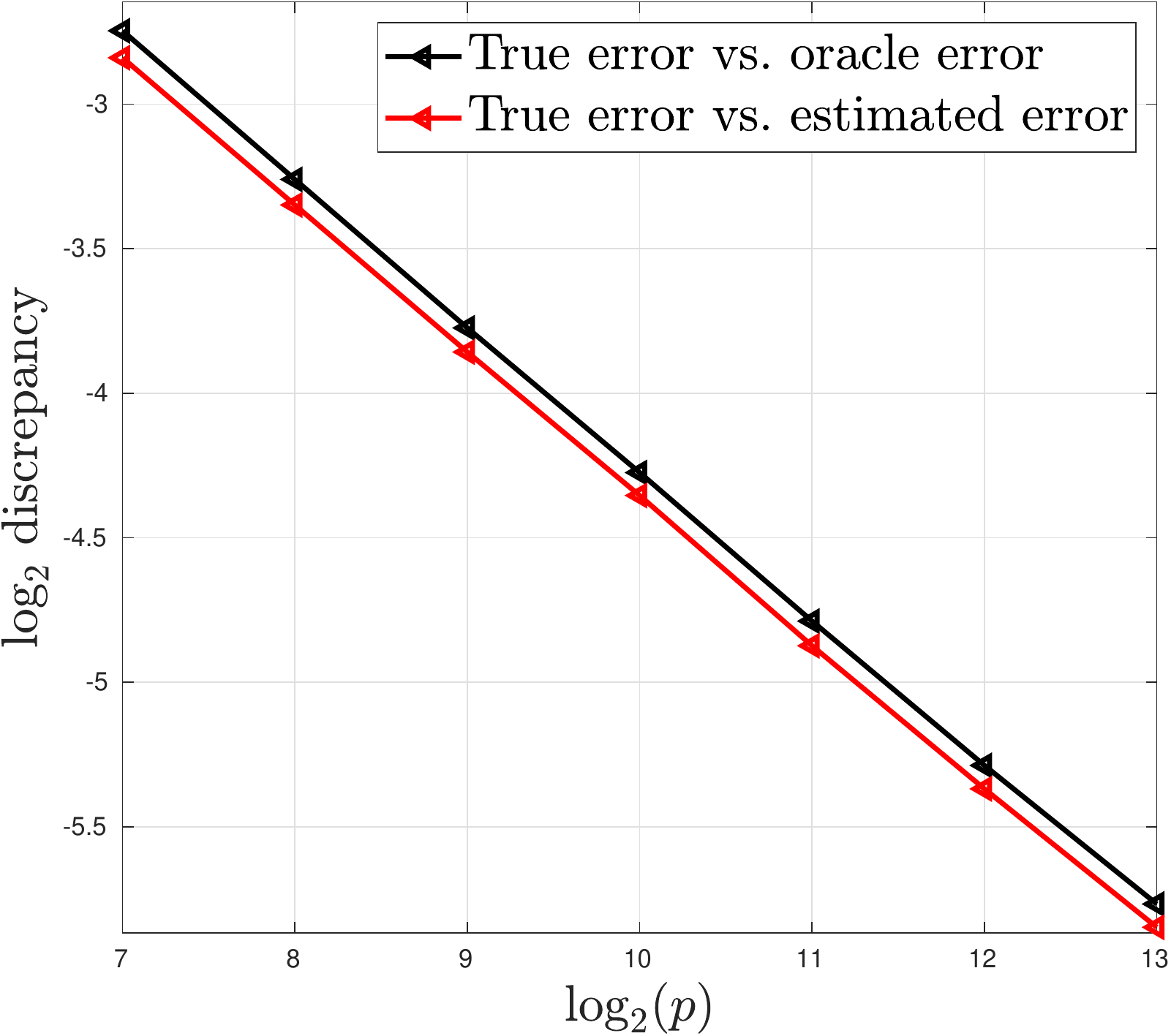}
\caption{Logarithm of the discrepancies $|\AMSE - \|\hat X - X\|_{\Fr}^2/n|$ and  $|\widehat \AMSE - \|\hat X - X\|_{\Fr}^2/n|$, versus $\log_2(p)$. $\AMSE$ is the oracle value of the error, and $\widehat \AMSE$ is estimated from the data itself.}
\label{fig-discrepancies}
\end{figure}

We make two observations. First, the slope of each plot is approximately $0.5$, indicating that the error formulas derived are accurate with error $O(n^{-1/2})$. This is precisely the rate we expect from \cite{bao2018singular}. Second, the discrepancies of $\AMSE$ and $\widehat \AMSE$ are very close, and in fact the discrepancy of $\widehat \AMSE$ is slightly smaller than that of $\AMSE$. This indicates that the observed $\widehat \AMSE$ provides a viable estimate for the actual error $\|\hat X - X\|^2_{\Fr}/n$.

%
%
\begin{table}
\centering
\begin{tabular}{|l| c  | c | }  
\hline  
 $\log_2(p)$ &  Discrepancy, $\AMSE$ & Discrepancy, $\widehat \AMSE$ \\
\hline  
    7 & 1.49e-01 & 1.40e-01  \\
    8 & 1.04e-01 & 9.82e-02  \\
    9 & 7.31e-02 & 6.90e-02  \\
   10 & 5.17e-02 & 4.89e-02  \\
   11 & 3.62e-02 & 3.41e-02  \\
   12 & 2.56e-02 & 2.42e-02  \\
   13 & 1.84e-02 & 1.74e-02  \\
\hline
\end{tabular}  
\caption{Discrepancies $|\AMSE - \|\hat X - X\|_{\Fr}^2/n|$ and  $|\widehat \AMSE - \|\hat X - X\|_{\Fr}^2/n|$. $\AMSE$ is the oracle value of the error, and $\widehat \AMSE$ is estimated from the data itself.}
\end{table}

\subsection{Comparing in-sample and out-of-sample prediction}

In this next experiment, we compare the performance of in-sample and out-of-sample prediction, as described in Section \ref{sec-oos}. Optimal in-sample prediction is identical to performing optimal singular value shrinkage with noise whitening to the in-sample data $Y_1,\dots,Y_n$. For out-of-sample prediction, we use the expression of the form \eqref{oos00} with the optimal coefficients $\eta_k^{\out}$ from Proposition \ref{prop-oos-summary}.

We ran the following experiments. For a fixed dimension $p$, we generated a random value of  $n > p$. We then chose three random PCs from the same model described in Section \ref{compare-blp}, and we generated pools of $n$ in-sample and out-of-sample observations. We performed optimal shrinkage with whitening on the in-sample observations, and applied the out-of-sample prediction to the out-of-sample data using the vectors $\hat u_k^{\wh}$ computed from the in-sample data. We then computed the MSEs for the in-sample and out-of-sample data matrices. This whole procedure was repeated 2000 times.

\begin{figure}[h]
\centering
\includegraphics[scale=0.4]{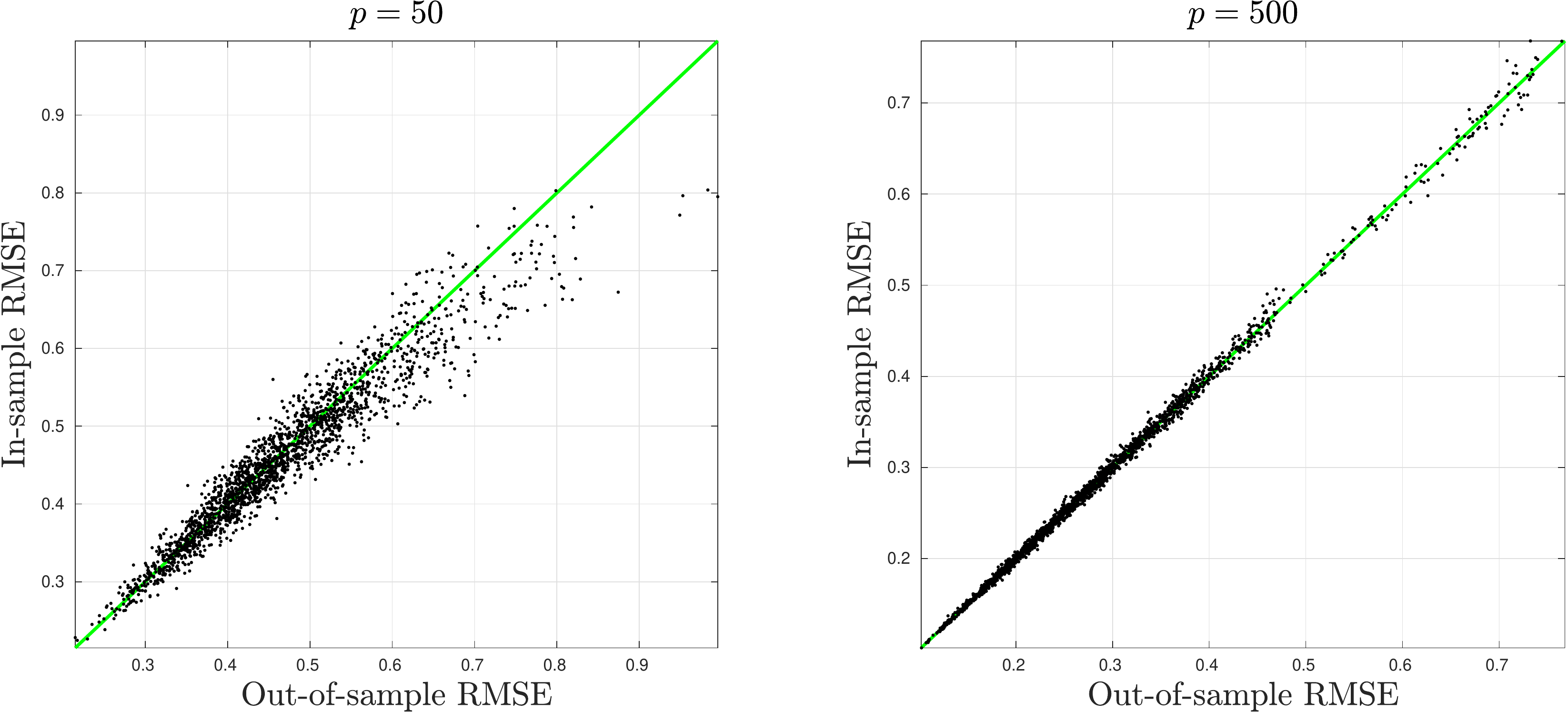}
\caption{
Comparison of in-sample and out-of-sample denoising for $p = 50$ and $p = 500$.
}
\label{fig-oos}
\end{figure}

Figure \ref{fig-oos} shows scatterplots of the in-sample and out-of-sample predictions for $p=50$ and $p=500$. In both plots, we see that there is not a substantial difference between the in-sample and out-of-sample prediction errors, validating the asymptotic prediction made by Proposition \ref{prop-oos-summary}. Even for the low-dimension of $p=50$, there is very close agreement between the performances, and for $p=500$ they perform nearly identically.

\subsection{Signal detection and rank estimation}
\label{sec-numerical-rank}

In this experiment, we show that whitening improves signal detection. We generated data from a rank 1 model, with a weak signal. We computed all the singular values of the original data matrix $Y$, and the whitened matrix $Y^{\wh}$. Figure \ref{fig-histograms} plots the the top 20 singular values for each matrix.

\begin{figure}[h]
\centering
\includegraphics[scale=0.4]{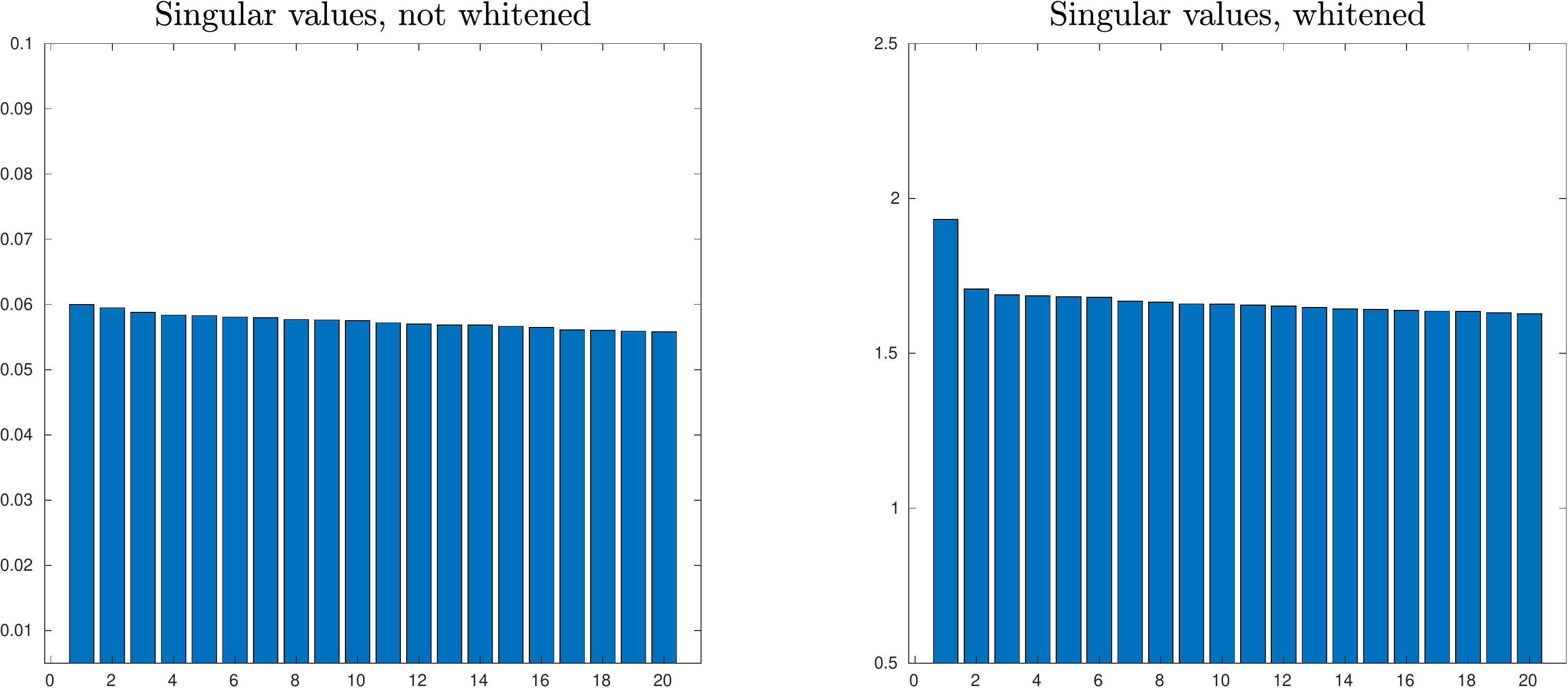}
\caption{
The top 20 empirical singular values of the raw data matrix $Y$ and the whitened data matrix $Y^{\wh}$, for a rank 1 signal.
}
\label{fig-histograms}
\end{figure}

It is apparent from the comparison of these figures that the top singular value of the whitened matrix pops out from the bulk of noise singular values, making detection of the signal component very easy in this case. By contrast, the top singular value of the raw, unwhitened matrix $Y^{\wh}$ does not stick out from the bulk. Proposition \ref{prop-snr} would lead us to expect this type of behavior, since the signal matrix increases in strength relative to the noise matrix.

\subsection{Non-gaussian noise}

The theory we have derived relies on the orthogonal invariance of the noise matrix $G$. In this experiment, we study the agreement between the theoretically predicted values for $c_k$ and $\tilde c_k$ and the observed values for finite $n$ and $p$ and non-Gaussian noise.

For different values of $n$ we generated rank 1 signal matrices of size $n/2$-by-$n$, with top PC $u$ having all entries equal to $1/\sqrt{1000}$, $z_{j}$ Gaussian, and signal energy $\ell = 1$. We generated a noise matrix, where each entry has mean 0 and variance 1, drawn iid from a specified distribution. We then colored the noise matrix by multiplying it by $\Sigma_\ep^{1/2} = \diag(\sqrt{\nu_1},\dots,\sqrt{\nu_p})$, where $\nu_1,\dots,\nu_p$ are linearly spaced, $\nu_1 = 1/500$, and $\nu_p = 1$.

We considered four different distributions for the entries of $G$: the Gaussian distribution; the Rademacher distribution; and the Student t distributions with $10$ and $3$ degrees of freedom (normalized to have variance 1). For each distribution, we drew signal/noise pairs, and computed the absolute value of the cosines between the topmost left and right singular vectors of the observed matrix and the left and right singular vectors of the signal matrix. We then computed the average absolute difference (the discrepancy) between the observed cosines and the theoretically predicted values $c$ and $\tilde c$ from Section \ref{sec-asymptotics}. The errors are averaged over 20000 runs.

Table \ref{table-outer} contains the average discrepancies for $c$, and Table \ref{table-inner} contains the average errors for $\tilde c$, both for $n=1000,2000,4000, 8000$. For the t distribution with 10 degrees of freedom and the Rademacher distribution, the discrepancies match those of the Gaussian to within the precision of the experiment. In particular, for these three noise distributions, the observed cosines appear to converge to the predicted asymptotic values at a rate of roughly $O(n^{-1/2})$. By contrast, for the t distribution with only $3$ degrees of freedom, there is substantial discrepancy between the theoretical and observed cosines, and the discrepancies do not decrease with $n$ (in fact, they grow).

These numerical results suggest that for noise distributions with sufficiently many finite moments, the distributions are approximately equal as those Gaussian noise, which in turn suggests that the limiting cosine values we have derived for Gaussian noise may hold for more general distributions.

\begin{table} 
\centering 
\begin{tabular}{|l| c  | c | c | c|}  
\hline  
 $n$&  Gaussian & Rademacher & t, df=10 & t, df=3\\
\hline  
 1000 & 8.173e-03 & 8.009e-03 & 8.147e-03 & 2.584e-01 \\
 2000 & 5.742e-03 & 5.794e-03 & 5.750e-03 & 3.610e-01 \\
 4000 & 4.069e-03 & 4.073e-03 & 4.071e-03 & 4.730e-01 \\
 8000 & 2.896e-03 & 2.933e-03 & 2.897e-03 & 5.866e-01 \\
\hline  
\end{tabular}  
\caption{Average discrepancies between $c$ and $|\langle u, \hat u\rangle|$.  }
\label{table-outer}
\end{table}

\begin{table} 
\centering 
\begin{tabular}{|l| c  | c | c | c|}  
\hline  
 $n$&  Gaussian & Rademacher & t, df=10 & t, df=3\\
\hline  
 1000 & 3.627e-03 & 3.625e-03 & 3.650e-03 & 2.598e-01 \\
 2000 & 2.704e-03 & 2.707e-03 & 2.712e-03 & 3.708e-01 \\
 4000 & 1.951e-03 & 1.939e-03 & 1.952e-03 & 4.895e-01 \\
 8000 & 1.409e-03 & 1.388e-03 & 1.410e-03 & 6.112e-01 \\
\hline  
\end{tabular}  
\caption{Average discrepancies between $\tilde c$ and $|\langle v, \hat v \rangle|$.   }
\label{table-inner}
\end{table}

\section{Conclusions and future work}
\label{sec-conclusion}

We have derived the optimal spectral shrinkers method for signal prediction and covariance estimation in the spiked model with heteroscedastic noise, where the data is whitened before shrinkage and unwhitened after shrinkage. We also showed the in that $\gamma\to0$ regime, optimal singular value shrinkage with whitening converges to the best linear predictor, whereas optimal shrinkage without whitening converges to a suboptimal linear filter. We showed that under certain additional modeling assumptions, whitening improves the estimation of the signal's principal components, and achieves the optimal rate for subspace estimation when $r=1$. We showed that the operator norm SNR of the observations increases after whitening. We also extended the analysis on out-of-sample prediction found in \cite{dobriban2017optimal} to the whitening procedure.

There are a number of interesting directions for future research. First, we plan to revisit previous works that have employed similar shrinkage-plus-whitening procedures, but with the optimal shrinkers we have derived.  It is of interest to determine how much of an improvement is achieved with the more principled choice we have presented.

As our current analysis is restricted to the setting of Gaussian noise, in future work we will try to extend the analysis to more general noise matrices. This likely requires a deeper understanding of the distribution of the projection of the empirical singular vectors onto the orthogonal complement of the population signal vectors in the setting of non-Gaussian noise.

While we have shown that whitening can improve subspace estimation generically, and matches the error rate (up to a constant) of \cite{zhang2018heteroskedastic}, it is not clear if whitening is the optimal transformation for subspace estimation. In a different but closely related model to the one we have studied, where the noise variances differ across observations rather than across coordinates, it was found that certain weighting schemes can outperform whitening \cite{hong2018optimally}. We note too that if the matrix $\Sigma_\ep$ is ill-conditioned, numerical instabilities may result from the whitening and unwhitening operations.

Finally, it is also of interest to better understand the procedure when the noise covariance $\Sigma_\ep$ is not known exactly, but must be estimated. This is a subject currently under investigation.



\section*{Acknowledgements}
The authors would like to thank Edgar Dobriban, Matan Gavish, and Amit Singer for stimulating discussions related to this work. William Leeb acknowledges support from the Simons Foundation Collaboration on Algorithms and Geometry, the NSF BIGDATA program IIS 1837992, and BSF award 2018230. Elad Romanov acknowledges support from Israeli Science Foundation grant number 1523/16.

\bibliographystyle{plain}
\bibliography{refs_homoge}

\appendix

\section{Proof from Section \ref{sec-asymptotics}}
\label{proof-asymptotics}

\subsection{Proof of Theorem \ref{thm-main1}}
We begin by recalling the result that describes the asymptotics of the spiked model with white noise. This result can be found in \cite{paul2007asymptotics, benaych2012singular}. We immediately obtain parts \ref{main1-spike} and \ref{main1-product2} of Theorem \ref{thm-main1}.

\begin{thm}
\label{thm-classical}
If $p/n \to \gamma > 0$ as $n\to\infty$, the $k^{th}$ largest singular value of $Y^{\wh}$ converges almost surely to
\begin{align}
\sigma_k^{\wh} = 
\begin{cases}
\sqrt{(\ell_k^{\wh} + 1)\left( 1 + \frac{\gamma}{\ell_k^{\wh}}\right) }
    &\text{ if } \ell_k^{\wh} > \sqrt{\gamma} \\
1 + \sqrt{\gamma}   &\text{ otherwise}
\end{cases}.
\end{align}
Furthermore, for $1 \le j, k \le r$:
\begin{align}
\langle u_j^{\wh}, \hat u_k^{\wh} \rangle^2  \to  
\begin{cases}
(c_k^{\wh})^2, & \text{ if } j=k \text{ and }  \ell_k^{\wh} > \sqrt{\gamma}\\
0, & \text{ otherwise }
\end{cases}
\end{align}
and
\begin{align}
\langle v_j^{\wh}, \hat v_k^{\wh} \rangle^2  \to  
\begin{cases}
(\tilde c_k^{\wh})^2, & \text{ if } j=k \text{ and }  \ell_k^{\wh} > \sqrt{\gamma}\\
0, & \text{ otherwise }
\end{cases}
\end{align}
where the limits hold almost surely as $p,n \to \infty$ and $p/n \to \gamma$.
\end{thm}

We now turn to proving parts \ref{main1-norm} and \ref{main1-product}.
Let $\W = \sp\{u_1^{\wh},\dots,u_r^{\wh}\}$ be the $r$-dimensional subspace spanned by the whitened population PCs (the left singular vectors of $X^{\wh}$). For fixed $n$ and $p$, write
\begin{align}
\hat u_k^{\wh} = \overline c_k^{\wh} w_k^{\wh} + \overline s_k^{\wh} \tilde u_k^{\wh},
\end{align}
where $(\overline c_k^k)^2 + (\overline s_k^k)^2 = 1$, and $w_k^{\wh} \in \W$, and $\tilde u_k^{\wh} \perp \W$ are unit vectors. Because the whitened noise matrix is Gaussian, and hence orthogonally invariant, the vector $\tilde u_k^{\wh}$ is uniformly distributed over the unit sphere in $\W^\perp$. Since the dimension of $\W$ is fixed, it follows immediately from Proposition 6.2 in \cite{benaych2011fluctuations} that for any unit vector $x \in \R^p$ independent of $\tilde u_k^{\wh}$, the following limits hold almost surely:
\begin{align}
\label{eq-vanishing}
\lim_{p \to \infty} (\tilde u_k^{\wh})^\top x = 0,
\end{align}
and 
\begin{align}
\label{eq-trace}
\lim_{p \to \infty} 
    \left\{ (\tilde u_k^{\wh})^\top A \tilde u_k^{\wh} - \mu_a \right\}
= \lim_{p \to \infty} 
    \left\{ (\tilde u_k^{\wh})^\top A \tilde u_k^{\wh} - \frac{1}{p} \tr{A} \right\}
= 0.
\end{align}

From Theorem \ref{thm-classical}, we know $|(w_k^{\wh})^\top u_k^{\wh}| \to 1$ and $(w_k^{\wh})^\top u_j^{\wh} \to 0$ almost surely when $j \ne k$; and $\overline c_k^{\wh} \to c_k^{\wh}$ almost surely. Consequently, we can write
\begin{align}
\hat u_k^{\wh} = c_k^{\wh} u_k^{\wh} + s_k^{\wh} \tilde u_k^{\wh} + \psi
\end{align}
where $\|\psi\| \to 0$ almost surely as $p \to \infty$. The inner product of $\psi$ with any vectors of bounded norm will therefore also converge to $0$. As a short-hand, we will write:
\begin{align}
\label{eq-main}
\hat u_k^{\wh} \sim c_k^{\wh} u_k^{\wh} + s_k^{\wh} \tilde u_k^{\wh},
\end{align}
to indicate that the norm of the difference of the two sides converges to $0$ almost surely as $p \to \infty$.


From \eqref{eq-main} we have:
\begin{align}
\label{eq246}
A^{1/2} \hat u_k^{\wh} 
    \sim c_k^{\wh} A^{1/2} u_k^{\wh} + s_k^{\wh} A^{1/2} \tilde u_k^{\wh}.
\end{align}
Taking the squared norm of each side of \eqref{eq246} and using \eqref{eq-vanishing} and \eqref{eq-trace}, we obtain:
\begin{align}
\|A^{1/2} \hat u_k^{\wh}\|^2 
    \sim (c_k^{\wh})^2 \|A^{1/2} u_k^{\wh}\|^2 + (s_k^{\wh})^2 \|A^{1/2} \tilde u_k^{\wh}\|^2
    \sim \frac{(c_k^{\wh})^2 }{\tau_k^a} + (s_k^{\wh})^2 \mu_a,
\end{align}
This completes the proof of part \ref{main1-norm}.

Part \ref{main1-product} is proved in the same fashion. Taking inner products with each side of \eqref{eq246}, and using \eqref{eq-vanishing}, we get
\begin{align}
\langle A u_k^{\wh},  \hat u_k^{\wh} \rangle
= \langle A^{1/2} u_k^{\wh}, A^{1/2} \hat u_k^{\wh} \rangle
\sim \frac{c_k^{\wh}}{\tau_k^a} + s_k^{\wh} ((u_k^{\wh})^\top A \tilde u_k^{\wh})
\sim \frac{c_k^{\wh}}{\tau_k^a},
\end{align}
which is the desired result.

\subsection{Proof of Theorem \ref{thm-main2}}

We can decompose $X$ as:
\begin{align}
X = \sum_{k=1}^{r} \ell_k^{1/2} u_k z_k^\top / \sqrt{n}.
\end{align}
Since $z_{jk}$ and $z_{j k^\prime}$ are uncorrelated when $k \ne k^\prime$, and both have variance 1, the vectors $z_k / \sqrt{n}$ are almost surely asymptotically orthonormal, i.e., $\lim_{n \to \infty} |\langle z_k , z_{k^\prime} \rangle| / n = \delta_{k k^\prime}$. It follows that the $z_k / \sqrt{n}$ are asymptotically equivalent to the right singular vectors $v_k$ of $X$, that is,
\begin{align}
\lim_{n \to \infty} \langle v_k,z_k \rangle^2 / n = 1
\end{align}
almost surely; and the singular values of $X$ are asymptotically equal to the $\ell_k^{1/2}$. That is, we can write:
\begin{align}
X \sim \sum_{k=1}^{r} \ell_k^{1/2} u_k v_k^\top,
\end{align}
where $C\sim D$ indicates $\|C - D\|_{\op} \to 0$ as $p,n \to \infty$. Similarly, we can also write
\begin{align}
\label{eq480}
X^{\wh} \sim \sum_{k=1}^{r} (\ell_k^{\wh})^{1/2} u_k^{\wh} (v_k^{\wh})^\top.
\end{align}

We can also decompose $X^{\wh}$ by applying $W$ to $X$:
\begin{align}
\label{eq481}
X^{\wh} = WX 
\sim \sum_{k=1}^{r} \ell_k^{1/2} W u_k v_k^\top
= \sum_{k=1}^{r} (\ell_k \|W u_k\|^2)^{1/2} \overline u_k^{\wh} v_k^\top.
\end{align}
The condition \eqref{incoherence-condition} immediately implies that $\overline u_j^{\wh}$ and $\overline u_k^{\wh}$ are asymptotically orthogonal whenever $j \ne k$. Comparing \eqref{eq480} and \eqref{eq481} then shows that almost surely,
\begin{align}
\ell_k^{\wh} \sim \ell_k \|W u_k\|^2,
\end{align}
\begin{align}
\label{eq561}
\lim_{p \to \infty} \langle u_k^{\wh}, \overline u_k^{\wh} \rangle^2 = 1,
\end{align}
and
\begin{align}
\lim_{n \to \infty} \langle v_k, v_k^{\wh}\rangle^2 = 1.
\end{align}
From \eqref{eq561}, $\langle u_k , \overline u_k \rangle^2 \sim 1$ follows immediately.

To prove the asymptotically equivalent formula for $\tau_k$, we use \eqref{eq561}:
\begin{align}
\tau_k \sim \|W^{-1} u_k^{\wh} \|^{-2} 
\sim \|W^{-1} \overline u_k^{\wh} \|^{-2} 
\sim \|W^{-1} W u_k \|^{-2} \| W u_k\|^2
= \| W u_k\|^2.
\end{align}

To prove the formulas for the asymptotic cosine between $u_j$ and $\hat u_k$ we take $A = W^{-1}$ in Theorem \ref{thm-main1}. When $j \ne k$, we have the formula
\begin{align}
\label{eq1515}
\hat u_k^{\wh} \sim c_k^{\wh} u_k^{\wh} + s_k^{\wh} \tilde u_k^{\wh}
\sim c_k^{\wh} \frac{W u_k}{\sqrt{\tau_k}} + s_k^{\wh} \tilde u_k^{\wh}
\end{align}
and consequently
\begin{align}
W^{-1} \hat u_k^{\wh} 
    \sim c_k^{\wh} \frac{u_k}{\sqrt{\tau_k}} + s_k^{\wh} W^{-1} \tilde u_k^{\wh}.
\end{align}
We take inner products of each side with $u_j$. From the orthogonality of $u_k$ and $u_j$, and using \eqref{eq-vanishing}, we have:
\begin{align}
\langle u_j, W^{-1} \hat u_k^{\wh} \rangle \sim 0,
\end{align}
and consequently $\langle u_j, \hat u_k \rangle \sim 0$. When $j=k$, the formula for $\langle u_j, \hat u_k \rangle$ follows from Theorem \ref{thm-main1}.

Finally, we show that $\hat u_j$ and $\hat u_k$ are asymptotically orthogonal when $j \ne k$. We use the following lemma.

\begin{lem}
\label{lem-residuals}
Suppose $X = \sum_{k=1}^{r} \ell_k^{1/2} w_k v_k^\top$ is a $p$-by-$n$ rank $r$ matrix, and $G$ is a matrix with iid Gaussian entries $g_{ij} \sim N(0,1/n)$. Let $\hat w_1,\dots,\hat w_m$ be the left singular vectors of $Y = X+G$, where $m = \min(p,n)$, and write
\begin{align}
\hat w_k \sim c_k w_k + s_k \tilde w_k
\end{align}
where $\tilde w_k$ is orthogonal to $w_1,\dots,w_r$. Then for any sequence of matrices $A = A_p$ with bounded operator norms and any $1 \le j \ne k \le r$,
\begin{align}
\lim_{p \to \infty} \tilde w_j^\top A \tilde w_k = 0
\end{align}
almost surely.
\end{lem}

\begin{proof}
First, we prove the cases where $A=I_p$; that is, we show $\tilde w_j$ and $\tilde w_k$ are asymptotically orthogonal whenever $1 \le j \ne k \le r$. Indeed, we have
\begin{align}
s_j s_k \langle \tilde w_j, \tilde w_k \rangle 
&\sim \langle \hat w_j, \hat w_k \rangle + c_j c_k \langle w_j, w_k \rangle
    -  c_j \langle w_j, \hat w_k \rangle - c_k \langle w_k, \hat w_j \rangle
    \nonumber \\
&= -  c_j \langle w_j, \hat w_k \rangle - c_k \langle w_k, \hat w_j \rangle.
\end{align}
Since $\tilde w_j$ and $\tilde w_k$ are uniformly distributed on the subspace orthogonal to $w_1,\dots,w_r$, the inner products $\langle w_j, \hat w_k \rangle$ and $\langle w_k, \hat w_j \rangle$ both converge to $0$ almost surely as $p \to \infty$, proving the claim.

For general $A$, we note that the joint distribution of $\tilde w_j$ and $\tilde w_k$ is invariant to orthogonal transformations which leave fixed the $r$-dimensional subspace $\sp\{w_1,\dots,w_r\}$. The result then follows from Proposition 6.2 in \cite{benaych2011fluctuations}, which implies that 
\begin{align}
\tilde w_j^\top A \tilde w_k^\top \sim \frac{1}{p} \tr{A} \tilde w_j^\top \tilde w_k \sim 0,
\end{align}
where we have used the asymptotic orthogonality of $\tilde w_j$ and $\tilde w_k$.
\end{proof}

Since $u_k \sim \overline u_k$ and $u_j$ and $u_k$ are orthogonal, taking inner products of each side of \eqref{eq1515} with $W^{-1} \hat u_j^{\wh}$ we get:
\begin{align}
\langle W^{-1} \hat u_j^{\wh}, W^{-1} \hat u_k^{\wh} \rangle 
    \sim s_j^{\wh} s_k^{\wh} \langle W^{-1} \tilde u_j^{\wh}, W^{-1}\tilde u_k^{\wh} \rangle
    = s_j^{\wh} s_k^{\wh} (\tilde u_j^{\wh})^\top \Sigma_\ep \tilde u_k^{\wh}.
\end{align}
The result now follows from Lemma \ref{lem-residuals}.
%


\section{Proofs from Section \ref{sec-linpred}}

\label{proof-blp}

First, we establish the consistency of covariance estimation in the $\gamma = 0$ regime:
\begin{prop}
\label{prop-consistency}
If $p_n / n \to 0$ as $n \to \infty$, and the subgaussian norm of $QY_j$ can be bounded by $C$ independently of the dimension $p$, then the sample covariance matrix of $QY_1,\dots,QY_n$ converges to the population covariance $Q \Sigma_y Q$ in operator norm.
\end{prop}

\begin{proof}
We first quote the following result, stated as Corollary 5.50 in \cite{vershynin2010intro}:

\begin{lem}
\label{covest123}
Let $Y_1,\dots,Y_n$ be iid mean zero subgaussian random vectors in $\R^p$ with covariance matrix $\Sigma_y$, and let $\epsilon \in (0,1)$ and $t \ge 1$. Then with probability at least $1 - 2\exp(-t^2 p)$,
\begin{align}
\text{If } n \ge C (t/\epsilon)^2 p, \text{ then } 
    \|\hat \Sigma_y - \Sigma_y\| \le \epsilon,
\end{align}
where $\hat \Sigma_y = \sum_{j=1}^n Y_j Y_j^\top / n$ is the sample covariance, and $C$ is a constant.
\end{lem}

We also state the well-known consequence of the Borel-Cantelli Lemma:

\begin{lem}
Let $A_1,A_2,\dots$ be a sequence of random numbers, and let $\epsilon > 0$. Define:
\begin{align}
\A_n(\epsilon) = \{|A_n| > \epsilon\}.
\end{align}
If for every choice of $\epsilon > 0$ we have
\begin{align}
\sum_{n=1}^{\infty} \P(\A_n(\epsilon)) < \infty,
\end{align}
then $A_n \to 0$ almost surely.
\end{lem}

Now take $t = \epsilon\sqrt{n / Cp}$; then $n \ge C (t/\epsilon)^2 p$, and $t \ge 1$ for $n$ sufficiently large. Consequently,
\begin{align}
\P(\|\hat \Sigma_y - \Sigma_y\| > \epsilon) 
\le 2 \exp(-t^2 p) 
= 2 \exp(-n \epsilon^2 / C),
\end{align}
and so the series $\sum_{n\ge1} \P(\|\hat \Sigma_y - \Sigma_y\| > \epsilon)$ converges, meaning $\|\hat \Sigma_y - \Sigma_y\| \to 0$ almost surely as $n \to \infty$.

We now need to check that the subgaussian norm of $Y_j = X_j + \ep_j$ from the spiked model is bounded independently of the dimension $p$. But this is easy if the distribution of variances of $\ep_j$ is bounded, using, for example, Lemma 5.24 of \cite{vershynin2010intro}.
\end{proof}

An immediate corollary of Proposition \ref{prop-consistency} is that the sample eigenvectors  of $\hat \Sigma_y^q = Q \hat \Sigma_y Q $ are consistent estimators of the eigenvectors of $\Sigma_y^q = Q \Sigma_y Q$.

\begin{cor}
\label{cor-eigs}
Let $\Sigma_y^q = Q \Sigma_y Q$ be the population covariance of the random vector $Y_j^q = QY_j$, and let $\hat \Sigma_y^q = Q \hat \Sigma_y Q$ be the sample covariance of $Y_1^q,\dots,Y_n^q$. Let $u_1^q,\dots ,u_r^q$ denote the top $r$ eigenvectors of $\Sigma_y^q = Q \Sigma_y Q$, and $\hat u_1^1,\dots ,\hat u_r^q$ the top $r$ eigenvectors of $\hat \Sigma_y^q$.

Then for $1 \le k \le r$,
\begin{align}
\lim_{n \to \infty}|\langle \hat u_k^q, u_k^q \rangle| = 1,
\end{align}
where the limit holds almost surely as $n \to \infty$ and $p/n \to 0$.
\end{cor}

We now turn to the proof of Theorem \ref{thm-general}. First, we derive an expression for the BLP $\hat X_j^{\opt}$. We have:
\begin{align}
\label{eq:blp456}
\hat X_j^{\opt} &= \Sigma_x \left( \Sigma_x + \Sigma_\ep \right)^{-1} Y_j
    \nonumber \\
&= W^{-1} W \Sigma_x W \left( W \Sigma_x W + I \right)^{-1} W Y_j 
    \nonumber \\
&= W^{-1} \sum_{k=1}^{r} \frac{\ell_k^{\wh}}{\ell_k^{\wh} + 1} \langle WY_j, u_k^{\wh} \rangle u_k^{\wh}
    \nonumber \\
&= \sum_{k=1}^{r} \eta_k^{\opt} \langle WY_j, u_k^{\wh} \rangle W^{-1} u_k^{\wh},
\end{align}
where $W\Sigma_xW = \sum_{k=1}^{r} \ell_k^{\wh} u_k^{\wh} (u_k^{\wh})^\top$, and $\eta_k^{\opt} = \ell_k^{\wh} / (\ell_k^{\wh} + 1)$.

Now, for any $s_1,\dots,s_r$ satisfying 
\begin{align}
\label{s1234}
\lim_{\gamma \to 0} \frac{s_k}{\sigma_k^{\wh}} = \frac{\ell_k^{\wh}}{\ell_k^{\wh} + 1}.
\end{align}
we define the predictor $\hat X'$:
\begin{align}
\hat X' = \sum_{k=1}^{r} s_k W^{-1} \hat u_k^{\wh} (\hat v_k^{\wh})^\top.
\end{align}

Following the same reasoning as in the proof of Lemma \ref{lem-colwise}, we can write each column $\hat X_j'$ of $\sqrt{n} \hat X'$ as follows:
\begin{align}
\label{eq:xjprime}
\hat X_j'
= \sum_{k=1}^{r} (s_k / \sigma_k^{\wh}) \langle Y_j^{\wh}, \hat u_k^{\wh} \rangle W^{-1} \hat u_k^{\wh}.
\end{align}
Theorem \ref{thm-general} now follows from condition \eqref{s1234}, formula \eqref{eq:blp456}, and Corollary \ref{cor-eigs}. Theorem \ref{thm-blp} follows immediately, after observing that $\hat X$ has the same form as $\hat X'$ with $s_k = t_k$, and
\begin{align}
\lim_{\gamma \to 0} \frac{t_k }{\sigma_k^{\wh}}
= \lim_{\gamma \to 0}\frac{(\ell_k^{\wh})^{1/2} c_k^{\wh} \tilde c_k }
            {(c_k^{\wh})^2 + (s_k^{\wh})^2 \mu_\ep \tau_k }
        \frac{1}{\sqrt{\ell_k^{\wh} + 1}}
= \lim_{\gamma \to 0} \frac{(\ell_k^{\wh})^{1/2} \tilde c_k}{\sqrt{\ell_k^{\wh} + 1}}
= \frac{\ell_k^{\wh}}{\ell_k^{\wh} + 1}.
\end{align}

Finally, we prove Theorem \ref{prop-lin-pred}. By definition,
\begin{align}
\hat Y_{Q,j} 
= \sum_{k=1}^{r} (s_k^q / \sigma_k^q) \langle Y_j^q, \hat u_k^q \rangle Q^{-1} \hat u_k^q
\end{align}
and
\begin{align}
\hat Y_{Q,j}^{\lin} = \sum_{k=1}^{r} \eta_k^q \langle Y_j^q, u_k^q \rangle Q^{-1} u_k^q.
\end{align}
The values $s_k^q$ and $\eta_k^q$ are each assumed to minimize the mean-squared error for their respective expressions. Consequently, since Corollary \ref{cor-eigs} states that $\hat u_k^q \sim u_k^q$, we establish \eqref{linear1234};  \eqref{comparison641} follows immediately from \eqref{eq:blp456}.

\section{Proof of Theorem \ref{prop-oos-summary}}
\label{proofs-oos}

\subsection{The optimal coefficients for in-sample prediction}
\label{sec-incoefs}

Before deriving the optimal out-of-sample coefficients $\eta_k^{\out}$, we will first derive the optimal in-sample coefficients $\eta_k$. That is, we will rewrite the optimal shrinkage with noise whitening in the form \eqref{ins00}.

From Lemma \ref{lem-colwise}, the in-sample coefficients $\eta_k$ are the ratios of the optimal singular values $t_k$ derived in Section \ref{sec-shrinker} and the observed singular values of $Y^{\wh}$, denoted $\sigma_1^{\wh}, \dots, \sigma_r^{\wh}$. From Theorem \ref{thm-classical}, we know that
\begin{align}
\sigma_k^{\wh} = \sqrt{\left( \ell_k^{\wh} + 1 \right) \left(1 + \frac{\gamma}{\ell_k^{\wh}} \right)},
\end{align}
and from Section \ref{sec-shrinker} we know that
\begin{align}
t_k = \frac{(\ell_k^{\wh})^{1/2} c_k^{\wh} \tilde c_k }{(c_k^{\wh})^2 + (s_k^{\wh})^2 \mu_\ep \tau_k }
= \alpha_k (\ell_k^{\wh})^{1/2} c_k^{\wh} \tilde c_k,
\end{align}
where $\alpha_k = \left((c_k^{\wh})^2 + (s_k^{\wh})^2 \mu_\ep \tau_k \right)^{-1}$.
Taking the ratio, and using formulas \eqref{cos_out} and \eqref{cos_inn} for $c_k^{\wh}$ and $\tilde c_k$, we obtain:
\begin{align}
\eta_k = \frac{t_k}{\sigma_k^{\wh}} 
= \alpha_k \frac{(\ell_k^{\wh})^{1/2} c_k^{\wh} \tilde c_k}{\sqrt{\left( \ell_k^{\wh} + 1 \right) \left(1 + \frac{\gamma}{\ell_k^{\wh}} \right)}}
= \alpha_k \frac{\ell_k^{\wh} (c_k^{\wh})^2}{\sqrt{\left( \ell_k^{\wh} + 1 \right) \left(\ell_k^{\wh} +\gamma \right)}} 
    \sqrt{\frac{(\ell_k^{\wh})^2 + \gamma\ell_k^{\wh}}{(\ell_k^{\wh})^2 + \ell_k^{\wh}}}
= \alpha_k \frac{\ell_k^{\wh} (c_k^{\wh})^2}{\ell_k^{\wh}  + 1} .
\end{align}
That is, we have found the optimal in-sample coefficients to be:
\begin{align} \label{ins_eta}
\eta_k = \frac{1}{(c_k^{\wh})^2 + (s_k^{\wh})^2 \mu_\ep \tau_k } 
            \cdot \frac{\ell_k^{\wh} (c_k^{\wh})^2}{\ell_k^{\wh}  + 1} .
\end{align}
%

%

\subsection{The optimal coefficients for out-of-sample prediction}
\label{sec-outcoefs}

In this section, we will derive the optimal out-of-sample coefficients $\eta_k^{\out}$. We have a predictor of the form
\begin{align}
\hat X_0 = \sum_{k=1}^{r} \eta_k^{\out} \langle WY_0, \hat u_k^{\wh} \rangle W^{-1} \hat u_k^{\wh},
\end{align}
where $\hat u_k^{\wh}$ are the top left singular vectors of the in-sample observation matrix  $Y^{\wh} = W[Y_1,\dots,Y_n] / \sqrt{n}$. We wish to choose the coefficients $\eta_k^{\out}$ that minimize the asymptotic mean squared error $\E \|X_0 - \hat X_0\|^2$. First, we can expand the MSE across the different principal components as follows:
\begin{align}%
\|X_0 - \hat X_0\|^2 
=& \sum_{k=1}^{r} 
    \| \ell_k^{1/2} z_{0k} u_k 
        - \eta_k^{\out} \langle WY_0, \hat u_k^{\wh} \rangle W^{-1} \hat u_k^{\wh} \|^2
    \nonumber \\
&+ \sum_{k \ne l} \langle \ell_k^{1/2} z_{0k} u_k 
        - \eta_k^{\out} \langle WY_0, \hat u_k^{\wh} \rangle W^{-1} \hat u_k^{\wh},
        \ell_l^{1/2} z_{0l} u_l 
            - \eta_l^{\out} \langle WY_0, \hat u_l^{\wh} \rangle W^{-1} \hat u_l^{\wh} \rangle.
\end{align}
After taking expectations, the cross-terms vanish and we are left with:
\begin{align}%
\E \|X_0 - \hat X_0\|^2 
=& \sum_{k=1}^{r} 
    \E \| \ell_k^{1/2} z_{0k} u_k 
        - \eta_k^{\out} \langle WY_0, \hat u_k^{\wh} \rangle W^{-1} \hat u_k^{\wh} \|^2.
\end{align}
Since the sum separates across the $\eta_k^{\out}$, we can minimize each summand individually. We write:
\begin{align}
&\E \| \ell_k^{1/2} z_{0k} u_k 
        - \eta_k^{\out} \langle WY_0, \hat u_k^{\wh} \rangle W^{-1} \hat u_k^{\wh} \|^2
    \nonumber \\
=\,& \ell_k 
    + (\eta_k^{\out})^2 \E \left[\langle WY_0, \hat u_k^{\wh} \rangle^2 \|W^{-1} \hat u_k^{\wh}\|^2 \right]
    - 2 \ell_k^{1/2} \eta_k^{\out} \E \left[ z_{0k} 
        \langle WY_0, \hat u_k^{\wh} \rangle \langle u_k , W^{-1} \hat u_k^{\wh}\rangle \right].
\end{align}
We first deal with the quadratic coefficient in $\eta$:
\begin{align}
\langle WY_0, \hat u_k^{\wh} \rangle^2 \|W^{-1} \hat u_k^{\wh}\|^2
&= \langle WX_0 + W\ep_0, \hat u_k^{\wh} \rangle^2 \|W^{-1} \hat u_k^{\wh}\|^2
    \nonumber \\
&= \left(\langle WX_0,\hat u_k^{\wh}\rangle^2 
    + \langle W\ep_0, \hat u_k^{\wh} \rangle^2 
    +  \langle WX_0,\hat u_k^{\wh}\rangle\langle W\ep_0, \hat u_k^{\wh} \rangle\right) 
        \|W^{-1} \hat u_k^{\wh}\|^2,
\end{align}
and taking expectations, we get:
\begin{align}
\E \left[ \langle WY_0, \hat u_k^{\wh} \rangle^2 \|W^{-1} \hat u_k^{\wh}\|^2 \right]
\sim \left(\E \left[ \langle WX_0,\hat u_k^{\wh}\rangle^2\right]  
    +  1\right) \|W^{-1} \hat u_k^{\wh}\|^2
\sim \left( \ell_k^{\wh} (c_k^{\wh})^2 + 1 \right) 
    \left( \frac{(c_k^{\wh})^2}{\tau_k} + (s_k^{\wh})^2 \mu_\ep\right).
\end{align}

Now we turn to the linear coefficient in $\eta$:
\begin{align}
\ell_k^{1/2} \E \left[z_{0k} \langle WY_0, \hat u_k^{\wh} \rangle 
            \langle u_k , W^{-1} \hat u_k^{\wh}\rangle \right]
&= \ell_k^{1/2} \E \left[ z_{0k} \left( (\ell_k^{\wh})^{1/2}z_{0k}c_k^{\wh} 
            + \langle W \ep_0, \hat u_k^{\wh} \rangle \right ) 
    \langle u_k , W^{-1} \hat u_k^{\wh}\rangle \right]
    \nonumber \\
&= \frac{\ell_k^{\wh} c_k^{\wh}  \E \left[\langle u_k , W^{-1} \hat u_k^{\wh}\rangle\right]} {\|W u_k\|}
    \nonumber \\
& \sim \ell_k^{\wh} (c_k^{\wh})^2 \frac{1}{\tau_k}.
\end{align}

Minimizing the quadratic for $\eta_k^{\out}$, we get:
\begin{align}
\eta_k^{\out} &= \left(\ell_k^{\wh} (c_k^{\wh})^2 \frac{1}{\tau_k}\right) \bigg/
            \left( \left( \ell_k^{\wh} (c_k^{\wh})^2 + 1 \right) 
                \left( \frac{(c_k^{\wh})^2}{\tau_k} + (s_k^{\wh})^2 \mu_\ep\right)
            \right)
    \nonumber \\
&= \frac{1}{(c_k^{\wh})^2 + (s_k^{\wh})^2 \mu_\ep \tau_k } 
            \cdot \frac{\ell_k^{\wh} (c_k^{\wh})^2}{\ell_k^{\wh} (c_k^{\wh})^2 + 1} .
\end{align}

\subsection{Equality of the AMSEs}

Evaluating the out-of-sample error at the optimal out-of-sample coefficients $\eta_k^{\out}$, we find the optimal out-of-sample AMSE (where $\alpha_k = \left((c_k^{\wh})^2 + (s_k^{\wh})^2 \mu_\ep \tau_k \right)^{-1}$):
\begin{align}
\text{AMSE} &= \sum_{k=1}^{r} \left( \ell_k 
                - \frac{(\ell_k^{\wh})^2 (c_k^{\wh})^4}{\ell_k^{\wh} (c_k^{\wh})^2 + 1} 
                    \frac{1}{ \alpha_k \tau_k  }\right)
= \sum_{k=1}^{r} \left( \frac{\ell_k^{\wh} }{\tau_k}
                - \frac{(\ell_k^{\wh})^2 (c_k^{\wh})^4}{\ell_k^{\wh} (c_k^{\wh})^2 + 1} 
                    \frac{1}{ \alpha_k \tau_k  } \right).
\end{align}
The AMSE of the in-sample predictor is:
\begin{align}
\sum_{k=1}^{r} \ell_k (1 - (c_k \tilde c_k)^2) 
= \sum_{k=1}^{r} \frac{\ell_k^{\wh}}{\tau_k} 
    \left(1 
    - \frac{(c_k^{\wh} \tilde c_k^{\wh})^2}{\alpha_k} \right) 
= \sum_{k=1}^{r}\left( \frac{\ell_k^{\wh}}{\tau_k} 
    - \frac{\ell_k^{\wh} (c_k^{\wh} \tilde c_k^{\wh})^2}{\alpha_k \tau_k}
    \right)
\end{align}

To show equality, we therefore need to show:
\begin{align}
\ell_k^{\wh} (c_k^{\wh} \tilde c_k^{\wh})^2 = \frac{(\ell_k^{\wh})^2 (c_k^{\wh})^4}{\ell_k^{\wh} (c_k^{\wh})^2 + 1} .
\end{align}
But this follows from the equality of in-sample and out-of-sample AMSEs for the standard spiked model with isotropic noise, established in \cite{dobriban2017optimal}.

\section{Proofs from Section \ref{sec-pca}}
\label{proofs-pca}

\subsection{Proof of Proposition \ref{prop-pc0}}

From Corollary \ref{cor-eigs}, $\hat u_k^{\wh} \sim u_k^{\wh}$, $1 \le k \le r$, in the sense that the angle between the vectors converges to $0$. Consequently
\begin{align}
\lim_{n \to 0} \Theta(\U^{\wh},\hat \U^{\wh}) = 0,
\end{align}
where $\U^{\wh} = \sp\{u_1^{\wh},\dots,u_r^{\wh} \}$ and $\hat \U^{\wh} = \sp\{\hat u_1^{\wh},\dots,\hat u_r^{\wh}\}$. Since $W^{-1}$ has bounded operator norm and $\U = W^{-1} \U^{\wh}$ and $\hat \U = W^{-1} \hat \U^{\wh}$, the result follows immediately.

\subsection{Proof of Theorem \ref{thm-angles}}

Since the inner products between random unit vectors in $\R^p$ vanish as $p \to \infty$, we may assume that the $u_k$ are drawn randomly with iid entries of variance $1/p$; the result will then follow for the orthonormalized vectors from the generic model. If $\Sigma_\ep = \diag(\nu_1,\dots,\nu_p)$, then
\begin{align}
\tau_k = \|\Sigma_\ep^{-1/2} u_k\|^2 
\sim \frac{1}{p} \sum_{j=1}^{p} \nu_j^{-1} = \tau.
\end{align}

We now define the $n$-by-$p$ matrix $\tilde Y = Y^\top / \sqrt{\gamma}$, given by
\begin{align}
\tilde Y
&= \sum_{k=1}^{r} \tilde \ell_k^{1/2} z_k u_k^\top + G^\top \Sigma_\ep^{1/2} / \sqrt{p},
\end{align}
where $\tilde \ell_k = \ell_k / \gamma$. Note that the noise matrix $G^\top \Sigma_\ep^{1/2}$ has colored rows, not columns, and has been normalized by dividing by the square root of the number of its columns. Since the vectors $u_k$ spanning the right singular subspace of $\tilde Y$ are assumed to be drawn uniformly from the unit sphere in $\R^p$, we may apply Corollary 2 to Theorem 2 of \cite{hong2018asymptotic} to the matrix $\tilde Y$. Defining $\tilde \gamma = 1/\gamma$ as the aspect ratio of $\tilde Y$, we have:
\begin{align}
|\langle \hat u_k' , u_k \rangle|^2
\le \frac{1 - \tilde \gamma / (\tilde \ell_k / \mu_\ep)^2}{1 + 1 / (\tilde \ell_k / \mu_\ep)}
= \frac{1 - \gamma / (\ell_k^{\wh} / \varphi)^2}{1 + \gamma / (\ell_k^{\wh} / \varphi)}
\equiv g(\ell_k^{\wh} / \varphi),
\end{align}
where we have defined the function
\begin{align}
g(\ell) = \frac{1 - \gamma / \ell^2}{1 + \gamma / \ell}.
\end{align}

On the other hand, the squared cosine $c_k^2 = |\langle \hat u_k , u_k \rangle|^2$ is equal to
\begin{align}
c_k^2 = \frac{(c_k^{\wh})^2}{(c_k^{\wh})^2 + (s_k^{\wh})^2 \varphi}
= \frac{g(\ell_k^{\wh})}{g(\ell_k^{\wh}) + \varphi (1-g(\ell_k^{\wh}))}.
\end{align}
Our goal is to show that for all $\ell_k^{\wh} > \sqrt{\gamma}$, and all $\varphi \ge 1$, that
\begin{align}
g(\ell_k^{\wh} / \varphi) \le \frac{g(\ell_k^{\wh})}{g(\ell_k^{\wh}) + \varphi (1-g(\ell_k^{\wh}))};
\end{align}
equivalently, we want to show that for all $\xi > 0$ and $\varphi > 1$,
\begin{align}
g(\xi) \le \frac{g(\xi \varphi)}{g(\xi \varphi) + \varphi (1-g(\xi \varphi))};
\end{align}
setting
\begin{align}
G(\varphi) = \frac{g(\xi \varphi)}{g(\xi \varphi) + \varphi (1-g(\xi \varphi))},
\end{align}
this is equivalent to showing that $G(\varphi) \ge G(1)$ for all $\varphi \ge 1$. The derivative of $G$ is equal to
\begin{align}
\frac{d}{d\varphi} G(\varphi)
=\frac{\gamma\xi^2\varphi^2 + 2\gamma^2\xi\varphi + \gamma^2 }
    {(\xi^2\varphi^2 - \gamma + (\gamma\xi\varphi + \gamma)\varphi)^2} > 0,
\end{align}
which completes the first statement of the theorem.

The second statement concerning $\hat v_k$ is proved similarly. Again applying Corollary 2 to Theorem 2 of \cite{hong2018asymptotic} to $\tilde Y$, we know that
\begin{align}
|\langle \hat v_k' , z_k \rangle|^2
\le \frac{1 - \gamma / (\tilde \ell_k / \mu_\ep)^2}
    {1 + \tilde \gamma / (\tilde \ell_k / \mu_\ep)}
= \frac{1 - \gamma / (\ell_k^{\wh} / \varphi)^2}{1 + 1 / (\ell_k^{\wh} / \varphi)}
\equiv h(\ell_k^{\wh} / \varphi),
\end{align}
where we have defined the function
\begin{align}
h(\ell) = \frac{1 - \gamma / \ell^2}{1 +1 / \ell}.
\end{align}
Since $h$ is an increasing function of $\ell$, and $|\langle \hat v_k, z_k\rangle|^2 = \tilde c_k^2 = h(\ell_k^{\wh})$, the result follows. 

\subsection{Proof of Theorem \ref{thm-sin-theta}}

We begin the proof with some lemmas.
\begin{lem}
\label{lem-numbig}
Let $0 < B < 1$, and suppose $q$ is the number of entries of $u_{k}$ where $|u_{jk}| > B/\sqrt{p}$. Then 
\begin{align}
q \ge p \cdot \frac{1 - B^2}{C^2 - B^2},
\end{align}
where $C$ is the incoherence parameter from \eqref{eq-incoherence}.
\end{lem}
\begin{proof}
Let $S_1$ be the set of indices $j$ on which $|u_{jk}| > B/\sqrt{p}$, and let $S_2$ be the set of indices $j$ on which $|u_{jk}| \le B/\sqrt{p}$. Because $u_k$ is a unit vector, we then have
\begin{align}
1 = \|u_k\|^2 = \sum_{j=1}^{p} u_{jk}^2 
= \sum_{j \in S_1} u_{jk}^2 + \sum_{j \in S_2} u_{jk}^2
\le (q/p) C^2 + (1-q/p) B^2.
\end{align}
Rearranging, we find
\begin{align}
\frac{q}{p} \ge \frac{1-B^2}{C^2-B^2},
\end{align}
as claimed.
\end{proof}

\begin{lem}
\label{lem-tauk}
For each $1 \le k \le r$,
\begin{align}
\tau_k \ge \max\left\{ \frac{\tilde K}{\mu_\ep}, \frac{1}{\|\Sigma_\ep\|_{\op}} \right\},
\end{align}
where $\tilde K$ is a constant depending only on $C$ from \eqref{eq-incoherence}.
\end{lem}
\begin{proof}
We will let $\nu_1,\dots,\nu_p$ denote the diagonal elements of $\Sigma_\ep$. Take any number $0 < B < 1$, and let $q$ be the number of indices where $|u_{jk}| > B/\sqrt{p}$. From Lemma \ref{lem-numbig}, $q/p \ge K_1$, a constant. Using the Cauchy-Schwarz inequality, we have:
\begin{align}
\mu_\ep \cdot \tau_k
= \left(\sum_{j=1}^{p} \left( \frac{\sqrt{\nu_j}}{\sqrt{p}}\right)^2 \right) \cdot
    \left(\sum_{j=1}^{p} \left( \frac{u_{jk}}{\sqrt{\nu_j}}\right)^2 \right)
\ge \left( \frac{1}{\sqrt{p}}\sum_{j=1}^{p} |u_{jk}| \right)^2
\ge \left( \frac{1}{\sqrt{p}} (K_1p) \frac{B}{\sqrt{p}} \right)^2
= K_1^2 B^2.
\end{align}
This proves that $\tau_k \ge \tilde K / \mu_\ep$.

Next, we observe that because $\sum_{j=1}^{p}u_{jk}^2=1$, we have
\begin{align}
\tau_k = \sum_{j=1}^{p} \left( \frac{u_{jk}}{\sqrt{\nu_j}}\right)^2
\ge \min_{1 \le j \le p} \nu_j^{-1}
= \left(\max_{1 \le j \le p} \nu_j \right)^{-1}
=\frac{1}{\|\Sigma_{\ep}\|_{\op}},
\end{align}
completing the proof.
\end{proof}

We now turn to the proof of Theorem \ref{thm-sin-theta}. We have
\begin{align}
\|U_\perp^\top \hat U\|_{\op} = \|U_\perp U_\perp^\top \hat U\|_{\op} = \|\wtilde U\|_{\op}
\end{align}
where
\begin{align}
\wtilde U = [\tilde w_1,\dots , \tilde w_r]
\end{align}
is the matrix whose columns are the projections $\tilde w_k$ of $\hat u_k$ onto the orthogonal complement of $\sp\{u_1,\dots,u_r\}$. Then from Lemma \ref{lem-residuals}, we know that asymptotically $\tilde w_j \perp \tilde w_k$ if $j \ne k$; consequently,
\begin{align}
\| \sin \Theta(\hat U, U) \|_{\op}^2 = \max_{1 \le k \le r} \|\tilde w_k\|^2 
= \max_{1 \le k \le r} (1 - \langle \hat u_k, u_k \rangle^2)
= \max_{1 \le k \le r} (1 - c_k^2).
\end{align}
From Theorem \ref{thm-main2}, for each $1 \le k \le r$, the squared sine between $\hat u_k$ and $u_k$ is
\begin{align}
1 - c_k^2 
= 1 - \frac{(c_k^{\wh})^2}
            {(c_k^{\wh})^2 + (s_k^{\wh})^2 \cdot \mu_\ep \cdot \tau_k}
= \frac{(s_k^{\wh})^2 \cdot \mu_\ep \cdot \tau_k}
            {(c_k^{\wh})^2 + (s_k^{\wh})^2 \cdot \mu_\ep \cdot \tau_k}.
\end{align}
Since 
\begin{align}
(c_k^{\wh})^2 = \frac{1-\gamma / (\ell_k^{\wh})^2}{1 + \gamma / \ell_k^{\wh}}
\end{align}
and
\begin{align}
(s_k^{\wh})^2 = \frac{\gamma / \ell_k^{\wh} + \gamma / (\ell_k^{\wh})^2}{1 + \gamma / \ell_k^{\wh}},
\end{align}
we can simplify the expression by multiplying numerator and denominator by $(\ell_k^{\wh})^2(1 + \gamma / \ell_k^{\wh})$:
\begin{align}
\label{eq-135}
1 - c_k^2 
&= \frac{\gamma (\ell_k^{\wh} + 1) \mu_\ep \tau_k}
    {(\ell_k^{\wh})^2 - \gamma + \gamma (\ell_k^{\wh} + 1) \mu_\ep \tau_k}
    \nonumber \\
&= \frac{\gamma (\ell_k^{\wh} + 1) \mu_\ep \tau_k}{(\ell_k^{\wh})^2} \cdot
    \frac{(\ell_k^{\wh})^2}{(\ell_k^{\wh})^2 - \gamma + \gamma (\ell_k^{\wh} + 1) \mu_\ep \tau_k}.
\end{align}
Now, using Lemma \ref{lem-tauk}, there is a constant $0 < \tilde K < 1$ so that $\tau_k \mu_{\ep} \ge \tilde K$. Consequently, since $\gamma < (\ell_k^{\wh})^2$, we have:
\begin{align}
\label{eq-555}
\frac{(\ell_k^{\wh})^2}{(\ell_k^{\wh})^2 - \gamma + \gamma (\ell_k^{\wh} + 1) \mu_\ep \tau_k}
\le \frac{(\ell_k^{\wh})^2}{(\ell_k^{\wh})^2 - (1-\tilde K)\gamma }
\le \frac{(\ell_k^{\wh})^2}{(\ell_k^{\wh})^2 - (1-\tilde K) (\ell_k^{\wh})^2 }
= \frac{1}{\tilde K}.
\end{align}

Combining equation \eqref{eq-135} and inequality \eqref{eq-555}, the fact that $\ell_k^{\wh} = \ell_k \cdot \tau_k$, and Lemma \ref{lem-tauk}, we obtain the bound:
\begin{align}
1 - c_k^2 
&\le \frac{1}{\tilde K} \left(\frac{\gamma (\ell_k^{\wh} + 1) \mu_\ep \tau_k}{(\ell_k^{\wh})^2} 
    \right)
    \nonumber \\
&= \frac{1}{\tilde K} \left(
\frac{\gamma \ell_k^{\wh} \mu_\ep \tau_k}{(\ell_k^{\wh})^2}
    + \frac{\gamma \mu_\ep \tau_k}{(\ell_k^{\wh})^2} \right)
    \nonumber \\
&= \frac{1}{\tilde K} \left(
    \frac{\gamma  \mu_\ep }{\ell_k}
    + \frac{\gamma \mu_\ep }{\ell_k^2 \tau_k} \right)
    \nonumber \\
&\le \frac{1}{\tilde K} \left(
    \frac{\gamma  \mu_\ep }{\ell_k}
    + \frac{\gamma \mu_\ep \|\Sigma_\ep\|_{\op}}{\ell_k^2 } \right).
\end{align}
Taking the maximum over $1 \le k \le r$ proves the desired result.

\subsection{Proof of Proposition \ref{prop-snr}}

As in the proof of Theorem \ref{thm-angles}, since the inner products between random unit vectors in $\R^p$ vanish as $p \to \infty$, we may assume that the $u_k$ are drawn randomly with iid entries of variance $1/p$; the result will then follow for the orthonormalized vectors from the generic model. We will use the fact that $\|\hat \Sigma_x\|_{\op} = \|X\|_{\op}^2$ and $\|\hat \Sigma_\ep\|_{\op} = \|N\|_{\op}^2$. To show the increase in SNR after whitening, we will first derive a lower bound on the operator norm of the noise matrix $N$ alone. Recall that $N = \Sigma_{\ep}^{1/2} G $, where $g_{ij}$ are iid $N(0,1/n)$.

Take unit vectors $c$ and $d$ so that $G d = \|G \|_{\op} c$. Then we have 
\begin{align}
\|N\|_{\op}^2 \ge \|\Sigma_\ep^{1/2} G d\|^2 
    = \|G \|_{\op}^2 \|\Sigma_\ep^{1/2}c\|^2 
\end{align}
Since the distribution of $G$ is orthogonally invariant, the distribution of $c$ is uniform over the unit sphere in $\mathbb{R}^n$. Consequently, $\|\Sigma_\ep^{1/2} c\|^2 \sim \tr{\Sigma_\ep} / p \sim \mu_\ep$. Therefore,
\begin{align}
\|N\|_{\op}^2 \gtrsim \mu_\ep \|G \|_{\op}^2,
\end{align}
where ``$\gtrsim$'' indicates that the inequality holds almost surely in the large $p$, large $n$ limit.

Next, from the assumption that the $u_k$ are uniformly random, the parameters $\tau_k$ are all asymptotically given by:
\begin{align}
\tau_k \sim \|\Sigma_\ep^{-1/2} u_k\|^2 \sim \frac{\tr{\Sigma_\ep^{-1}}}{p} \sim \tau.
\end{align}
%

With this, we can show the improvement in SNR after whitening. We have:
\begin{align}
\snr \sim \frac{\ell_1}{\|N\|_{\op}^2} 
    \lesssim \frac{\ell_1}{ \mu_\ep \|G\|_{\op}^2 }
    \sim \frac{1}{\varphi}\frac{\ell_1 \tau}{  \|G\|_{\op}^2 }
    \sim \frac{1}{\varphi}\frac{\ell_1^{\wh}}{  \|G\|_{\op}^2 }
    \sim  \frac{\snr^{\wh}}{\varphi}.
\end{align}
This completes the proof.

\end{document}